\documentclass[12pt]{article}
\usepackage{amsmath}
\usepackage{graphicx,psfrag,epsf}
\usepackage{enumerate}
\usepackage{booktabs}
\usepackage{siunitx}
\usepackage{multirow}
\usepackage{makecell}
\usepackage{placeins}
\usepackage{url} % not crucial - just used below for the URL 
\usepackage{setspace, graphicx, caption, subcaption, float, relsize, rotating, color, verbatim, multirow}
\usepackage[round, comma]{natbib}
\usepackage[colorlinks=true, urlcolor=blue,linkcolor=blue, citecolor=magenta]{hyperref}
\newcommand{\blind}{0}

\usepackage{empheq}

\usepackage{enumerate}

\usepackage{dsfont}
\usepackage{mathtools}
\usepackage{bm}

\usepackage{stmaryrd}
\usepackage{amssymb} %to use \triangleq
\usepackage{amsmath}
\usepackage{amsfonts}
\usepackage{amsthm} %to use proof
\usepackage{mathrsfs}
\usepackage{bigints} %big integral
\usepackage{enumerate} %to use enumerate

\usepackage{color} %to use color

\usepackage{bbm}

\usepackage{extarrows}

\usepackage{amsfonts}
\usepackage{graphicx}
\usepackage{times}
\usepackage{bm}

\usepackage[plain,noend]{algorithm2e}

\newtheorem{thm}{Theorem}

\newtheorem{prop}[thm]{Proposition}

\newtheorem{rk}{Remark}
\newtheorem{theorem}[thm]{Theorem}
\newtheorem{assumption}[]{Assumption}
\newtheorem{lemma}{Lemma}

\counterwithin{table}{section}

\newcommand{\nc}{\newcommand}

\nc{\dps}{\displaystyle}
\nc{\tr}{\text{tr}}

\def\boxit#1{\vbox{\hrule\hbox{\vrule\kern6pt
			\vbox{\kern6pt#1\kern6pt}\kern6pt\vrule}\hrule}}
\def\zzcomment#1{\vskip 2mm\boxit{\vskip 2mm{\color{blue}\bf#1} {\color{blue}\bf -- ZZ\vskip 2mm}}\vskip 2mm}
\numberwithin{equation}{section}

\allowdisplaybreaks
\begin{document}
		\def\spacingset#1{\renewcommand{\baselinestretch}%
		{#1}\small\normalsize} \spacingset{1}

	%%%%%%%%%%%%%%%%%%%%%%%%%%%%%%%%%%%%%%%%%%%%%%%%%%%%%%%%%%%%%%%%%%%%%%%%%%%%%%
	
	\if0\blind
	{
		\title{\bf Testing Conditional Independence via Density Ratio Regression \thanks{
				The authors are alphabetically ordered. Zixuan Xu is a PhD student. The present work constitutes part of her research study leading to
			her doctoral thesis at Renmin University of China. Zheng Zhang is the corresponding author. E-mail: zhengzhang@ruc.edu.cn.}}
		\author{Chunrong Ai\\School of Management and Economics\\Chinese University of Hong Kong, Shenzhen\\and\\Zixuan Xu\\
			 Institute of Statistics \& Big Data, Renmin University of China\\and\\Zheng Zhang\\
	Center for Applied Statistics, Institute of Statistics \& Big Data\\Renmin University of China}
		\maketitle
	} \fi
	
	\if1\blind
	{
		\bigskip
		\bigskip
		\bigskip
		\begin{center}
			{\LARGE\bf Testing Conditional Independence via Density Ratio Regression}
		\end{center}
		\medskip
	} \fi  
	
	\bigskip

	\begin{abstract}
	 This paper develops a conditional independence (CI) test from a conditional density ratio (CDR) for weakly dependent data. The main contribution is presenting a closed-form expression for the estimated conditional density ratio function with good finite-sample performance. The key idea is exploiting the linear sieve combined with the quadratic norm. \cite{matsushita2022estimating} exploited the linear sieve to estimate the unconditional density ratio. We must exploit the linear sieve twice to estimate the conditional density ratio. First, we estimate an unconditional density ratio with an unweighted sieve least-squares regression, as done in \cite{matsushita2022estimating}, and then the conditional density ratio with a weighted sieve least-squares regression, where the weights are the estimated unconditional density ratio. 
     The proposed test has several advantages over existing alternatives. First, the test statistic is invariant to the monotone transformation of the data distribution and has a closed-form expression that enhances computational speed and efficiency. Second, the conditional density ratio satisfies the moment restrictions. The estimated ratio satisfies the empirical analog of those moment restrictions. As a result, the estimated density ratio is unlikely to have extreme values. Third, the proposed test can detect all deviations from conditional independence at rates arbitrarily close to $n^{-1/2}$ , and the local power loss is independent of the data dimension.  A small-scale simulation study indicates that the proposed test outperforms the alternatives in various dependence structures. 
	\end{abstract}
	
	\vfill
	
	\newpage
	\spacingset{1.8} % DON'T change the spacing!

	\section{Introduction}
	Conditional independence (hereafter CI) is vital in econometrics, statistics, and machine learning studies. For example, in time-series regressions, researchers usually assume that the far-lagged dependent variables and the error terms are independent conditional on the most recent information. In the literature of program evaluation, the independence between program participation and latent outcomes conditional on the confounders is a key condition for identifying treatment effects \citep{linton2014testing,huang2016flexible}. CI is also a key condition for exploring queries in statistical inference, such as sufficiency, parameter identification, adequacy, and ancillary \citep{dawid1979conditional,dawid1980conditional}. Furthermore, CI is a widely used feature selection and screening measure in high-dimensional data learning \citep{xing2001feature,tong2022model}.   
    
	Because of the critical role of CI in the studies above, testing its validity has received increasing attention from the econometric literature. Existing studies proposed several tests based on conditional densities, distributions, or characteristic functions, with distance measures like $L^2$ distance, Hellinger distance, or (generalized) mutual information (also termed Kullback–Leibler (hereafter KL) information or entropy). These tests are called density-based, characteristics-based, and distribution-based tests respectively. Most of these tests use a kernel-smoothed nonparametric estimation. For example, \cite{su2008nonparametric} proposed a density-based test with kernel smoothing; 
 \cite{su2007consistent} and \cite{wang2018characteristic} proposed characteristics-based tests with kernel smoothing; \cite{su2014testing} proposed a distribution-based test with a kernel-smoothed conditional empirical likelihood. All of these tests are asymptotically pivotal with a standard limiting distribution and are consistent in detecting any deviation from CI. However, a key limitation is their significant loss of local power, and the power loss increases with the dimension of the data. 
    \cite{linton2014testing}, on the other hand, proposed a moment-restrictions test without kernel smoothing. However, their test needs a complicated bias correction, is not asymptotically pivotal, and is inconsistent against all deviations from CI. \cite{huang2016flexible} proposed an integrated moment test without kernel smoothing. Though consistent against all deviations from CI and achieving  $n^{-1/2}$ local power, their test needs numerical integration. \cite{ai2021testing} proposed a mutual information test for CI. Though consistent against all departures from CI, their test suffers a slight loss of local power and needs to numerically solve an optimization. Moreover, \cite{linton2014testing}, \cite{huang2016flexible} , and \cite{ai2021testing} only considered independent and identically distributed (i.i.d.) data, so their tests do not apply to dependent data. For an example of a CI test for dependent data, see testing Granger non-causality for two stationary ergodic time series \citep{florens1982noncausality,florens1996noncausality}.  
    
	The literature lacks a CI test that is asymptotically pivotal, consistent against all departures from CI, suffers a slight loss of power, and is easy to compute. This paper fills this gap with a weighted $L^2$-test for CI that applies to both independent and weakly dependent data. The main idea is to estimate a conditional density ratio via a weighted sieve least-squares regression, with an unconditional density ratio as weight. We estimate the unconditional density ratio by an unweighted sieve least-squares regression. \cite{kanamori2009least} used the unweighted sieve least-squares regression to calculate the ratio of two marginal probability densities, and \cite{matsushita2022estimating} applied the same idea to estimate the unconditional density ratio (i.e. the product of two marginal densities over the joint density). We extend their approach to estimate the conditional density ratio. The estimated conditional density ratio has a closed-form expression, satisfies the empirical moment restrictions, and has good finite-sample performance. Under sufficient conditions, we show that the proposed test is asymptotically pivotal with a standard limiting distribution and is consistent against all departures from CI. The test suffers only a slight loss of local power and the loss is independent of the data dimension. 
    
    The rest of the paper is organized as follows. Section \ref{sec:framework} provides a basic framework and describes the least-squares approach for estimating density ratios. Section \ref{sec:asymptotics} studies the convergence rates of the estimated density ratio and establishes the limiting distribution and the power against local and fixed alternatives, respectively, for the test statistic. We report a simulation analysis of the size and power of the proposed test in Section \ref{sec:simulation} and an application in Section \ref{sec:application}. The final section summarizes the main results. All technical proofs are contained in the appendix.
		
	\section{Basic Framework}\label{sec:framework}
Let $X\in  \mathbb{R}^{d_{X}}$, $Y\in 
	\mathbb{R}^{d_{Y}}$, and  $Z\in \mathbb{R}^{d_{Z}}$ denote three random vectors whose supports are denoted by $\mathcal{X}$, $\mathcal{Y}$ and $\mathcal{Z}$, where $%
	d_{X}$, $d_{Y}$  and $d_{Z}$ are integers signifying the dimensions of $X$, $Y$, $Z$ respectively. The goal of this paper is to develop a consistent test for the null hypothesis that $X$ and $Y$ are conditionally independent given $Z$, i.e.
	\begin{align}
		H_0: X \perp Y|Z;\ \text{against} \  H_1: X \not \perp Y|Z, \label{H0:primal}
	\end{align}
Let $\{X_t,Y_t,Z_t\}_{t=1}^n$ denote a weakly dependent and identically distributed sample of $(X,Y,Z)$. Let $f_{X,Y|Z}(x,y|z)$, $f_{X|Z}(x|z)$, and $f_{Y|Z}(y|z)$ denote the conditional probability density functions. The primal null hypothesis \eqref{H0:primal} is equivalent to 
	\begin{align}
		H_0:\pi_0(x,y,z)=1\ \text{for almost all} \ (x,y,z) ;\ \text{against} \  H_1: H_0\ \text{fails}. \label{H0:primal2}
	\end{align}
where 
\begin{equation*}
	\pi_0(x,y,z):= \frac{f_{X|Z}(x|z)f_{Y|Z}(y|z)}{f_{X,Y|Z}(x,y|z)}
\end{equation*}
is a \emph{conditional density ratio function} (CDR). Let $Q(\pi)$ be the weighted $L_2(X,Y,Z)$ distance between an arbitrary function  $\pi:\mathcal{X}\times\mathcal{Y}\times\mathcal{Z}\to \mathbb{R}$ and  $\pi_0(x,y,z)$:  
\begin{align}\label{def:Q}
&Q(\pi):=	\frac{1}{2}\int_{\mathcal{X}\times\mathcal{Y}\times\mathcal{Z}}\left\{\pi(x,y,z)-\pi_0(x,y,z)\right\}^2r_0(y,z)f_{X,Y,Z}(x,y,z)dxdydz,
\end{align}
where $$r_0(y,z):=\frac{f_Y(y)f_Z(z)}{f_{Y,Z}(y,z)}$$ is an unconditional density ratio function. Obviously, $Q(\pi)\geq 0$ for any function $\pi(\cdot)$, and $Q(\pi)=0$ if and only if $\pi=\pi_0$ a.s.. 
    Define the pseudo-likelihood ratio test statistic by
	\begin{align}
	I_0:=&\min_{\{\pi:\pi(\cdot)\equiv 1\}}Q(\pi)-\min_{\{\pi:\ \text{no restriction} \}}Q(\pi)=Q(1) \notag\\
	=&\frac{1}{2}\int_{\mathcal{X}\times\mathcal{Y}\times\mathcal{Z}}\left\{1-\pi_0(x,y,z)\right\}^2r_0(y,z)f_{X,Y,Z}(x,y,z)dxdydz.	\label{def:I0}
\end{align}
Then, $X$ and $Y$ are conditionally independent given $Z$ if and only if $I_0=0$. Thus, $I_0$ is an indicator of the deviation from the null.   

\begin{rk}\label{rk:invariance}
It is easy to show that $I_0$ is invariant to continuous monotonic transformations of $(X,Y,Z)\to (\widetilde{X},\widetilde{Y},\widetilde{Z})=(\phi_X(X),\phi_Y(Y),\phi_Z(Z))$: 
\begin{align}\label{eq:invariance}
	&\int\left\{1-\frac{f_{X|Z}(x|z)f_{Y|Z}(y|z)}{f_{X,Y|Z}(x,y|z)}\right\}^2\frac{f_Y(y)f_Z(z)}{f_{Y,Z}(y,z)}f_{X,Y,Z}(x,y,z)dxdydz\\
	=&\int\left\{1-\frac{f_{\widetilde{X}|\widetilde{Z}}(x|z)f_{\widetilde{Y}|\widetilde{Z}}(y|z)}{f_{\widetilde{X},\widetilde{Y}|\widetilde{Z}}(x,y|z)}\right\}^2\frac{f_{\widetilde{Y}}(y)f_{\widetilde{Z}}(z)}{f_{\widetilde{Y},\widetilde{Z}}(y,z)}f_{\widetilde{X},\widetilde{Y},\widetilde{Z}}(x,y,z)dxdydz. \notag
\end{align}
The proof of \eqref{eq:invariance} is given in Appendix \ref{app:invariance}. This invariance property allows for $(X,Y,Z)$ unbounded or having no finite moments since we can equivalently test the hypothesis $\widetilde{H}_0:F_X(X) \perp F_Y(Y)|F_Z(Z)$, where $F_X,F_Y,F_Z$ are cumulative distribution functions (CDFs) of $X,Y,Z$ respectively. Hence, we can assume that supports, $\mathcal{X},\mathcal{Y}$, and $\mathcal{Z}$, are compact without loss of generality. 
\end{rk}  
    To construct a test statistic for \eqref{def:I0}, we must estimate the conditional density ratio $\pi_{0}(x,y,z)$. There are two ways of calculating it in the literature. The first one estimates the conditional density functions $f_{X|Z}(x|z)$, $f_{Y|Z}(y|z)$, $f_{X,Y|Z}(x,y|z)$ and then forms the ratio estimator. However, this approach is sensitive to small denominator values. The second one estimates the conditional density ratio function $\pi_0(x,y,z)$ by maximizing entropy \citep{ai2021testing}. However, this approach does not produce a closed-form solution, and the estimator requires solving a numerical optimization problem, which is computationally costly. We present a third, easy-to-compute estimator for the conditional density ratio with a good finite sample performance.
	Our approach approximates the conditional density ratio with a linear sieve \citep{chen2007large}:
	\begin{equation*}
		\pi_0(x,y,z)\approx\pi_K(x,y,z;\boldsymbol{\beta})=\boldsymbol{\beta}^\top v_{K}(x,y,z) = \sum_{k=1}^{K}\beta_kv_{K,k}(x,y,z),
	\end{equation*} 
where $v_{K}(x,y,z) = (v_{K,1}(x,y,z),...,v_{K,K}(x,y,z))^\top$ is a $K$-dimensional sieve basis and $\boldsymbol{\beta} = (\beta_1,...,\beta_{K})^\top\in \mathbb{R}^K$ is a vector of sieve coefficients.  The sieve coefficients minimize the following weighted $L^2(X,Y,Z)$-norm:
	\begin{align}\label{def:QK}
			Q_K(\boldsymbol{\beta}):=&Q\left(\pi_K(\boldsymbol{\beta})\right)=\frac{1}{2} \int\{	\pi_K(x,y,z;\boldsymbol{\beta})-\pi_0(x,y,z)\}^2r_0(y,z)f_{X,Y,Z}(x,y,z)dxdydz.	
	\end{align}
By manipulation, we obtain: 
	\begin{align*}
		Q_K(\boldsymbol{\beta})=&	\frac{1}{2}\int\pi_K(x,y,z;\boldsymbol{\beta})^2r_0(y,z)f_{X,Y,Z}(x,y,z)dxdydz\\
		&-\int\pi_K(x,y,z;\boldsymbol{\beta})\pi_0(x,y,z)r_0(y,z)f_{X,Y,Z}(x,y,z)dxdydz\notag\\
		&+\frac{1}{2}\int\pi_0^2(x,y,z)r_0(y,z)f_{X,Y,Z}(x,y,z)dxdydz\notag \\
		=&\frac{1}{2}\cdot \boldsymbol{\beta}^\top \mathbb{E}[r_0(Y,Z)\cdot v_K(X,Y,Z)v^\top_K(X,Y,Z)]\boldsymbol{\beta}-\boldsymbol{\beta}^\top\int \mathbb{E}[v_K(X,y,Z)]f_Y(y)dy\\
		&+\frac{1}{2}\mathbb{E}\left[\pi_0^2(X,Y,Z)r_0(Y,Z)\right]\notag\\
		=& \frac{1}{2}\cdot \boldsymbol{\beta}^\top H_{ K}\boldsymbol{\beta}-\boldsymbol{\beta}^\top h_K+\frac{1}{2}\mathbb{E}\left[\pi_0^2(X,Y,Z)r_0(Y,Z)\right],
	\end{align*}
	where 
	\begin{align*}
		H_{K}:=\mathbb{E}[r_0(Y,Z)\cdot v_K(X,Y,Z)v^\top_K(X,Y,Z)]\ \text{and} \ h_K:=\int \mathbb{E}[v_K(X,y,Z)]f_Y(y)dy.
	\end{align*}
	The minimizer of $Q_K(\boldsymbol{\beta})$ is  
	\begin{equation*}
		\boldsymbol{\beta}^*_K:=\arg\min_{\boldsymbol{\beta}\in\mathbb{R}^K}	Q_K(\boldsymbol{\beta})=H_{K}^{-1}h_K
	\end{equation*}
and the weighted $L^2(X,Y,Z)$ approximation of $\pi_0(x,y,z)$ is 
	$$\pi_K^*(x,y,z):=\{\boldsymbol{\beta}^*_K\}^\top v_K(x,y,z).$$
	
To estimate $\pi_K^*(x,y,z)$, we need an estimate of $r_0(y,z)$. We approximate $r_0(y,z)$ using the sieve basis $u_{K_0}(y,z)$:
	\begin{equation*}
		r_{K_0}(y,z;\boldsymbol{\gamma})=\boldsymbol{\gamma}^\top u_{K_0}(y,z) = \sum_{k=1}^{K_0}\gamma_ku_{K_0,k}(y,z),
	\end{equation*} 
where $u_{K_0}(y,z) = (u_{K_0,1}(y,z),...,u_{K_0,K_0}(y,z))^\top$ is a $K_0$-dimensional sieve basis and $\boldsymbol{\gamma} = (\gamma_1,...,\gamma_{K_0})^\top\in \mathbb{R}^{K_0}$ is a vector of sieve coefficients. By manipulation, we obtain: 
	\begin{align*}
		G_{K_0}(\boldsymbol{\gamma}):=&\frac{1}{2} \int\{	r_{K_0}(y,z;\boldsymbol{\gamma})-r_0(y,z)\}^2f_{Y,Z}(y,z)dydz\\
		=&\frac{1}{2}\int r_{K_0}(y,z;\boldsymbol{\gamma})^2f_{Y,Z}(y,z)dydz-\int r_{K_0}(y,z;\boldsymbol{\gamma})r_0(y,z)f_{Y,Z}(y,z)dydz\notag\\
		&+\frac{1}{2}\int r_0^2(y,z)f_{Y,Z}(y,z)dydz\\
	=&\frac{1}{2}\cdot \boldsymbol{\gamma}^\top \mathbb{E}[ u_{K_0}(Y,Z)u^\top_{K_0}(Y,Z)]\boldsymbol{\gamma}-\boldsymbol{\gamma}^\top\int \mathbb{E}[u_{K_0}(y,Z)]f_Y(y)dy\\&+\frac{1}{2}\int r_0^2(y,z)f_{Y,Z}(y,z)dydz\notag\\
		=& \frac{1}{2}\cdot \boldsymbol{\gamma}^\top \Sigma_{K_0}\boldsymbol{\gamma}-\boldsymbol{\gamma}^\top b_{K_0} +\frac{1}{2}\cdot\mathbb{E}[r_0^2(Y,Z)],
	\end{align*}
where 
	\begin{align*}
		\Sigma_{K_0}:=\mathbb{E}[ u_{K_0}(Y,Z)u^\top_{K_0}(Y,Z)]\ \text{and} \ b_{K_0}:=\int \mathbb{E}[u_{K_0}(y,Z)]f_Y(y)dy.
	\end{align*}
The minimizer of $G_{K_0}(\boldsymbol{\gamma})$ is 
	\begin{equation*}
		\boldsymbol{\gamma}^*_{K_0}:=\arg\min_{\boldsymbol{\gamma}\in\mathbb{R}^{K_0}}	G_{K_0}(\boldsymbol{\gamma})=	\Sigma_{K_0}^{-1}b_{K_0}
	\end{equation*}
and the  $L^2(Y,Z)$ approximation of $r_0(y,z)$ is 
	$$r_{K_0}^*(y,z):=\{\boldsymbol{\gamma}^*_{K_0}\}^\top u_{K_0}(y,z).$$
We estimate $r_0(y,z)$ by
	\begin{align}\label{def:rhat}
		\widehat{r}_{K_0}(y,z):=	\widehat{\boldsymbol{\gamma}}^{\top}_{K_0}u_{K_0}(y,z),
	\end{align}
where $\widehat{\boldsymbol{\gamma}}_{K_0}:=	\widehat{\Sigma}_{K_0}^{-1}\widehat{b}_{K_0}
	$, $\ \widehat{b}_{K_0}:=\left\{n(n-1)\right\}^{-1}\sum_{j=1,j\neq i}^n\sum_{i=1}^nu_{K_0}(Y_i,Z_j)$, and
	\begin{align*}
		\widehat{\Sigma}_{K_0}:= \frac{1}{n}\sum_{t=1}^nu_{K_0}(Y_t,Z_t)u^\top_{K_0}(Y_t,Z_t).
	\end{align*}
	
With the estimated $\widehat{r}_{K_0}(y,z)$, we estimate $\pi_0(x,y,z)$ by 
	\begin{align}\label{def:pihat}
		\widehat{\pi}_{K,K_0}(x,y,z):=	\widehat{\boldsymbol{\beta}}^{\top}_{K,K_0}v_{K}(x,y,z),
	\end{align}
where $\widehat{\boldsymbol{\beta}}_{K,K_0}:=	\widehat{H}_{K,K_0}^{-1}\widehat{h}_{K}
	$, $\ \widehat{h}_{K}:=\left\{n(n-1)\right\}^{-1}\sum_{j=1,j\neq i}^n\sum_{i=1}^nv_{K}(X_i,Y_j,Z_i)$, and 
	\begin{align*}
		\widehat{H}_{K,K_0}:= \frac{1}{n}\sum_{t=1}^n\widehat{r}_{K_0}(Y_t,Z_t)v_{K}(X_t,Y_t,Z_t)v^\top_{K}(X_t,Y_t,Z_t).
	\end{align*}
The \emph{conditional density ratio test statistic} is now given by 
	\begin{align}\label{def:Ihat}
		\widehat{I}_{K,K_0}:=\frac{1}{2n}\sum_{t=1}^n\left\{	\widehat{\pi}_{K,K_0}(X_t,Y_t,Z_t)-1\right\}^2\widehat{r}_{K_0}(Y_t,Z_t).
	\end{align}

\begin{rk}
	By definition, the estimators $\widehat{r}_{K_0}(Y,Z)$ and  $\widehat{\pi}_{K,K_0}(X,Y,Z)$ satisfy the following  constraints:
\begin{align}\label{eq:balance_r}
	&\frac{1}{n}\sum_{t=1}^{n}\widehat{r}_{K_0}(Y_t,Z_t)u_{K_0}(Y_t,Z_t)=\frac{1}{n(n-1)}\sum_{j=1,j\neq i}^n\sum_{i=1}^nu_{K_0}(Y_i,Z_j);
\end{align} 
\begin{align}\label{eq:balance_pi}
	\frac{1}{n}\sum_{t=1}^{n}\widehat\pi_{K,K_0}(X_t,Y_t,Z_t)\widehat r_{K_0}(Y_t,Z_t)v_{K}(X_t,Y_t,Z_t)=\frac{1}{n(n-1)}\sum_{j=1,j\neq i}^n\sum_{i=1}^nv_{K}(X_i,Y_j,Z_i).
\end{align}
Since the linear space spanned by $v_K(x,y,z)$ contains the linear space spanned by $u_{K_0}(y,z)$, i.e. $span\{u_{K_0}(y,z)\}\subset span\{v_K(x,y,z)\}$, the empirical moment restrictions \eqref{eq:balance_r}  and \eqref{eq:balance_pi}  imply that $\widehat{\pi}_{K,K_0}(x,y,z)$ and $\widehat{r}_{K_0}(y,z)$ satisfy the following three-way balancing empirical moments:
\begin{align*}
	&\frac{1}{n}\sum_{t=1}^{n}\widehat\pi_{K,K_0}(X_t,Y_t,Z_t)\widehat r_{K_0}(Y_t,Z_t)u_{K_0}(Y_t,Z_t)\\ &=\frac{1}{n(n-1)}\sum_{j=1,j\neq i}^n\sum_{i=1}^nu_{K_0}(Y_i,Z_j)
	=\frac{1}{n}\sum_{t=1}^{n}\widehat{r}_{K_0}(Y_t,Z_t)u_{K_0}(Y_t,Z_t).
\end{align*} 
These empirical moments prevent the density ratio estimators from having extreme values. With \eqref{eq:balance_r} and \eqref{eq:balance_pi},  $\widehat{I}_{K,K_0}$ can be written as
\begin{align}\label{def:Ihat_equiv}
	\widehat{I}_{K,K_0}=\frac{1}{2n}\sum_{t=1}^n	\widehat{\pi}_{K,K_0}(X_t,Y_t,Z_t)^2\widehat{r}_{K_0}(Y_t,Z_t)-\frac{1}{2}.
\end{align}
\end{rk}	
\begin{rk} For independent and identically distributed (i.i.d.) data, \cite{kanamori2009least}  used the sieve least-squares regression to estimate  the ratio of two probability densities, and \cite{matsushita2022estimating} used the sieve least-squares regression to estimate the unconditional density ratio.  Our approach is an extension of theirs. 
	\end{rk}	\begin{rk}
	\cite{ai2021testing} estimated the density ratios $r_0(y,z)$ and $\pi_0(x,y,z)$ maximizing the entropy subject to moment restrictions. Their approach does not produce a closed-form solution and is computationally expensive.  Moreover, the lack of closed-form expressions complicates extending their results to dependent data.  
	\end{rk}
	\begin{rk}
		The inverse probability-weighted estimator is known to be sensitive to the small value of the estimated denominator. The proposed density ratio estimators do not have such problems. Moreover, the empirical moment balancing property of the calculated ratios is vital for efficiently estimating semiparametric models, see \cite{ai2021unified} and \cite{matsushita2022estimating}.
	\end{rk}

\section{Asymptotic Distribution of The Test Statistic}\label{sec:asymptotics} 	
\subsection{Useful Convergence Rates}  
  This section  establishes the convergence rates for $\widehat{r}_{K_0}$ and $\widehat{\pi}
	_{K,K_0}$, which will help establish the large sample properties of $
	\widehat{I}_{K,K_0}$. 	Let $\Vert\cdot\Vert$ be the Euclidean norm. The following conditions are maintained throughout the remainder of the paper:
	\begin{assumption}\label{ass:ciid} (i)
	$\{X_t, Y_t, Z_t\}_{t\in\mathbb{N}}$	is a strictly stationary $\beta$-mixing process with coefficients $\beta_m = O(\rho^m)$ for some $0<\rho<1$; (ii) the support $\mathcal{X}\times \mathcal{Y}\times\mathcal{Z}$ of $(X,Y,Z)$ is compact.
	\end{assumption}
	
	\begin{assumption}\label{ass:cbounded}
		%There are two positive constants $\eta_1$ and $\eta_2$ such that
		Suppose that the density ratios are bounded away from zero and from above: (i)
		$
		0<\inf_{(y,z)\in\mathcal{Y}\times\mathcal{Z}} r_0(y,z)\leq\sup_{(y,z)\in\mathcal{Y}\times\mathcal{Z}} r_0(y,z)<\infty$; (ii) $
		0<\inf_{(x,y,z)\in\mathcal{X}\times\mathcal{Y}\times \mathcal{Z}}\pi_0(x,y,z)\leq\sup_{(x,y,z)\in\mathcal{X}\times\mathcal{Y}\times \mathcal{Z}}\pi_0(x,y,z)<\infty
		$.
	\end{assumption}
	
	\begin{assumption}\label{ass:cpi0-pi*} Suppose that 
		(i)	there exists  $\boldsymbol{\gamma}_{K_0} \in \mathbb{R}^{K_0}$ and  a positive constant $\omega_r>0$ such that $\sup_{(y,z)\in\mathcal{Y}\times\mathcal{Z}}| r_0(y,z)-\boldsymbol{\gamma}_{K_0}^\top u_{K_0}(y,z) |=O(K_0^{-\omega_{r}})$; (ii) there exists   $\boldsymbol{\beta}_K \in \mathbb{R}^{K} $ and a positive constant $\omega_{\pi}>0$ such that $\sup_{(x,y,z)\in\mathcal{X}\times\mathcal{Y}\times \mathcal{Z}}|\pi_0(x,y,z)-\boldsymbol{\beta}^\top_Kv_{K}(x,y,z)|=O(K^{-\omega_{\pi}})$.
	\end{assumption} 
	\begin{assumption}\label{ass:eigen}
		(i) $K_0$ is a sequence depending on the sample size $n$ such that $K_0\to \infty$ as $n\to\infty$ and the eigenvalues of  $\Sigma_{K_0}$ are  bounded from above and away from zero uniformly in $K_0$;
		(ii)  $K\in\mathbb{N}$ is a sequence depending on the sample size $n$ such that $K\rightarrow\infty$ as $n\to \infty$ and the eigenvalues of  $H_K$ are  bounded from above and away from zero uniformly in $K$; (iii) There exists $\eta > 0$ and a constant $C$ such that  $\sup_{1\leq k\leq K_0}\mathbb{E}\left[|u_{K_0,k}(Y,Z)|^{4+\eta}\right]<C$ uniformly in $K_0$ and  $\sup_{1\leq k\leq K}\mathbb{E}\left[|v_{K,k}(X,Y,Z)|^{4+\eta}\right]<C$ uniformly in $K$. 
	\end{assumption}
		\begin{assumption}\label{ass:cK} Denote $\zeta_{K_0}:=\sup_{(y,z)\in\mathcal{Y}\times\mathcal{Z}}\Vert u_{K_0}(y,z)\Vert$ and $\xi_{K}:=\sup_{(x,y,z)\in\mathcal{X}\times\mathcal{Y}\times \mathcal{Z}}\Vert v_{K}(x,y,z)\Vert$. Suppose that the following conditions hold: (i)  $\zeta_{K_0}K_0^{-w_{r}}\rightarrow 0$ and $\zeta_{K_0}\sqrt{K_0/n}\rightarrow 0$;  (ii)  $\xi_{K}\{K^{-\omega_{\pi}}+K_0^{-\omega_{r}}\}\rightarrow 0$ and $\xi_{K}\sqrt{K/n}\rightarrow 0$.
	\end{assumption}
	Assumption \ref{ass:ciid} allows for weakly dependent data. The compact support assumption is not as restrictive as it appears because the proposed test statistic is invariant to monotonic and bounded data transformations; see Remark \ref{rk:invariance}.  
	Assumption \ref{ass:cbounded} requires the density ratio functions bounded and bounded away from zero. Again, this condition is not restrictive because of the invariant property. Assumption	\ref{ass:cpi0-pi*} requires the sieve approximation errors to shrink at a polynomial rate, which is positively related to the smoothness of $r_0(y,z)$ and $\pi_0(x,y,z)$ and negatively related to the number of the continuous cvariates. For power series and B-splines, $\omega_r=s_{r}/(d_Y+d_Z)$ and $\omega_{\pi}=s_{\pi}/(d_X+d_Y+d_Z)$, where $s_r$ and $s_{\pi}$ are the smoothness of $r_0(y,z)$ and $\pi_0(x,y,z)$, respectively. Assumption \ref{ass:eigen} ensures  the sieve estimator to be non-degenerate. Assumption \ref{ass:cK}  restricts the growth rate of the smoothing parameters to achieve consistent estimation. \cite{newey1994asymptotic,Newey97} showed that $\zeta_{K_0}=O(K_0)$ and $\xi_K=O(K)$  for power series, $\zeta_{K_0}=O(\sqrt{K_0})$ and $\xi_K=O(\sqrt{K})$ for  B-splines.
    The following theorem establishes the convergence rates for $\widehat{r}_{K_0}$ and $\widehat{\pi}_{K,K_0}$whose proof is relegated to Appendix  \ref{sec:th_crate}.
		\begin{theorem}\label{th_crate}
		(i)	Suppose that Assumptions  \ref{ass:ciid}, \ref{ass:cbounded}(i), \ref{ass:cpi0-pi*}(i), \ref{ass:eigen}(i) and \ref{ass:cK}(i) hold. Then, as $n\rightarrow \infty$, we obtain 
        
	\begin{equation}
	\Vert\widehat{\boldsymbol{\gamma}}_{K_0}-\boldsymbol{\gamma}^*_{K_0}\Vert  =O_P\left(\sqrt{\frac{K_0}{n}}\right);	\label{prop_gamma}
	\end{equation}	
		\begin{equation*}
			\label{pi_sup}
			\sup_{(y,z)\in\mathcal{Y}\times\mathcal{Z}}|\widehat{r}_{K_0}(y,z)-r_0(y,z)|=O\left(\zeta_{K_0}\left\{K_0^{-\omega_{r}}+\sqrt{\frac{K_0}{n}}\right\} \right);
		\end{equation*}
\begin{equation*}
			\label{pi_norm}
			\int|\widehat{r}_{K_0}(y,z)-r_0(y,z)|^2dF_{Y,Z}(y,z)=O\left(K_0^{-2\omega_r}+\frac{K_0}{n}\right);
		\end{equation*}
\begin{equation*}
			\frac{1}{n}\sum_{t=1}^{n}|\widehat{r}_{K_0}(Y_t,Z_t)-r_0(Y_t,Z_t)|^2=O_P\left(K_0^{-2\omega_r}+\frac{K_0}{n}\right).
		\end{equation*}
(ii) Suppose that Assumptions \ref{ass:ciid}-\ref{ass:cK} hold. Then, as $n\rightarrow \infty$, we obtain 
\begin{equation}
				\Vert\widehat{\boldsymbol{\beta}}_{K,K_0}-\boldsymbol{\beta}^*_K\Vert=O_P\left(\left\{\sqrt{\frac{K_0}{n}}+K_0^{-\omega_{r}}\right\}+\sqrt{\frac{K}{n}}\right);\label{prop_beta}
			\end{equation}
\begin{equation*}
				\label{pi_csup}
				\sup_{(x,y,z)\in\mathcal{X}\times\mathcal{Y}\times\mathcal{Z}}|\widehat{\pi}_{K,K_0}(x,y,z)-\pi_0(x,y,z)|=O_P\left(\xi_{K}\left\{\sqrt{\frac{K_0}{n}}+K_0^{- \omega_r} \right\}+\xi_K\left\{\sqrt{\frac{K}{n}}+K^{-\omega_\pi}\right\} \right);
			\end{equation*}
			\begin{equation*}
				\label{pi_cnorm}
				\int|\widehat{\pi}_{K,K_0}(x,y,z)-\pi_0(x,y,z)|^2dF_{X,Y,Z}(x,y,z)=	O_P\left(\left\{\frac{K_0}{n}+K_0^{-2 \omega_r}\right\}+\left\{\frac{K}{n}+K^{-2\omega_\pi}\right\}\right);
			\end{equation*}
\begin{equation*}
				\frac{1}{n}\sum_{t=1}^{n}|\widehat{\pi}_{K,K_0}(X_t, Y_t, Z_t)-\pi_0(X_t, Y_t, Z_t)|^2=O_P\left(\left\{\frac{K_0}{n}+K_0^{- 2\omega_r}\right\}+\left\{\frac{K}{n}+K^{-2\omega_\pi}\right\}\right).
			\end{equation*}
		\end{theorem}
Theorem \ref{th_crate} (i) extends the result of \cite{matsushita2022estimating} for i.i.d. data to weakly dependent data. Theorem \ref{th_crate} (ii) gives the convergence rate of $\widehat{\pi}_{K,K_0}(x,y,z)$ under $L^2$ and $L^{\infty}$. It has two components: the estimation error of $r_0(y,z)$ (i.e. $K_0/n+K_0^{- 2\omega_r}$) and the regression error of $\pi_0(x,y,z)$ on $v_{K}(x,y,z)$ (i.e. $K/n+K^{-2\omega_\pi}$).

\subsection{Asymptotic Distribution under $H_0$}
Notice that, under the null hypothesis $H_0:X\perp Y|Z$, we have $\pi_0(x,y,z)\equiv 1$ for almost all $(x,y,z)\in \mathcal{X}%
		\times \mathcal{Y}\times \mathcal{Z}$ and Assumption \ref{ass:cpi0-pi*} holds with $\omega_{\pi} =\infty $.  Since the sieve basis $v_K(x,y,z)$ contains the constant ``1" for any fixed $K$, we have $\pi_{K}^*(x,y,z)=\pi_0(x,y,z)\equiv 1$ for all $(x,y,z)\in\mathcal{X}\times\mathcal{Y}\times\mathcal{Z}$ under $H_0$. We write 
		\begin{align}
			\widehat{I}_{K,K_0}&=\frac{1}{2n}\sum_{t=1}^{n}\{\widehat{\pi}_{K,K_0}(X_t, Y_t, Z_t)-1\}^2\widehat{r}_{K_0}(Y_t,Z_t) \notag\\
			&=\frac{1}{2n}\sum_{t=1}^{n}\{\widehat{\pi}_{K,K_0}(X_t, Y_t, Z_t)-\pi_{K}^*(X_t, Y_t, Z_t)\}^2\widehat{r}_{K_0}(Y_t,Z_t) \notag\\
&=\frac{1}{2}\left( \widehat{\boldsymbol{\beta}}_{K,K_0}-\boldsymbol{\beta}^*_K\right)^{\top} \widehat{H}_{K,K_0} \left( \widehat{\boldsymbol{\beta}}_{K,K_0}-\boldsymbol{\beta}^*_K\right). \label{eq:Ihat_quad}
		\end{align}
	
The key step in deriving the asymptotic distribution of $\widehat{I}_{K,K_0}$ is finding the influence function (Bahadur representation) of $\widehat{\boldsymbol{\beta}}_{K,K_0}-\boldsymbol{\beta}^*_K$, which is given in Lemma \ref{prop:vtilde}. To establish this result, we need the following conditions.
		\begin{assumption}\label{ass:H0} We assume 	(i) $K_0^{-2 \omega_r}=o(\sqrt{K}/n)$, (ii) $K_0^2=o(\sqrt{n})$,  and (iii) $\xi_KK_0=o(\sqrt{n})$ and $\xi_KK=o(\sqrt{n})$.
		\end{assumption}  
        
        \begin{assumption}\label{ass:proj_error}
 Let   $\text{proj}_{u_{K_0}}v_K(x,y,z):=\mathbb{E}[v_K(X,Y,Z)u_{K_0}^\top(Y,Z)]\mathbb{E}[u_{K_0}(Y,Z)u_{K_0}^\top(Y,Z)]^{-1}$ $\times u_{K_0}(y,z)$ be the least-square projection of $v_K(X,Y,Z)$ on the linear space spanned by $u_{K_0}(y,z)$. Suppose there exists a positive constant $\omega_0>0$ such that
            $$\sup_{(y,z)\in\mathcal{Y}\times\mathcal{Z}}\left\Vert\mathbb{E}\left[v_K(X,Y,Z)|Y=y,Z=z\right]-\text{proj}_{u_{K_0}}v_K(x,y,z)\right\Vert=O\left(\sqrt{K}K_0^{-\omega_0}\right)$$
             and we assume $KK_0^{-2\omega_0}=o(\sqrt{K}/n)$.
        \end{assumption}
      Assumptions  \ref{ass:H0} and \ref{ass:proj_error}  impose conditions on the growth rate of the smoothing parameters and the sieve approximation error to ensure the asymptotic negligibility of the remainder in Lemma \ref{prop:vtilde}. Similar assumptions are also imposed in \cite{ai2021testing}.
\begin{lemma}\label{prop:vtilde}
	Suppose that Assumptions \ref{ass:ciid}-\ref{ass:proj_error} hold. We obtain 
\begin{equation*}
	 \widehat{\boldsymbol{\beta}}_{K,K_0}-\boldsymbol{\beta}^*_K=\frac{1}{n}\sum_{t=1}^{n}H_{K}^{-1}\widetilde v_{K}(X_t,Y_t,Z_t)+o_P\left(\frac{K^{1/4}}{\sqrt{n}}\right),
	\end{equation*}
	where
	\begin{align*}
		\widetilde v_{K}(X_t,Y_t,Z_t):=&\mathbb{E}\left[v_K(X_t,Y_t,Z_t)r_0(Y_t,Z_t)\pi_{0}(X_t,Y_t,Z_t)|X_t,Z_t\right]\\
  &+\mathbb{E}\left[v_K(X_t,Y_t,Z_t)r_0(Y_t,Z_t)\pi_{0}(X_t,Y_t,Z_t)|Y_t\right]\\	&-\mathbb{E}\left[v_K(X_t,Y_t,Z_t)r_0(Y_t,Z_t)\pi_K^*(X_t,Y_t,Z_t)|Y_t\right]\\
  &-\mathbb{E}\left[v_K(X_t,Y_t,Z_t)r_0(Y_t,Z_t)\pi_K^*(X_t,Y_t,Z_t)|Z_t\right]\\
		&+\mathbb{E}[v_{K}(X,Y,Z)r_0(Y,Z)\pi_K^*(X,Y,Z)]
		\\
		&-v_{K}(X_t,Y_t,Z_t)r_0(Y_t,Z_t)\pi_K^*(X_t,Y_t,Z_t).
	\end{align*}
In addition, if $H_0$ holds, then $\pi_K^*=\pi_{0}$ and $\widetilde v_{K}(X_t,Y_t,Z_t)$ reduces to
\begin{equation*}
\begin{aligned}
	\widetilde v_{K}(X_t,Y_t,Z_t)
=&\mathbb{E}\left[v_K(X_t,Y_t,Z_t)r_0(Y_t,Z_t)|X_t,Z_t\right]-\mathbb{E}\left[v_K(X_t,Y_t,Z_t)r_0(Y_t,Z_t)|Z_t\right]\\
&+\mathbb{E}\left[v_K(X,Y,Z)r_0(Y,Z)\right]-v_{K}(X_t,Y_t,Z_t)r_0(Y_t,Z_t).
\end{aligned}
\end{equation*}
\end{lemma}
The proof of Lemma \ref{prop:vtilde} is relegated to Appendix 	\ref{app:lemma_vtilde}.
With \eqref{eq:Ihat_quad} and Lemma \ref{prop:vtilde}, we establish the limiting distribution of $\widehat{I}_{K,K_0}$ under $H_0$,  whose proof is relegated to Appendix \ref{app:th_cH0_fixk}. 
		\begin{theorem}\label{th_cH0_fixk}
Suppose that Assumptions  \ref{ass:ciid}-\ref{ass:proj_error}  hold. Then, for every fixed $K$ and as $n\rightarrow \infty$, 
			\begin{equation*}
				2n\times \widehat{I}_{K,K_0}\xrightarrow{d}\int|\mathbb{G}_0(x,y,z)|^2dF_{X,Y,Z}(x,y,z)\overset{d}{=} \sum_{j=1}^\infty \lambda_j
				\chi^2_j(1), \text{ under } H_0.
			\end{equation*}
where $\mathbb{G}_0(\cdot)$ is a Gaussian process with mean zero and covariance function $\{V_{K}((x,y,z),\\(x',y',z')):(x,y,z),(x',y',z')\in \mathcal{X}\times\mathcal{Y}\times\mathcal{Z}\}$ defined in (\ref{cov_VK}),  and $\chi^2_j(1)$s are independent chi-squared random variables with one degree of freedom  and the nonnegative constants, $\{\lambda_j\}$, are eigenvalues of the following equation: 
			\begin{align*}  
				\int V_{K}((x,y,z),(x^{\prime },y^{\prime
				},z^{\prime}))\phi(x,y,z)dF_{X,Y,Z}(x,y,z)=\lambda\phi(x^{\prime },y^{\prime },z^{\prime}).
			\end{align*}
		\end{theorem}
Theorem \ref{th_cH0_fixk} shows that the proposed test statistic converges to a limiting distribution with a parametric rate $n$ under the null. Unfortunately, the proposed test has a nonstandard limiting distribution, which makes critical values challenging to compute. Although bootstrap methods can help calculate the critical values, researchers prefer tests with a standard limiting distribution. Next, we establish such a distribution when $K$ goes to infinity.  Denote $O_t:=(X_t^\top,Y_t^\top,Z_t^\top)^\top\in  \mathbb{R}^{d_{X}+d_Y+d_Z}$,  $o:=(x^\top,y^\top,z^\top)^\top$, $\mathcal{O}=\mathcal{X}\times\mathcal{Y}\times\mathcal{Z}$. We impose the following additional conditions.
		\begin{assumption}\label{ass:Kinf}
			 Suppose  $n^{\gamma_1}/\sqrt{K}=O(1)$ for an arbitrarily small constant $\gamma_1>0.$
		\end{assumption}
  \begin{assumption}\label{ass:fo_ratio}
     Suppose that $\max_{1\leq i\leq n}\sup_{(o_0,o_i)\in\mathcal{O}\times\mathcal{O}} f_{O_0,O_i}\left(o_0,o_i\right)/\left\{f_{O_0}(o_0)f_{O_i}(o_i)\right\}<\infty. $
  \end{assumption}
  \begin{assumption}\label{ass:covK}
      Suppose that there exists a universal constant C such that \begin{equation*}
          \max_{1\leq i<j\leq n}\sum_{k=1}^KCov\left(v_{K,k}(O_i),v_{K,k}(O_j)\right)<C,
      \end{equation*}
      where $C$ does not depend on $K$.
  \end{assumption}
 Assumptions \ref{ass:Kinf} and \ref{ass:fo_ratio} are used to establish the asymptotic normality of the standardized test statistic by applying the dependent $U$-statistic theory of Lemma \ref{T1997} in the Appendix \ref{appendix:pre}. Assumption \ref{ass:Kinf} enables $K$ to grow at an arbitrarily slow polynomial rate, which is critical to demonstrate that our proposed test suffers only a slight loss of power; see Theorems \ref{th_cH0} and \ref{th_cH1n}. Assumption \ref{ass:covK} requires that the linear correlation between $v_{K,k}(O_i)$ and $v_{K,k}(O_j)$ for $i\neq j$ decays sufficiently fast as the degree of the nonlinear transform, $K$, increases. This assumption can be easily verified. For example, let $O_t=(X_t,Y_t,Z_t)^{\top}\in\mathbb{R}^3$ and $O_{t}=\rho O_{t-1}+\varepsilon_t$ be the autoregressive model (AR),  where $0<\rho<1$ and $\varepsilon_t\sim N(\boldsymbol{0},I_{3\times 3})$. Let $v_{K}(O)=\left(1,\frac{X^1}{\sqrt{\mathbb{E}\left[X^{2}\right]}},...,\frac{X^{k_1}}{\sqrt{\mathbb{E}\left[X^{2k_1}\right]}},\frac{Y^1}{\sqrt{\mathbb{E}\left[Y^{2}\right]}} ,...,\frac{Y^{k_2}}{\sqrt{\mathbb{E}\left[Y^{2k_2}\right]}},\frac{Z^{1}}{\sqrt{\mathbb{E}\left[Z^{2}\right]}},...,\frac{Z^{k_3}}{\sqrt{\mathbb{E}\left[Z^{2k_3}\right]}}\right)$ be a normalized polynomial sieve, where $k_1+k_2+k_3=K-1$. Then it is easy to derive 
 \begin{align*}
     &\max_{1\leq i<j\leq n}\sum_{k=1}^KCov(v_{K,k}(O_i),v_{K,k}(O_j))=\sum_{k=1}^KCov(v_{K,k}(O_1),v_{K,k}(O_2))\\
     =&\sum_{k=1}^K\left\{\mathbb{E}\left[v_{K,k}(O_1)v_{K,k}(O_2)\right]-\left\{\mathbb{E}\left[v_{K,k}(O)\right]\right\}^2\right\}\\=&1+\sum_{k=1}^{k_1}\rho^k+\sum_{k=1}^{k_2}\rho^k+\sum_{k=1}^{k_3}\rho^k-\sum_{k=1}^K\left\{\mathbb{E}\left[v_{K,k}(O)\right]\right\}^2<1+\frac{3\rho}{1-\rho},
 \end{align*} 
 which implies the validity of Assumption \ref{ass:covK}.
 
 The following theorem is proved in Appendix \ref{app:th_cH0}.
		\begin{theorem}\label{th_cH0}
Suppose that Assumptions \ref{ass:ciid}-\ref{ass:covK} hold. Then, as $K\rightarrow \infty$ and $n\rightarrow \infty$, we obtain
		\begin{equation*}
		\frac{2n\{\widehat{I}_{K,K_0}-B_{K}\}}{\sigma_{K}}\xrightarrow{d}\mathcal{N}(0,1) \text{ under } H_0.
	\end{equation*}
where
	\begin{equation*}
B_{K}:=\frac{1}{2n}\mathbb{E}\left[\widetilde{v}_{K}^\top(O)H_K^{-1}\widetilde{v}_{K}(O)\right],
		\end{equation*}
and  $\overline{O}_0$ is an independent copy of $O_0$,
         \begin{equation*}
\sigma_{K}^2:=2\mathbb{E}\left[\widetilde{V}_{K}^2(O_0,\overline{O}_0)\right]=2\mathbb{E}\left[\{\widetilde v_{K}^\top(O_0)H_K^{-1}\widetilde v_{K}(\overline{O}_0)\}^2\right].
		\end{equation*}
		\end{theorem}
We estimate $B_{K}$ and $\sigma_{K}^2$ by their sample analogs:
			\begin{equation*}
		\begin{aligned}
			\widehat B_{K,K_0}=\frac{1}{2n^2}\sum_{i=1}^{n}\widehat{v}_{K,K_0}^\top(X_i,Y_i,Z_i)\widehat{H}_{K,K_0}^{-1}\widehat{v}_{K,K_0}(X_i,Y_i,Z_i),
		\end{aligned}
	\end{equation*}
		\begin{equation*}
			\widehat{\sigma}_{K,K_0}=\left[\frac{2}{ n(n-1)}\sum_{j=1}^{n}\sum_{i=1,i\neq j}^{n}\left\{\widehat v_{K,K_0}^\top(X_i,Y_i,Z_i)\widehat{H}_{K,K_0}^{-1}\widehat v_{K,K_0}(X_j,Y_j,Z_j)\right\}^2\right]^{1/2},
		\end{equation*}
		\begin{align*}
			\widehat v_{K,K_0}(X_i,Y_i,Z_i)
			=&\frac{1}{n}\sum_{j=1}^{n}v_K(X_i,Y_j,Z_i)-\left\{\frac{1}{n}\sum_{i=1}^{n}v_K(X_i,Y_i,Z_i)u_{K_0}^\top(Y_i,Z_i)\right\}\\
            &\left\{\frac{1}{n}\sum_{j=1}^nu_{K_0}(Y_j,Z_i)\right\}+\left\{\frac{1}{n}\sum_{i=1}^{n}v_K(X_i,Y_i,Z_i)\widehat r_{K_0}(Y_i,Z_i)\right\}
			\\&-v_K(X_i,Y_i,Z_i)\widehat r_{K_0}(Y_i,Z_i).
		\end{align*}
We reject the null hypothesis when $2n\times\{\widetilde{I}_{K,K_0}-\widehat{B}_{K,K_0}\}/\widehat{\sigma}_{K,K_0}>c_{\alpha}$, where $c_{\alpha}$ is the $(1-\alpha)$-quantile of the standard normal distribution. Note that $B_K\asymp K/n$	and $\sigma_K\asymp \sqrt{K}$, hence our proposed test is $n/\sqrt{K}$-consistent with a standard normal limiting distribution. Since we only require $n^{\gamma_1}/\sqrt{K}=O(1)$ in Assumption \ref{ass:Kinf} for an arbitrarily small $\gamma_1>0$ that does not depend on the dimension of the data, our proposed test suffers only a slight loss of the convergence rate.
		
		\subsection{Asymptotic Local Power}
		To analyze the power of the proposed test, we consider the following local alternatives:
		\begin{equation*}
			H_{1n}:  \pi_n(X,Y,Z)=f_{X|Z}(X|Z)f_{Y|Z}(Y|Z)/f_{X,Y|Z}(X,Y|Z) = 1+d_n\cdot\Delta(X,Y,Z),
		\end{equation*}
where $d_n \rightarrow 0$ as $n \rightarrow \infty$ and the function $\Delta(\cdot)$ is continuous. The continuity condition and Assumption \ref{ass:ciid} imply $\Delta(X,Y,Z)$ bounded. For sufficiently small $d_n$, we have 
		\begin{equation*}
			0<\eta_1< \pi_n(x,y,z)<\eta_2<\infty, \quad\forall(x,y,z)\in\mathcal{X} \times \mathcal{Y}\times \mathcal{Z}.
		\end{equation*}
Thus, Assumption \ref{ass:cbounded} with $\pi_0$ replaced by $ \pi_n$ still holds. 		
		Let 
		\begin{equation*}
			\varepsilon_{K}:=\min_\lambda	\sup_{(x,y,z)\in\mathcal{X}\times\mathcal{Y}\times\mathcal{Z}}|\Delta(x,y,z)-\lambda^\top v_{K}(x,y,z)|	
		\end{equation*}
be the $L^\infty$-approximation error of $\Delta(x,y,z)$ using the basis $v_{K}(x,y,z)$. Because $v_{K}(x,y,z)$  contains the constant, we have	
			\begin{align*}		&\min_\lambda\sup_{(x,y,z)\in\mathcal{X}\times\mathcal{Y}\times\mathcal{Z}}| \pi_n(x,y,z)-\lambda^\top v_{K}(x,y,z)|\\
				=&\min_\lambda\sup_{(x,y,z)\in\mathcal{X}\times\mathcal{Y}\times\mathcal{Z}}|1+d_n\Delta(x,y,z)-\lambda^\top v_{K}(x,y,z)|\\
				=&\min_\lambda\sup_{(x,y,z)\in\mathcal{X}\times\mathcal{Y}\times\mathcal{Z}}|d_n\Delta(x,y,z)-\lambda^\top v_{K}(x,y,z)|\\
				=&d_n\cdot \min_\lambda\sup_{(x,y,z)\in\mathcal{X}\times\mathcal{Y}\times\mathcal{Z}}|\Delta(x,y,z)-\lambda^\top v_{K}(x,y,z)/d_n|
				=d_n\varepsilon_{K}.\label{dnepsilon}
			\end{align*}
The approximation error in Assumption \ref{ass:cpi0-pi*} with $ \pi_0(x,y,z)$ replaced by $\pi_n(x,y,z)$ is $d_n\varepsilon_{K}$ , and $\xi_KK^{-\omega_\pi} \rightarrow 0$ in Assumption  \ref{ass:cK} (ii) is replaced by $\xi_K d_n\varepsilon_K\rightarrow 0$.
		
		Let $\text{proj}^{wls}_{v_K} \Delta(x,y,z)$ denote the weighted least-square projection of $\Delta(x,y,z)$ on the linear space spanned by $v_K(x,y,z)$ under the $L^2(r_0(y,z)dF_{X,Y,Z}(x,y,z))$-distance, i.e.
        \begin{align*}
		\text{proj}^{wls}_{v_K}\phi(x,y,z)
		:=&\mathbb{E}\left[r_0(Y,Z)\phi(X,Y,Z)v_K^\top(X,Y,Z)\right]\notag \\
  &\times \mathbb{E}\left[r_0(Y,Z)v_K(X,Y,Z)v_K^\top(X,Y,Z)\right]^{-1}v_K(x,y,z).
        \end{align*}
        The power of the proposed test for a fixed $K$ against the local alternative $H_{1n}$ is given in the following theorem, which is proved in Appendix \ref{sec:th_cH1n_fixK}.
		\begin{theorem}\label{th_cH1n_fixK}
Suppose that Assumptions \ref{ass:ciid}-\ref{ass:proj_error} hold. For every fixed $K$ and $d_n = n^{-1/2}$, we have 
			\begin{equation*}
			2n\times \widehat I_{K,K_0}\xrightarrow{d}\int\left\{\mathbb{G}(x,y,z)+\sqrt{r_0(y,z)}\ \text{proj}^{wls}_{v_K}\Delta(x,y,z)\right\}^2dF_{X,Y,Z}(x,y,z)\text{ under }  H_{1n}.
		\end{equation*}
where $\mathbb{G}(\cdot)$ is a Gaussian process with mean zero and covariance function defined in Appendix \ref{sec:th_cH1n_fixK}. 
		\end{theorem}
		Theorem \ref{th_cH1n_fixK} shows that the proposed test does not have power against all local deviations. For example, the proposed test cannot detect any deviation $\Delta(x,y,z)$ satisfying $\text{proj}^{wls}_{v_K}\Delta(x,y,z)=0$ with probability one. Thus, to gain power against all local alternatives, we must have one that  $\text{proj}^{wls}_{v_K}\Delta(x,y,z)\neq 0$ holds with a positive probability for all deviations $\Delta(x,y,z)$. One way to ensure that is to allow $K$ to grow so that $\text{proj}^{wls}_{v_K}\Delta(x,y,z)$ converges to $\Delta(x,y,z)$. The local power is given in the theorem below for a slow growing, which is proved  in Appendix \ref{sec:th_cH1n}.
		\begin{theorem}\label{th_cH1n}
Suppose that Assumptions  \ref{ass:ciid}-\ref{ass:covK} hold, with $d_n=\sigma_{K}^{1/2}/n^{1/2}\asymp {K}^{1/4}/n^{1/2}$. As $K \rightarrow \infty$ and $n \rightarrow \infty$, we have
          		\begin{equation*}
			\begin{aligned}
				\frac{2n}{ \sigma_{K}}(\widehat{I}_{K,K_0}- B_{K})\xrightarrow{d}\mathbb{E}[\Delta^2(X,Y,Z)r_0(Y,Z)]+\mathcal N\left( 0,1\right)  \text{ under }  H_{1n}. 
			\end{aligned}
		\end{equation*}
   
		\end{theorem}
Since $K$ is only required to grow arbitrarily slowly according to Assumption \ref{ass:Kinf}, Theorem \ref{th_cH1n} shows that, with a slight loss of local power, the proposed test can detect all local deviations satisfying $\mathbb{E}[\Delta^2(X,Y,Z)r_0(Y,Z)]\neq 0$.

\subsection{Asymptotic Global Power}
Under a fixed alternative, the density ratio $\pi_0$ is no longer constant, and the approximation error is no longer zero. The error must be restricted so that it will not bias the proposed test. So Assumptions \ref{ass:H0} (i) and \ref{ass:proj_error}  are strengthened to Assumption \ref {ass:K0omega} with the following theorems  proved in Appendix \ref{app:cglobal_fixK} and  \ref{app:cglobal}.
  \begin{assumption}\label{ass:K0omega} Suppose that	$K_0^{- \omega_r}=o(1/\sqrt{n})$ and $KK_0^{-2\omega_0}=o\left(1/n\right)$.
	\end{assumption}
   
		\begin{theorem}
			\label{thm:cglobal_fixK} Suppose that Assumptions  \ref{ass:ciid}-\ref{ass:proj_error}, \ref{ass:K0omega} hold and $K$ is a fixed integer. Under the alternative hypothesis $H_{1}$, we show that
			
			\begin{enumerate}
				\item if $\mathbb{P}\left(\pi _{K}^{\ast }(X,Y,Z)=1\right) =1,$ then%
				\begin{equation*}
					2n\times \widehat{I}_{K,K_0}\xrightarrow{d}\int \left\{\mathbb{G}_1(x,y,z)\right\}
					^{2}dF_{X,Y,Z}(x,y,z),
				\end{equation*}
		 where $\mathbb{G}_1(\cdot)$ is a Gaussian process defined in Appendix \ref{app:cglobal_fixK}.
				
				\item if $\mathbb{P}\left( \pi _{K}^{\ast }(X,Y,Z)=1\right) <1,$ then 
				\begin{equation*}
					\frac{\sqrt{n}\left\{ 2\widehat{I}_{K,K_0}-\mathbb{E}\left[ \left\{\pi_K^*(X,Y,Z)-1\right\}^2r_0(Y,Z) \right] \right\} }{\sqrt{Var\left( \left\{\pi_K^*(X,Y,Z)-1\right\}^2 r_0(Y,Z)+\phi _{K}\left( X,Y,Z\right) \right) }}\xrightarrow{d}\mathcal N(0,1)
				\end{equation*}%
		where $\phi _{K}\left( X,Y,Z\right) $ is defined in Appendix \ref{app:cglobal_fixK}.
			\end{enumerate}
		\end{theorem}
		For a fixed $K$, the proposed test has a parametric convergence rate of $n$, but it cannot detect certain deviations that satisfy $\pi_K^*(X,Y,Z)=1$ a.s.. If $K$ is allowed to grow such that $\mathbb{E}[\{ \pi^*_{K}(X,Y,Z)-1\}^2r_0(Y,Z)]$ converges to a nonzero constant $\mathbb{E}[\{ \pi_{0}(X,Y,Z)-1\}^2r_0(Y,Z)]$, then our proposed test has power against all deviations from the conditional independence. The formal result is stated below.	
		\begin{theorem}\label{thm:cglobal}
Suppose that Assumptions \ref{ass:ciid}-\ref{ass:proj_error}, \ref{ass:K0omega} hold, and $K\to +\infty$. Under the alternative hypothesis $H_1$, we have
			\begin{align*}
				\frac{\sqrt{n}\left\{2\widehat{I}_{K,K_0}-\mathbb{E}\left[\{ \pi^*_{K}(X,Y,Z)-1\}^2r_0(Y,Z)\right]\right\}}{\sqrt{Var\left( \left\{ \pi_{0}(X,Y,Z)-1\right\}^2r_0(Y,Z)+\phi_0(X,Y,Z)\right) }}\xrightarrow{d}\mathcal N(0,1)	\end{align*}
			where $\phi _{0}\left( X,Y,Z\right) $ is defined in Appendix \ref{app:cglobal}.
		\end{theorem}

		\section{Monte Carlo Simulations}\label{sec:simulation}
In this section, we illustrate the finite sample performance of the proposed test through a small-scale simulation study.
		\subsection{Selection of Tuning Parameters}
		
		\label{sec:select_uv} The asymptotic theory in previous sections does not determine a unique value of $(K,K_0)$ or the optimal choice of the base $\{v_{K}(X,Y,Z),$ $u_{K_{0}}(Y,Z)\}$. This presents a dilemma for applied researchers with only one finite sample who need guidance on selecting the basis functions. We use the data-driven approach proposed in \cite{JHVS2001} and consider a set of basis candidates. Let $u_{K_{0},k}(Y,Z)\in\{1,Y,Z,YZ\}$ and $v_{K,k}(X,Y,Z)\in\{X^{i}Y^{j}Z^{l} |i=0,...,4,  j=0,1,2 , l=0,1,2\}$.\footnote{{We can also use other basis
				functions such as $B$-splines. The choice of basis functions does not affect the asymptotic properties of the test, but may affect its finite
				sample performance. We choose the polynomial basis in this study because of its simplicity. }} We choose the basis functions by maximizing the following criteria: 
		\begin{equation*}
			\max_{(v_{K},u_{K_{0}})\in {\mathcal{A}}} \frac{2n\left\{ \widehat{I}_{K,K_0}-\widehat{B}_{K,K_0}\right\}}{\widehat{\sigma}%
			_{K,K_0} } ,
		\end{equation*}%
		where ${\mathcal{A}}:=\{(v_{K},u_{K_{0}}):\text{under}\ H_{0}\ 
		\text{the rejection rate of the test statistic constructed from}$ $%
		(v_{K},u_{K_{0}}) \\ \text{is}\ \alpha \}$, and $\alpha =0.05$ is the nominal size. This data-driven approach maximizes the power of test statistics subject to size control. The detailed algorithm is given below:
		\begin{enumerate}
			\item Generate a bootstrap sample $\{Z_{t}^{\ast }\}_{t=1}^{n}$ from the smoothed kernel density $f(z)=n^{-1}\sum_{t=1}^{n}$ $\Psi _{h}(Z_{t}-z)$,
			where $\Psi _{h}(z)=\Psi (z/h)/h^{d_{Z}}$ is Gaussian kernel with $\Psi
			(\cdot )$ being a product kernel of the universal Gaussian and $%
			h=(4/3)^{-1/5}n^{-1/5}$, following \cite{zxt2018}.
			
			\item For each $t\in \{1,...,n\}$, generate $X_{t}^{\ast }$ and $Y_{t}^{\ast
			}$ independently from the smoothed conditional densities: 
			\begin{align*}
				f(x|Z_{t}^{\ast }) =\frac{\sum_{j=1}^{n}\Psi _{h}(X_{j}-x)\Psi
					_{h}(Z_{j}-Z_{t}^{\ast })}{\sum_{j=1}^{n}\Psi _{h}(Z_{j}-Z_{t}^{\ast })} \ 
				\text{and} \ f(y|Z_{t}^{\ast }) =\frac{\sum_{j=1}^{n}\Psi _{h}(Y_{j}-y)\Psi
					_{h}(Z_{j}-Z_{t}^{\ast })}{\sum_{j=1}^{n}\Psi _{h}(Z_{j}-Z_{t}^{\ast })}.
			\end{align*}%
Then, the bootstrap sample $\{X_{t}^{\ast },Y_{t}^{\ast },Z_{t}^{\ast }\}_{t=1}^{n}$ satisfies $X_{t}^{\ast }\perp Y_{t}^{\ast }|Z_{t}^{\ast }$ for $t \in
			\{1,...,n\}$.
			
			\item Repeat Steps 1 and 2 $B$ times (e.g., $B=100$) to obtain $B$ bootstrap samples. For each bootstrap sample, compute the test statistic $\{\widehat{I}%
			_{K,K_0,b}^{\ast }\}_{b=1}^{B}$ for each basis candidate. The rejection frequency of the test statistic is calculated by 
			\begin{equation*}
				\frac{1}{B}\sum_{b=1}^{B}I\left\{ \frac{2n\left( \widehat{I}_{K,K_0,b}^{\ast }-\widehat{B}_{K,K_0,b}^{\ast
				}\right)}{\widehat{\sigma}_{K,K_0,b}^{\ast
				}} >c_{\alpha }\right\} ,
			\end{equation*}%
where $c_{\alpha }$ is the $(1-\alpha )$-quantile of the standard normal distribution, and $\widehat{B}_{K,K_0,b}^{\ast }$ and $\widehat{\sigma}_{K,K_0,b}^{\ast }$ are computed from the $b^{th}$ bootstrap sample. The set $\mathcal{A}$ collects the basis candidates with approximately the corresponding rejection frequency $\alpha $.
			
			\item The optimal basis functions are chosen from $\mathcal{A}$ to maximize $%
			2n\widehat{\sigma}_{K,K_0}^{-1}\left\{ \widehat{I}_{K,K_0}-\widehat{B}_{K,K_0}\right\} $.
		\end{enumerate}
		
		\subsection{Simulation Results}\label{sec:simu_I}
		
We mimic the simulation designs in \cite{su2008nonparametric} and consider the following designs where $\{\varepsilon_{1,t},\varepsilon_{2,t}\}$ are i.i.d. $\mathcal{N}(0,I_2)$.

		\begin{itemize}
			\item DGP1s: $Y_t=0.5Y_{t-1}+\varepsilon_{1,t}$, $X_t=0.5X_{t-1}+\varepsilon_{2,t}$.
			\item DGP2s: $Y_t=\sqrt{h_t}\varepsilon_{1,t}$,  $X_t=0.5X_{t-1}+\varepsilon_{2,t}$, $h_t=0.01+0.5Y_{t-1}^2$.
			
			\item DGP3s: 
$Y_t=\sqrt{h_{1, t}}\varepsilon_{1, t}$, $X_t=\sqrt{h_{2, t}}\varepsilon_{2, t}$, where $h_{1, t}=0.01+0.9h_{1, t-1}+0.05Y_{t-1}^2$, 
                     $h_{2,t}=0.01+0.9h_{2, t-1}+0.05X_{t-1}^2$.
               
			\item DGP1p: $Y_t=0.5Y_{t-1}+0.5X_{t-1}+\varepsilon_{1,t}$,  $X_t=0.5X_{t-1}+\varepsilon_{2,t}$.
			
			\item DGP2p: $Y_t=0.5Y_{t-1}+0.5X_{t-1}^2+\varepsilon_{1,t}$,  $X_t=0.5X_{t-1}+\varepsilon_{2,t}$.
			
			\item DGP3p: $Y_t=0.5Y_{t-1}X_{t-1}+\varepsilon_{1,t}$,  $X_t=0.5X_{t-1}+\varepsilon_{2,t}$.
			
			\item DGP4p: $Y_t=0.5Y_{t-1}+0.5X_{t-1}\varepsilon_{1,t}$,  $X_t=0.5X_{t-1}+\varepsilon_{2,t}$.

            \item DGP5p: $Y_t=\sqrt{h_{1, t}}\varepsilon_{1, t}$, $X_t=\sqrt{h_{2, t}}\varepsilon_{2, t}$, where $h_{1, t}=0.01+0.1h_{1, t-1}+0.4Y_{t-1}^2+0.1X_{t-1}^2$, 
                     $h_{2,t}=0.01+0.9h_{2, t-1}+0.05X_{t-1}^2$.
            \item DGP6p: $Y_t=0.5Y_{t-1}+4\varphi\left(Y_{t-1}/0.1\right)X_{t-1}+\varepsilon_{1,t},  \varphi(x)=\frac{1}{\sqrt{2\pi}}e^{-\frac{x^2}{2}}$, $X_t=0.5X_{t-1}+\varepsilon_{2,t}$.
		\end{itemize}
	
Here, we test whether $X_{t-1}$ is independent of $Y_t$ conditional on $Y_{t-1}$.  DGP1s-3s are designed to examine the sizes of the tests, while DGP1p-6p are intended to check the powers. We perform Monte Carlo trials 1,000 times, and the optimal basis is chosen from the candidates using the data-driven algorithm mentioned in section \ref{sec:select_uv}. For comparison with our density ratio (DR) test, we compute the mutual information (MI) test statistic of \cite{ai2021testing}, the Hellinger (HEL) test statistic of \cite{su2008nonparametric}, and the characteristic function (CF) based test statistic of \cite{wang2018characteristic}. Following \cite{su2008nonparametric}, the critical values for the above alternatives are computed using the bootstrap procedure with $100$ repetitions. We also compute the standard linear Granger causality test (LIN) proposed by \cite{Granger} with 1000 repetitions, where we examined whether $X_{t-1}$  should enter the regression of $Y_t$ on $Y_{t-1}$ linearly. 

Table \ref{table_model_DF_CIT_1} reports the rejection rates of the five tests at the 5\% level when the sample sizes are 100, 150 and 200 respectively.
		 It reveals that under DGP1s-3s, the sizes of all tests are reasonable, and the power of the tests under DGP1p-6p
			generally increases with the sample size $n$ , and our test outperforms 
			other tests in most cases, especially in DGP6p. 		
		\begin{table}[hbtp]
			\caption{Comparison of tests  at the 5\% level}
			\label{table_model_DF_CIT_1}
			\begin{center}
				\vskip 0.2in {\fontsize{10pt}{16pt} \selectfont 
					\begin{tabular}{cc|cccccccccc}
						\hline\hline
						$n$ &  & DGP1s  & DGP2s&DGP3s&DGP1p&DGP2p&DGP3p&DGP4p&DGP5p&DGP6p  \\ 
						\hline\hline
						& Lin&0.064&0.045&0.059&\textbf{0.999}&0.334&0.209&0.287&0.119&0.291 \\ 
    						 &MI&0.071&0.057&0.057&0.301&\textbf{0.973}&0.251&0.677&0.082&0.064\\ 
                             $100$
                        &HEL&0.070&0.072	&0.164&	0.576	&0.705&	0.386	&\textbf{0.978}&0.369&0.133\\
                        &CF&0.028&	0.037&	0.037	&0.434	&0.281	&0.325	&0.675&0.135	&0.104\\
						&DR &0.060&0.067&0.068&0.947&0.955&\textbf{0.488}&0.837&\textbf{0.447}&\textbf{0.663}\\

						\hline
					& Lin &0.046&0.045&0.053&\textbf{1.000}&0.351&0.246&0.263
                &0.142&0.361\\ 
						 &MI&0.095&0.052&0.056&0.502&\textbf{0.989}&0.361&0.838&0.137&0.065\\ 
                         $150$
					&HEL&0.080	&0.065	&0.144	&0.755	&0.865&	0.498	&\textbf{0.998}&0.448&0.152\\	&CF&0.033&	0.041&	0.049	&0.626	&0.387	&0.545&	0.830&0.233	&0.198\\
                    &DR&0.047&0.062&0.065&0.958 &0.967&\textbf{0.589}&0.851&\textbf{0.527}&\textbf{0.773}\\ 
\hline
                        
						& Lin &0.045&0.054&0.057&\textbf{1.000}&0.347&0.217&0.267&0.117&0.497\\ 
						 &MI&0.111&0.075&0.058&0.624&\textbf{0.998}&0.519&0.924&0.224&0.054\\ 
                         $200$
                        &HEL&0.059	&0.067	&0.137	&0.879	&0.940	&0.572&	\textbf{1.000}&0.515&0.169\\
						&CF&0.036&	0.042	&0.049	&0.777	&0.479	&0.704	&0.921&0.306	&0.316\\
                        
                        &DR&0.070&0.050 &0.065&0.972&0.964&\textbf{0.739}&0.913&\textbf{0.616}&\textbf{0.828}\\ 
						\hline\hline
					\end{tabular}
				}
			\end{center}
		\end{table}

		\section{Application}\label{sec:application}
		\begin{comment}
			Boston housing data was collected from 506  census tracts of Boston from the 1970 census (\cite{HARRISON197881}). Let C, T, and L denote per capita crime rate by town,	full-value property-tax rate per USD 10,000, and	percentage of lower status of the population, respectively. In a study of the crime rate and the lower status of the population,  we aim to test 
			\begin{equation*}
				H_0: C \perp L | T
			\end{equation*}
			We choose the polynomial bases
			$v_K(X,Y,Z)\subseteq \{1,X^3,X^2,X,Y^2,Y,Z^2,Z\}$ and $u_{K_0}(Y,Z)=\{1,Y,Z,YZ\}$. The sample size is $ n = 506$, and the critical value is determined by a bootstrap procedure with 100 repetitions. The p-values for the MI (Jackknife), MI, and DR tests are 1.00, 0.30, and 0.18, respectively, all larger than the 5\% significance level. Thus, we do not reject the null hypothesis.
		\end{comment}
		The relationship between stock prices and trading volume is widely investigated, and previous studies have given several explanations. This section will  check the interaction between returns and volume for the Dow Jones price index. We collect the daily adjusted closing prices and trading volume from January 2nd, 2014, to December 30th, 2022. The data that support the findings of this study are openly available in Yahoo Finance by using the command ``yf.download" in Python. Let $P_t$ and $V_t$ be 100 times the natural logarithm of the  stock prices  and daily trading volume for day $t$.  Granger causality
tests will be conducted on the first differenced data $\Delta P_t$ and $\Delta V_t$.
    The augmented Dickey-Fuller Test indicates that both  $\Delta P_t$ and $\Delta V_t$ are stationary. To test whether $\Delta V_t$ Granger causes $\Delta P_t$ linearly  at lag one, we check  whether 
    \begin{equation*}
    H_0: b_1=0
    \end{equation*}
    in $\Delta P_t=a_0 +a_1\Delta P_{t-1}+b_1\Delta V_{t-1}$. To test whether $\Delta V_t$ Granger causes $\Delta P_t$ nonlinearly, we check whether
		\begin{equation*}
			H_0: \Delta P_t \perp \Delta V_{t-1} | \Delta P_{t-1}.
	\end{equation*}
Both of the null hypotheses above are represented as ``$\Delta V_{t-1} \nRightarrow \Delta P_t$". Similarly, we can also define  "$\Delta P_{t-1} \nRightarrow \Delta V_t$".
		We choose the optimal bases from candidates through the data-driven algorithm mentioned in section \ref{sec:select_uv}.  The critical value is determined by a bootstrap procedure with 200 repetitions. The p-values for LIN (\cite{Granger}),  MI (\cite{ai2021testing}), HEL(\cite{su2008nonparametric}) and DR (proposed)  tests are revealed in Table \ref{table_pvalue}. DR and HEL indicate a bidirectional Granger causal link between returns and  volume for the Dow Jones price index at the 5\% level. This is consistent with \cite{Hiemstra1994} , who found evidence of bidirectional  nonlinear Granger causality  between stock  returns
and  trading volume based on daily Dow Jones price index data from 1915 to 1990. However, LIN and MI fail to find any significant Granger causal relationship  between returns and  volume  at the 5\% level. 
\begin{table}[htbp]
			\caption{p values for testing the relationship between returns and volume}
			\label{table_pvalue}
			\begin{center}
				\vskip 0.2in {\fontsize{10pt}{16pt} \selectfont 
					\begin{tabular}{c|cc}
						\hline\hline
						$H_0$ &$\Delta V_{t-1} \nRightarrow \Delta P_t$  & $\Delta P_{t-1} \nRightarrow \Delta V_t$  \\ 
						\hline\hline
						 Lin&0.120&0.051 \\ 
                        MI&0.100&0.100\\
                        HEL&0.015&0.025\\
                        DR&0.015&0.035\\
				
						\hline\hline
					\end{tabular}
				}
			\end{center}
		\end{table}
		
		\section{Conclusions}
		This paper proposes a conditional independence test for weakly dependent data by estimating the conditional density ratio under the weighted least-squares distance. The test is asymptotically pivotal with a standard
limiting distribution and is consistent against all departures from conditional independence,  suffering only a slight loss of local power. Our test is easy to implement, and a small-scale simulation study illustrates the usefulness of our approach. It reveals that the test statistic is distribution-free. Thus, our test can be applied to various
dependence structures.

		%\section*{Acknowledgments}
		%Zheng Zhang is supported by the fund from the National Key R\&D Program of China [grant number 2022YFA1008300],  the fundamental research funds for the central universities, and the research funds of Renmin University of China [grant number 23XNA025]. Chunrong Ai gratefully acknowledges funding from Project 72133005, supported by the NSFC, and Project 2022B1515120060, supported by the Guangdong Basic and Applied Basic Research Foundation. Zixuan Xu acknowledges the financial support from the Renmin University of China. The present work constitutes part of her research study leading to her doctoral thesis at Renmin University of China.  

\section*{Disclosure Statement}
        The authors report that there are no financial or non-financial competing interests to declare.
        
	\bibliographystyle{chicago}
		\bibliography{references}		
		\clearpage
		\appendix
		\section*{Appendix}

		\section{Proof of \eqref{eq:invariance}}\label{app:invariance}
Without loss of generality, we assume $d_X=d_Y=d_Z=1$. Note that $x=\phi_X^{-1}(\widetilde{x}), y=\phi_Y^{-1}(\widetilde{y}), z=\phi_Z^{-1}(\widetilde{z})$ are well-defined in light of the monotonicity of the transformation,  then we have
		\begin{equation*}
			\begin{aligned}
				f_{\widetilde{X}|\widetilde{Z}}(\widetilde{x}|\widetilde{z})=&\frac{f_{\widetilde{X},\widetilde{Z}}(\widetilde{x},\widetilde{z})}{f_{\widetilde{Z}}(\widetilde{z})}
				=\frac{f_{\widetilde{X},\widetilde{Z}}(\phi_X(x),\phi_Z(z))}{f_{\widetilde{Z}}(\phi_Z(z))}
				=\frac{f_{X,Z}(x,z)/|\phi_X'(x)\phi_Z'(z)|}{f_Z(z)/|\phi_Z'(z)|}
				=\frac{f_{X|Z}(x|z)}{|\phi_X'(x)|},
			\end{aligned}
		\end{equation*}
		where $\phi_X'(\cdot)$ is the derivative of $\phi_X(\cdot)$.
		Similarly, we have
		\begin{align*}
			f_{\widetilde{Y}|\widetilde{Z}}(\widetilde{y}|\widetilde{z})=\frac{f_{Y|Z}(y|z)}{|\phi_Y'(y)|},\ 
			f_{\widetilde{X},\widetilde{Y}|\widetilde{Z}}(\widetilde{x},\widetilde{y}|\widetilde{z})=\frac{f_{X,Y|Z}(x,y|z)}{|\phi_X'(x)\phi_Y'(y)|}, f_{\widetilde{X},\widetilde{Y},\widetilde{Z}}(\widetilde{x},\widetilde{y},\widetilde{z})=\frac{f_{X,Y,Z}(x,y,z)}{|\phi_X'(x)\phi_Y'(y)\phi_Z'(z)|}.
		\end{align*}
		Then we obtain the desired result:
		\begin{equation*}
			\begin{aligned}
				&\int\left\{\frac{f_{\widetilde{X}|\widetilde{Z}}(\widetilde{x}|\widetilde{z})f_{\widetilde{Y}|\widetilde{Z}}(\widetilde{y}|\widetilde{z})}{f_{\widetilde{X},\widetilde{Y}|\widetilde{Z}}(\widetilde{x},\widetilde{y}|\widetilde{z})}-1\right\}^2\frac{f_{\widetilde{Y}}(y)f_{\widetilde{Z}}(z)}{f_{\widetilde{Y},\widetilde{Z}}(y,z)}f_{\widetilde{X},\widetilde{Y},\widetilde{Z}}(\widetilde{x},\widetilde{y},\widetilde{z})d\widetilde{x}d\widetilde{y}d\widetilde{z}\\
				=&\int\left\{\frac{f_{X|Z}(x|z)f_{Y|Z}(y|z)}{f_{X,Y|Z}(x,y|z)}-1\right\}^2\frac{f_Y(y)f_Z(z)}{f_{Y,Z}(y,z)}\frac{f_{X,Y,Z}(x,y,z)}{|\phi_X'(x)\phi_Y'(y)\phi_Z'(z)|}|\phi_X'(x)\phi_Y'(y)\phi_Z'(z)|dxdydz\\
				=&\int\left\{\frac{f_{X|Z}(x|z)f_{Y|Z}(y|z)}{f_{X,Y|Z}(x,y|z)}-1\right\}^2\frac{f_Y(y)f_Z(z)}{f_{Y,Z}(y,z)}  f_{X,Y,Z}(x,y,z)dxdydz.\\
			\end{aligned}
		\end{equation*}

		\clearpage
	
  \newpage
		
\section{Preliminaries}\label{appendix:pre}
Let $\{O_t, t\in\mathbb{N}\}$ be a strictly stationary stochastic
process of dimension $d\in \mathbb{N}$ and its distribution be denoted by $F(\cdot)$.  Let $\mathcal{F}_s^t$ denote the $\sigma$-algebra generated by $(O_s,...,O_t)$ for $s \leq t$. The stochastic process $\{O_t, t\in\mathbb{N}\}$
is $\beta$-mixing if the following condition holds:
$$\beta_{\tau}:=\sup_{s\in \mathbb{N}}\mathbb{E}\left[\sup_{A \in \mathcal{F}_{s+\tau}^\infty}\left\{|\mathbb{P}\left(A|\mathcal{F}_{0}^s \right)-\mathbb{P}(A)|\right\}\right]\rightarrow 0, \ \text{as} \ \tau\rightarrow \infty.$$ The following result is from \cite{bosq2012nonparametric}.
\begin{lemma} [Davydov's inequality]  \label{Davydov} Let $\{O_t, t\in\mathbb{N}\}$ be a strictly stationary $\beta$-mixing process. For any $q>1$, $r>1$, and $p>0$ satisfying $1/q+1/r=1-1/p$, we have
\begin{align*}
 |Cov(f(O_t),f(O_{t+\tau}))|\leq2p\beta_\tau^{1/p}\Vert f(O_t)\Vert_q\Vert f(O_{t+\tau})\Vert_r,
\end{align*}
where $f(\cdot)$ is  a measurable function satisfying  $f(O_t) \in L^q(P)\cap L^{r}(P)$.
\end{lemma}
 For a symmetric vector-valued function $h(\cdot)$ of $m$ variables,  let $U_n$ be a $U$-statistic generated from $h(\cdot)$: 
$$
U_n:=\binom{n}{m}^{-1} \sum_{1 \leq t_1<t_2<\ldots<t_m \leq n} h\left(O_{t_1}, \ldots, O_{t_m}\right) \quad(n \geq m) .
$$
By the Hoeffding decomposition, $U_n$ can be represented by
$$
U_n=\sum_{c=0}^m\binom{m}{c} U_n^c,
$$
where $U_n^c$ is the $U$-statistic constructed from the degenerate kernel function
$$
\widehat{h}_c\left(o_1, \ldots, o_c\right):=\sum_{r=0}^c\binom{c}{r}(-1)^{c-r} h_r\left(o_1, \ldots, o_r\right)$$
and
$
h_r\left(o_1, \ldots, o_r\right):= \int h\left(o_1, \ldots, o_m\right) \prod_{i=r+1}^m d F\left(o_i\right).
$

Denote the reminder of $U_n$ by
$$
R_n:=\sum_{c=2}^m\binom{m}{c} U_n^c.
$$
Then $U_n$ can be further represented by 
\begin{equation}\label{Un_decomposition}
    U_n=\theta+\frac{m}{n}\sum_{t=1}^n\left(h_1(O_t)-\theta\right)+R_n,
\end{equation}
where $\theta:=\int  h\left(o_1, \ldots, o_m\right) \prod_{i=1}^m d F(o_i)$. Define
$$s_\delta:=\sup_{1\leq t_1<...<t_m\leq n}\left\{\mathbb{E}\left[\left \Vert h\left(O_{t_1},...,O_{t_m}\right)\right\Vert ^{2+\delta}\right]\right\}^{\frac{1}{2+\delta}} \quad (\delta \geq 0).$$
The following result is from \cite{denker1983u}, which gives a $L^2(P)$  bound for $R_n$. 
\begin{lemma}\label{lemma:denker1983u}
For two positive constants $\delta$ and $\epsilon$ satisfying $\beta_n^{\delta/(2+\delta)}=O\left(n^{-2+\epsilon}\right)$,  there exists a constant $\Gamma_\epsilon$ such that
\begin{equation*}
    \mathbb{E}[\Vert R_n\Vert ^2]\leq\Gamma_\epsilon^2n^{-2+\epsilon}s_\delta^2.
\end{equation*}

\end{lemma}
	Let $h_n(\cdot,\cdot)$ be a symmetric function on $\mathbb{R}^{d}\times\mathbb{R}^d$ that possibly depends on the sample size $n$,  such that $\mathbb{E}\left[h_n(O_0,v)\right]=0$.  Define  $$\mathcal{H}_n:=\frac{1}{n}\sum_{1 \leq i<j\leq n}\left\{h_n(O_i,O_j) -\mathbb{E}\left[h_n(O_i,O_j)\right]\right\}.$$  Let $\left\{\overline{O}_t, t\in\mathbb{N}\right\}$ be a sequence of $i.i.d.$ random variables following the distribution $F(\cdot)$, where $\overline{O}_0$ is an independent copy of $O_0$. For any $p>0$, define
	$$u_n(p):=\max\left\{\max_{1\leq i \leq n}\left\Vert h_n(O_i,O_0)\right \Vert_p, \left\Vert h_n(O_0,\overline{O}_0)\right \Vert_p\right\},$$
		$$v_n(p):=\max\left\{\max_{1\leq i \leq n}\left\Vert G_{n0}(O_i,O_0)\right \Vert_p, \left\Vert G_{n0}(O_0,\overline{O}_0)\right \Vert_p\right\},$$
			$$w_n(p):=\left\Vert G_{n0}(O_0,O_0)\right \Vert_p,$$
				$$z_n(p):=\max_{0\leq i\leq n}\max_{1\leq j \leq n}\left\{\left\Vert G_{nj}(O_i,O_0)\right \Vert_p, \left\Vert G_{nj}(O_0,O_i)\right \Vert_p, \left\Vert G_{nj}(O_0,\overline{O}_0)\right \Vert_p\right\},$$
				where $G_{n,j}(u,v):=\mathbb{E}\left[
				h_n(O_j,u)h_n(O_0,v)\right]$ and $\Vert\cdot\Vert_p=\left\{\mathbb{E}|\cdot|^p\right\}^{1/p}$.
    The following result is from \cite{tenreiro1997loi}, which ensures the asymptotic normality of $\mathcal{H}_n$ under sufficient conditions.
		\begin{lemma} \label{T1997}
	 Suppose there exist $\delta_0>0$, $\gamma_0<\frac{1}{2}$ and $\gamma_1>0$ such that	
  \begin{enumerate}[(1)]
      \item 
 $u_n(4+\delta_0)=O\left(n^{\gamma_0}\right)$;
		
	\item $v_n(2)=o(1)$;
		
		\item  $w_n(2+\delta_0/2)=o(n^{1/2})$
		
		\item  $z_n(2)n^{\gamma_1}=O(1)$;
		
		\item $\mathbb{E}\left[h_n^2(O_0,\overline{O}_0)\right]=2\tilde{\sigma}^2+o(1)$.
		 \end{enumerate}
		Then $ \mathcal{H}_n$ is asymptotically normally distributed with mean zero and variance $\tilde{\sigma}^2$.
		\end{lemma}

  \clearpage				
		
\section{Key Lemmas}\label{app:pre}
		For any square-integrable function $\phi(x,y,z)$, let 
        \begin{align*}
\text{proj}_{n,v_K}\phi(x,y,z):=&\left[\sum_{t=1}^{n}\phi(X_t, Y_t, Z_t)v_K^\top(X_t, Y_t, Z_t)\right]\left[\sum_{t=1}^{n}v_K(X_t, Y_t, Z_t)v_K^\top(X_t, Y_t, Z_t)\right]^{-1}\\&v_K(x,y,z)
        \end{align*} 
        be the least squares regression of $\phi(x,y,z)$ on the space linearly spanned by $v_K(x,y,z)$.
		\begin{lemma}\label{lemma:proj}
			Under Assumptions \ref{ass:ciid}, \ref{ass:eigen} and \ref{ass:cK}, we have
			\begin{align}
    &\label{eq:u_eigen}
				\left\|\frac{1}{n}\sum_{t=1}^nu_{K_0}(Y_t,Z_t)u^{\top}_{K_0}(Y_t,Z_t)-\mathbb{E}\left[u_{K_0}(Y,Z)u^{\top}_{K_0}(Y,Z)\right]\right\| =O_P\left(\zeta_{K_0}\sqrt{\frac{K_0}{n}}\right),  \\
	& \label{eq:v_eigen}
				\left\|\frac{1}{n}\sum_{t
					=1}^nv_{K}(X_t, Y_t, Z_t)v^{\top}_{K}(X_t, Y_t, Z_t)-\mathbb{E}\left[v_{K}(X,Y,Z)v^{\top}_{K}(X,Y,Z)\right]\right\|
				=O_P\left(\xi_{K}\sqrt{\frac{K}{n}}\right), \\
 & \label{eq:vu_eigen}
			\left\Vert\frac{1}{n}\sum_{t=1}^{n}v_{K}(X_t, Y_t, Z_t)u_{K_0}^\top(Y_t,Z_t)-\mathbb{E}\left[v_K(X,Y,Z)u_{K_0}^\top(Y,Z)\right]\right\Vert
				=O_P\left(\xi_K\sqrt{\frac{K_0}{n}}\right),
			\end{align}
	and			\begin{align}\label{eq:lsproj}
				\frac{1}{n}\sum_{t=1}^{n}\left\Vert \text{proj}_{n,v_K}\phi(X_t, Y_t, Z_t)\right\Vert^2\leq  \frac{1}{n}\sum_{t=1}^{n}\left\Vert \phi(X_t, Y_t, Z_t)\right\Vert^2
			\end{align}
			and 
			\begin{align}\label{vkphi}
				\left\|\frac{1}{n}\sum_{t=1}^{n}v_{K}(X_t, Y_t, Z_t)\phi(X_t, Y_t, Z_t)\right\|^2
				\leq O_P\left(1\right)\cdot \frac{1}{n}\sum_{t=1}^{n}| \phi(X_t, Y_t, Z_t)|^2.
			\end{align}
		\end{lemma}
		\begin{proof}
We first prove \eqref{eq:u_eigen}. Under Assumptions \ref{ass:ciid}, \ref{ass:eigen} and \ref{ass:cK}, we have	
			\begin{align*}
	&\mathbb{E}\left[\left\|\frac{1}{n}\sum_{t=1}^nu_{K_0}(Y_t,Z_t)u^{\top}_{K_0}(Y_t,Z_t)-\mathbb{E}\left[u_{K_0}(Y,Z)u^{\top}_{K_0}(Y,Z)\right]\right\|^2\right]\\
				=&\frac{1}{n}\mathbb{E}\left[\left\Vert u_{K_0}(Y,Z)u_{K_0}^\top(Y,Z)-\mathbb{E}\left[u_{K_0}(Y,Z)u^{\top}_{K_0}(Y,Z)\right]\right\Vert^2\right]\\
    & +\frac{2}{n^2}\sum_{i=1}^{K_0}\sum_{j=1}^{K_0}\sum_{t=1}^{n}\sum_{\tau =1}^{n-t}Cov\left(u_{K_0,i}(Y_t,Z_t)u_{K_0,j}(Y_t,Z_t), u_{K_0,i}(Y_{t+\tau},Z_{t+\tau})u_{K_0,j}(Y_{t+\tau},Z_{t+\tau})\right) \\
				\leq&\frac{4(4+\eta)}{ \eta n^2 }\sum_{i=1}^{K_0}\sum_{j=1}^{K_0}\sum_{t=1}^{n}\sum_{\tau =1}^{n-t}\left\{\mathbb{E}\left[|u_{K_0,i}(Y_t,Z_t)u_{K_0,j}(Y_t,Z_t)|^{2+\frac{\eta}{2}}\right]\right\}^{\frac{4}{4+\eta}}\beta_{\tau}^{\frac{\eta}{4+\eta}}+O\left( \frac{\zeta_{K_0}^2K_0}{n}\right) \\
				\leq&\frac{4(4+\eta)}{\eta n^2}\sum_{i=1}^{K_0}\sum_{j=1}^{K_0}\sum_{t=1}^{n}\sum_{\tau =1}^{n-t}\left\{\mathbb{E}\left[|u_{K_0,i}(Y_t,Z_t)|^{4+\eta}\right]\right\}^{\frac{2}{4+\eta}}\left\{\mathbb{E}\left[|u_{K_0,j}(Y_t,Z_t)|^{4+\eta}\right]\right\}^{\frac{2}{4+\eta}}\beta_{\tau}^{\frac{\eta}{4+\eta}}\\
                &+O\left( \frac{\zeta_{K_0}^2K_0}{n}\right) \\
				\leq&\frac{4(4+\eta)}{\eta n^2}\sum_{i=1}^{K_0}\sum_{j=1}^{K_0}C^{\frac{4}{4+\eta}}\cdot\sum_{t=1}^{n}\sum_{\tau =1}^{n-t}\beta_{\tau}^{\frac{\eta}{4+\eta}}+O\left( \frac{\zeta_{K_0}^2K_0}{n}\right)  \\
				=&O\left(\frac{K_0^2}{n}\right)  +O\left( \frac{\zeta_{K_0}^2K_0}{n}\right) \quad (\text{by Assumption \ref{ass:ciid}}) \\
				 =&O\left( \frac{\zeta_{K_0}^2K_0}{n}\right), \label{eq:u_eigen_proof}
		\end{align*}
			where $C$ is a universal constant arising from Assumption \ref{ass:eigen},  and the first inequality follows from  Lemma \ref{Davydov} by applying $q=r=2+\eta/2$,  $p=(4+\eta)/\eta$, $f(y,z)=u_{K_0,i}(y,z)u_{K_0,j}(y,z)$, and using the following result:
				\begin{equation*}
				\begin{aligned}
&\frac{1}{n}\mathbb{E}\left[\left\Vert u_{K_0}(Y,Z)u_{K_0}^\top(Y,Z)-\mathbb{E}\left[u_{K_0}(Y,Z)u^{\top}_{K_0}(Y,Z)\right]\right\Vert^2\right]\\			
  \leq  &\frac{1}{n}\mathbb{E}\left[\left\Vert u_{K_0}(Y,Z)u_{K_0}^\top(Y,Z)\right\Vert^2\right]\\
					=&\frac{1}{n}\mathbb{E}\left[\tr\left\{u_{K_0}(Y,Z)u_{K_0}^\top(Y,Z)u_{K_0}(Y,Z)u_{K_0}^\top(Y,Z)\right\}\right]\\
					\leq&\frac{\zeta_{K_0}^2}{n}\tr\left\{\mathbb{E}\left[u_{K_0}(Y,Z)u_{K_0}^\top(Y,Z)\right]\right\}=O\left( \frac{\zeta_{K_0}^2K_0}{n}\right).
				\end{aligned}
			\end{equation*}
Then the result \eqref{eq:u_eigen} follows from	Chebyshev's inequality. The results  \eqref{eq:v_eigen}  and \eqref{eq:vu_eigen} can be proved similarly.  \eqref{eq:lsproj} follows from the projection property of least squares regression. 

Note that $\lambda_{\max}\left\{\frac{1}{n}\sum_{t=1}^{n}v_K(X_t, Y_t, Z_t)v_K^\top(X_t, Y_t, Z_t)\right\}$, the largest eigenvalue of the matrix $\frac{1}{n}\sum_{t=1}^{n}v_K(X_t, Y_t, Z_t)v_K^\top(X_t, Y_t, Z_t)$,   equals $\lambda_{\min}\left\{\left(\frac{1}{n}\sum_{t=1}^{n}v_K(X_t, Y_t, Z_t)v_K^\top(X_t, Y_t, Z_t)\right)^{-1}\right\}$ , which is defined analogously. Under Assumption \ref{ass:eigen}, we have 
			\begin{align}
				&\left\|\frac{1}{n}\sum_{t=1}^{n}v_{K}(X_t, Y_t, Z_t)\phi(X_t, Y_t, Z_t)\right\|^2\notag\\
				=&\left[\frac{1}{n}\sum_{t=1}^{n}v_{K}(X_t, Y_t, Z_t)\phi(X_t, Y_t, Z_t)\right]^\top\frac{\lambda_{\min}\left\{\left(\frac{1}{n}\sum_{t=1}^{n}v_K(X_t, Y_t, Z_t)v_K^\top(X_t, Y_t, Z_t)\right)^{-1}\right\}}{\lambda_{\min}\left\{\left(\frac{1}{n}\sum_{t=1}^{n}v_K(X_t, Y_t, Z_t)v_K^\top(X_t, Y_t, Z_t)\right)^{-1}\right\}}\notag\\
               & \cdot I_{K\times K}\cdot
				\left[\frac{1}{n}\sum_{t=1}^{n}v_K(X_t, Y_t, Z_t)\phi(X_t, Y_t, Z_t)\right]\notag\\
				\leq&\frac{1}{\lambda_{\max}\left\{\frac{1}{n}\sum_{t=1}^{n}v_K(X_t, Y_t, Z_t)v_K^\top(X_t, Y_t, Z_t)\right\}}\left[\frac{1}{n}\sum_{t=1}^{n}v_{K}(X_t, Y_t, Z_t)\phi(X_t, Y_t, Z_t)\right]^\top\notag\\
				&	\cdot\left[\frac{1}{n}\sum_{t=1}^{n}v_K(X_t, Y_t, Z_t)v_K^\top(X_t, Y_t, Z_t)\right]^{-1}\left[\frac{1}{n}\sum_{t=1}^{n}v_K(X_t, Y_t, Z_t
				)v_K^\top(X_t, Y_t, Z_t)\right]\notag\\
				&\cdot\left[\frac{1}{n}\sum_{t=1}^{n}v_K(X_t, Y_t, Z_t)v_K^\top(X_t, Y_t, Z_t)\right]^{-1}\left[\frac{1}{n}\sum_{t=1}^{n}v_K(X_t, Y_t, Z_t)\phi(X_t, Y_t, Z_t)\right]\notag\\
				\leq&O_P\left(1\right)\cdot \frac{1}{n}\sum_{t=1}^{n}| \text{proj}_{n, v_K}\phi(X_t, Y_t, Z_t)|^2	\leq O_P\left(1\right)\cdot \frac{1}{n}\sum_{t=1}^{n}| \phi(X_t, Y_t, Z_t)|^2, \ \text{ (by \eqref{eq:lsproj})} \notag 
			\end{align}
			which gives the result \eqref{vkphi}.
		\end{proof}
	
 \begin{lemma}
 \label{uproj}Under Assumptions \ref{ass:ciid}, \ref{ass:eigen} and \ref{ass:cK}, for any $0<\epsilon<1/4$,  we have
     	\begin{equation}\label{uproj_u}
	\begin{aligned}
		&\frac{1}{n(n-1)}\sum_{j=1,j\neq i}^n\sum_{i=1}^nu_{K_0}(Y_i,Z_j)\\
		=&\frac{1}{n}\sum_{t=1}^{n}	\mathbb{E}\left[u_{K_0}(Y_t,Z_t)r_0(Y_t,Z_t)|Y_t\right]+\frac{1}{n}\sum_{t=1}^{n}	\mathbb{E}\left[u_{K_0}(Y_t,Z_t)r_0(Y_t,Z_t)|Z_t\right]\\
        &-\mathbb{E}\left[u_{K_0}(Y,Z)r_0(Y,Z)\right]+O_P\left( \frac{\sqrt{\zeta_{K_0}}K_0^{1/4}}{n^{(2-\epsilon)/2}}\right),
	\end{aligned}
\end{equation}
and
\begin{equation}\label{uproj_v}
		\begin{aligned}
			&\frac{1}{n(n-1)}\sum_{i=1}^{n}\sum_{j=1,j\neq i}^{n}v_{K}(X_i,Y_j,Z_i)\\
			=&\frac{1}{n}\sum_{t=1}^{n}\mathbb{E}[v_K(X_t,Y_t,Z_t)r_0(Y_t,Z_t)\pi_0(X_t,Y_t,Z_t)|X_t,Z_t]
            \\
            &+\frac{1}{n}\sum_{t=1}^{n}\mathbb{E}[v_K(X_t,Y_t,Z_t)r_0(Y_t,Z_t)\pi_0(X_t,Y_t,Z_t)|Y_t]\\
			&-\mathbb{E}[v_K(X,Y,Z)r_0(Y,Z)\pi_0(X,Y,Z)]
			+O_P\left( \frac{\sqrt{\xi_K}K^{1/4}}{n^{(2-\epsilon)/2}}\right).
		\end{aligned}
	\end{equation}
 \end{lemma}
\begin{proof}
  Consider the following second order $U$-statistic:
    \begin{equation*}
        U_{n,K_0}:=\frac{1}{n(n-1)}\sum_{1\leq i<j\leq n}\left\{u_{K_0}(Y_i,Z_j)+u_{K_0}(Y_j,Z_i)
        \right\}.
    \end{equation*}
 Define $$
        s_{\delta,K_0}=\sup_{1\leq i<j\leq n}\left\{\mathbb{E}\left[\left \Vert\frac{1}{2}\left\{ u_{K_0}(Y_i,Z_j)+u_{K_0}(Y_j,Z_i)\right\}\right\Vert ^{2+\delta}\right]\right\}^{1/(2+\delta)} .
    $$
 Under Assumption \ref{ass:cbounded},  since $
    \mathbb{E}\left[\Vert u_{K_0}(Y,Z)\Vert^{2+\delta}\right]\leq\zeta_{K_0}^\delta\mathbb{E}\left[\Vert u_{K_0}(Y,Z)\Vert^2\right]=O\left(\zeta_{K_0}^\delta K_0\right),$
    we have $s_{\delta,K_0}=O\left(\zeta_{K_0}^{\delta/(2+\delta)}K^{1/(2+\delta)}\right)$.  Under Assumption \ref{ass:ciid}, for any $\delta>0$ and $0<\epsilon<1/4$, we have $$\beta_n^{\delta/(2+\delta)}=O\left(\rho^{\frac{n\delta}{2+\delta}}\right)=O\left(n^{-2+\epsilon}\right).$$ Then we can apply (\ref{Un_decomposition}) and Lemma \ref{lemma:denker1983u}
with  $\delta=2$ to obtain
\begin{equation*}
	\begin{aligned} 
		U_{n,K_0}
		=&\frac{1}{n}\sum_{t=1}^{n}	\int_{\mathcal{Z}}u_{K_0}(Y_t,z)f_Z(z)dz+\frac{1}{n}\sum_{t=1}^{n}\int_{\mathcal{Y}}u_{K_0}(y,Z_t)f_Y(y)dy\\
        &-\mathbb{E}\left[u_{K_0}(Y,Z)r_0(Y,Z)\right]+O_P\left( \frac{s_{\delta,K_0}}{n^{(2-\epsilon)/2}}\right)\\
  =&\frac{1}{n}\sum_{t=1}^{n}	\mathbb{E}\left[u_{K_0}(Y_t,Z_t)r_0(Y_t,Z_t)|Y_t\right]+\frac{1}{n}\sum_{t=1}^{n}	\mathbb{E}\left[u_{K_0}(Y_t,Z_t)r_0(Y_t,Z_t)|Z_t\right]\\
  &-\mathbb{E}\left[u_{K_0}(Y,Z)r_0(Y,Z)\right]+O_P\left( \frac{\sqrt{\zeta_{K_0}}K_0^{1/4}}{n^{(2-\epsilon)/2}}\right),
	\end{aligned}
\end{equation*}
which gives the result \eqref{uproj_u}.

Similarly, for establishing (\ref{uproj_v}), we define
$$
        s'_{\delta,K}=\sup_{1\leq i<j\leq n}\left\{\mathbb{E}\left[\left \Vert\frac{1}{2}\left\{ v_{K}(X_i,Y_j,Z_i)+v_{K}(X_j,Y_i,Z_j)\right\}\right\Vert ^{2+\delta}\right]\right\}^{1/(2+\delta)}.
$$
Again, by applying (\ref{Un_decomposition}) and Lemma \ref{lemma:denker1983u} with $\delta=2$, we  have
\begin{equation*}
		\begin{aligned}
			&\frac{1}{n(n-1)}\sum_{i=1}^{n}\sum_{j=1,j\neq i}^{n}v_{K}(X_i,Y_j,Z_i)\\
   =&\frac{1}{n}\sum_{t=1}^{n}\left\{\int v_K(X_t,y,Z_t)f_Y(y)dy
+\int v_K(x,Y_t,z)f_{X,Z}(x,z)dxdz\right\}\\
  & -\mathbb{E}[v_K(X,Y,Z)r_0(Y,Z)\pi_0(X,Y,Z)]
			+O_P\left(\frac{s'_{\delta,K}}{n^{(2-\epsilon)/2}}\right)\\
			=&\frac{1}{n}\sum_{t=1}^{n}\mathbb{E}[v_K(X_t,Y_t,Z_t)r_0(Y_t,Z_t)\pi_0(X_t,Y_t,Z_t)|X_t,Z_t]\\
            &+\frac{1}{n}\sum_{t=1}^{n}\mathbb{E}[v_K(X_t,Y_t,Z_t)r_0(Y_t,Z_t)\pi_0(X_t,Y_t,Z_t)|Y_t]\\
			&-\mathbb{E}[v_K(X,Y,Z)r_0(Y,Z)\pi_0(X,Y,Z)]
			+O_P\left( \frac{\sqrt{\xi_K}K^{1/4}}{n^{(2-\epsilon)/2}}\right),
		\end{aligned}
	\end{equation*}
where the second equality holds by noting $r_0(y,z)\pi_0(x,y,z) = f_{X,Z}(x,z)f_Y(y)/f_{X,Y,Z}(x,y,z)$.
\end{proof}

For any square-integrable function $\phi(x,y,z)$, we define  
	\begin{align}
		\text{proj}^{wls}_{v_K}\phi(X,Y,Z)
		:=&\mathbb{E}\left[r_0(Y,Z)\phi(X,Y,Z)v_K^\top(X,Y,Z)\right]\notag \\
  &\times \mathbb{E}\left[r_0(Y,Z)v_K(X,Y,Z)v_K^\top(X,Y,Z)\right]^{-1}v_K(X,Y,Z),\label{def:rojphi}
	\end{align}
which is the weighted least-square projection (w.r.t. the norm $L^2(r_0(y,z)dF_{X,Y,Z}(x,y,z))$) of $\phi(x,y,z)$  on the space linearly spanned by $v_K(x,y,z)$. Define
	\begin{equation}
			\begin{aligned}\label{def:phiK}
				\phi_K(X_t,Y_t,Z_t)
				:=&\mathbb{E}\left[r_0(Y_t,Z_t)\pi_{0}(X_t,Y_t,Z_t)\text{proj}^{wls}_{v_K}\phi(X_t,Y_t,Z_t)|X_t,Z_t\right]\\
                &+\mathbb{E}\left[r_0(Y_t,Z_t)\pi_{0}(X_t,Y_t,Z_t)\text{proj}^{wls}_{v_K}\phi(X_t,Y_t,Z_t)|Y_t\right]\\	&-\mathbb{E}\left[r_0(Y_t,Z_t)\pi_K^*(X_t,Y_t,Z_t)\text{proj}^{wls}_{v_K}\phi(X_t,Y_t,Z_t)|Y_t\right]\\
                &-\mathbb{E}\left[r_0(Y_t,Z_t)\pi_K^*(X_t,Y_t,Z_t)\text{proj}^{wls}_{v_K}\phi(X_t,Y_t,Z_t)|Z_t\right]\\
				&+\mathbb{E}[r_0(Y,Z)\pi_K^*(X,Y,Z)\text{proj}^{wls}_{v_K}\phi(X,Y,Z)]\\
			&-r_0(Y_t,Z_t)\pi_K^*(X_t,Y_t,Z_t)\text{proj}^{wls}_{v_K}\phi(X_t,Y_t,Z_t).		
			\end{aligned}
		\end{equation}
It is worth noting that $\mathbb{E}[\phi_K(X_t,Y_t,Z_t)]=0$ since $\mathbb{E}[r_0(Y_t,Z_t)\pi_{0}(X_t,Y_t,Z_t)\text{proj}^{wls}_{v_K}\phi(X_t,Y_t,Z_t)]\\=\mathbb{E}[r_0(Y_t,Z_t)\pi_K^*(X_t,Y_t,Z_t)\text{proj}^{wls}_{v_K}\phi(X_t,Y_t,Z_t)]$ using the definition of $\text{proj}^{wls}_{v_K}\phi(X,Y,Z)$ and (\ref{proj}).
		\begin{lemma}\label{lemma_inf}
			Under  Assumptions  \ref{ass:ciid}-\ref{ass:proj_error},  suppose $\phi(X,Y,Z)$ is a function satisfying $\mathbb{E}\left[|\phi(X,Y,Z)|^{4+\eta}\right]\\<\infty$ for $\eta$ defined in Assumption \ref{ass:eigen}(iii). For  any fixed $K$ and any $0<\epsilon<1/4$,  we have
           \begin{small}
			\begin{equation*}
				\begin{aligned}
					&\frac{1}{n}\sum_{t=1}^{n}\left\{\widehat\pi_{K,K_0}(X_t, Y_t, Z_t)-\pi_K^*(X_t, Y_t, Z_t)\right\}r_0\left(Y_t,Z_t \right) \phi(X_t, Y_t, Z_t)\\
					=&\frac{1}{n}\sum_{t=1}^{n}\phi_K(X_t, Y_t, Z_t)	+O_P\left(\sqrt{\frac{K}{n}}\left\{\sqrt{\frac{K_0}{n}}+K_0^{-\omega_r}+\sqrt{\frac{K}{n}}\right\} \left\{\mathbb{E}\left[|\phi(X,Y,Z)|^{4+\eta}\right]\right\}^{\frac{1}{4+\eta}}\right)\\
                    &+O_P\left(\left\{K_0^{-\omega_{r}}+\frac{\sqrt{\xi_K}K^{1/4}}{n^{(2-\epsilon)/2}}+\frac{\sqrt{\zeta_{K_0}}K_0^{3/4}}{n^{(2-\epsilon)/2}}+\frac{\xi_{K}K_0}{n}+ \frac{\xi_{K}K}{n}+\sqrt{K}K_0^{-\omega_0}\right\}\cdot \sqrt{\mathbb{E}[|\phi(X,Y,Z)|^2]}\right)
				\end{aligned}
			\end{equation*}
            \end{small}
		where the constant included in $O_P(\cdot)$  does not depend on $K$.
		\end{lemma}
		\begin{proof}
			 Under Assumptions \ref{ass:ciid}, \ref {ass:cbounded} and \ref{ass:eigen}, using  Lemma \ref{Davydov}, we have
				\begin{align*}					&\mathbb{E}\left[\left\Vert\frac{1}{n}\sum_{t=1}^{n}\left\{v_K(X_t, Y_t, Z_t)r_0\left(Y_t,Z_t \right)\phi(X_t, Y_t, Z_t)-\mathbb{E} \left[v_K(X,Y,Z)r_0(Y,Z)\phi(X,Y,Z)\right]\right\}\right\Vert^2\right]\\
     =&\frac{1}{n}\mathbb{E}\left[\left\Vert v_K(X, Y, Z)r_0\left(Y,Z \right)\phi(X, Y, Z)-\mathbb{E}\left[v_K(X, Y, Z)r_0\left(Y,Z \right)\phi(X, Y, Z)\right]\right\Vert^2\right]\\
     &+\frac{2}{n^2}\sum_{k=1}^{K}\sum_{t=1}^{n}\sum_{\tau=1}^{n-t}Cov\big(v_{K,k}(X_t,Y_t,Z_t)r_0(Y_t,Z_t)\phi(X_t,Y_t,Z_t), \\
     &\qquad\qquad\qquad \qquad\qquad v_{K,k}(X_{t+\tau},Y_{t+\tau},Z_{t+\tau})r_0(Y_{t+\tau},Z_{t+\tau})\phi(X_{t+\tau},Y_{t+\tau},Z_{t+\tau})\big)\\
     \leq&\frac{4(4+\eta)}{ \eta n^2 }\sum_{k=1}^{K}\sum_{t=1}^{n}\sum_{\tau =1}^{n-t}\left\{\mathbb{E}\left[| v_{K,k}(X_t,Y_t,Z_t)r_0(Y_t,Z_t)\phi(X_t,Y_t,Z_t)|^{2+\frac{\eta}{2}}\right]\right\}^{\frac{4}{4+\eta}}\beta_{\tau}^{\frac{\eta}{4+\eta}}\\
     &+O\left(\frac{\xi_K^2}{n}\cdot \mathbb{E}\left[|\phi(X,Y,Z)|^2\right] \right)
     \\ 
     \leq&\frac{4(4+\eta)}{ \eta n^2 }\cdot\sup_{(y,z)\in\mathcal{Y}\times\mathcal{Z}}r_0^2(y,z)\cdot\sum_{k=1}^{K}\sum_{t=1}^{n}\sum_{\tau =1}^{n-t}\bigg\{\left\{\mathbb{E}\left[| v_{K,k}(X_t,Y_t,Z_t)|^{4+\eta}\right]\right\}^{\frac{2}{4+\eta}}\\
     &\left\{\mathbb{E}\left[|\phi(X_t,Y_t,Z_t)|^{4+\eta}\right]\right\}^{\frac{2}{4+\eta}}\beta_{\tau}^{\frac{\eta}{4+\eta}}\bigg \}+O\left(\frac{\xi_K^2}{n}\cdot \mathbb{E}\left[|\phi(X,Y,Z)|^2\right] \right)\\
     \leq&\frac{4(4+\eta)C^{\frac{2}{4+\eta}}}{ \eta n^2 }\cdot\sup_{(y,z)\in\mathcal{Y}\times\mathcal{Z}}r_0^2(y,z)\cdot\sum_{k=1}^K \left\{\mathbb{E}\left[|\phi(X_t,Y_t,Z_t)|^{4+\eta}\right]\right\}^{\frac{2}{4+\eta}}\sum_{t=1}^{n}\sum_{\tau =1}^{n-t}\beta_{\tau}^{\frac{\eta}{4+\eta}}\\
     &+O\left(\frac{\xi_K^2}{n}\cdot \mathbb{E}\left[|\phi(X,Y,Z)|^2\right] \right)\\
     =&O\left(\frac{K}{n}\left\{\mathbb{E}\left[|\phi(X_t,Y_t,Z_t)|^{4+\eta}\right]\right\}^{\frac{2}{4+\eta}}\right)+O\left(\frac{\xi_K^2}{n}\cdot \mathbb{E}\left[|\phi(X,Y,Z)|^2\right] \right),
				\end{align*}
    where $C$ is a universal constant arising from Assumption \ref{ass:eigen}. With this result, we have        
				\begin{align}
					&\frac{1}{n}\sum_{t=1}^{n}\left\{\widehat\pi_{K,K_0}(X_t, Y_t, Z_t)-\pi_K^*(X_t, Y_t, Z_t)\right\}r_0\left(Y_t,Z_t \right) \phi(X_t, Y_t, Z_t)\notag\\
					=&\frac{1}{n}\sum_{t=1}^{n}\bigg\{(\widehat\pi_{K,K_0}(X_t, Y_t, Z_t)-\pi_K^*(X_t, Y_t, Z_t))r_0\left(Y_t,Z_t \right)\phi(X_t, Y_t, Z_t)\notag\\
					&\qquad-\int(\widehat\pi_{K,K_0}(x,y,z)-\pi_K^*(x,y,z))r_0\left(y,z \right)\phi(x,y,z)dF_{X,Y,Z}(x,y,z)\bigg\}\notag\\
					&+\int(\widehat\pi_{K,K_0}(x,y,z)-\pi_K^*(x,y,z))r_0(y,z)\phi(x,y,z)dF_{X,Y,Z}(x,y,z)\notag\\
					=&\frac{1}{n}\sum_{t=1}^{n}\left\{v_K^\top(X_t, Y_t, Z_t)r_0\left(Y_t,Z_t \right)\phi(X_t, Y_t, Z_t)-\mathbb{E}\left[ v_K^\top(X,Y,Z)r_0(Y,Z)\phi(X,Y,Z)\right]\right\}\notag\\
     &\times \left( \widehat{\boldsymbol{\beta}}_{K,K_0}-\boldsymbol{\beta}^*_K\right) \notag\\
					&+\int(\widehat\pi_{K,K_0}(x,y,z)-\pi_K^*(x,y,z))r_0(y,z)\phi(x,y,z)dF_{X,Y,Z}(x,y,z)\notag\\
					\leq&\left\Vert\frac{1}{n}\sum_{t=1}^{n}\left\{v_K(X_t, Y_t, Z_t)r_0\left(Y_t,Z_t \right)\phi(X_t, Y_t, Z_t)-\mathbb{E}\left[ v_K(X,Y,Z)r_0(Y,Z)\phi(X,Y,Z)\right]\right\}\right\Vert\notag\\
					&\cdot\left\Vert\widehat{\boldsymbol{\beta}}_{K,K_0}-\boldsymbol{\beta}^*_K\right\Vert+\int(\widehat\pi_{K,K_0}(x,y,z)-\pi_K^*(x,y,z))r_0(y,z)\phi(x,y,z)dF_{X,Y,Z}(x,y,z)\notag\\
					=&\int(\widehat\pi_{K,K_0}(x,y,z)-\pi_K^*(x,y,z))r_0(y,z)\phi(x,y,z)dF_{X,Y,Z}(x,y,z)\notag\\
     &+O_P\left(\sqrt{\frac{K}{n}}\left\{\sqrt{\frac{K_0}{n}}+K_0^{-\omega_r}+\sqrt{\frac{K}{n}}\right\} \left\{\mathbb{E}\left[|\phi(X_t,Y_t,Z_t)|^{4+\eta}\right]\right\}^{\frac{1}{4+\eta}}\right)\notag\\
     &+O_P\left(\frac{\xi_K}{\sqrt{n}} \left\{\sqrt{\frac{K_0}{n}}+K_0^{-\omega_r}+\sqrt{\frac{K}{n}}\right\} \sqrt{\mathbb{E}\left[|\phi(X,Y,Z)|^2\right]} \right), \label{Q2n1}
				\end{align}
			where the last equality comes from (\ref{prop_beta}). 
            Using (\ref{betahat-*_H1n}), for any $0<\epsilon<1/4$,  we have
            \begin{small}
				\begin{align*}
					&\int(\widehat\pi_{K,K_0}(x,y,z)-\pi_K^*(x,y,z))r_0(y,z)\phi(x,y,z)dF_{X,Y,Z}(x,y,z)\\
					=&\int v_K^\top(x,y,z)(\widehat{\boldsymbol{\beta}}_{K,K_0}-\boldsymbol{\beta}^*_K)r_0(y,z)\phi(x,y,z)dF_{X,Y,Z}(x,y,z)\\
					=&\int v_K^\top(x,y,z)H_K^{-1}\cdot\frac{1}{n}\sum_{t=1}^{n}\bigg\{\mathbb{E}\left[v_K(X_t,Y_t,Z_t)r_0(Y_t,Z_t)\pi_{0}(X_t,Y_t,Z_t)|X_t,Z_t\right]\\
					&+\mathbb{E}\left[v_K(X_t,Y_t,Z_t)r_0(Y_t,Z_t)\pi_{0}(X_t,Y_t,Z_t)|Y_t\right]-\mathbb{E}\left[v_K(X_t,Y_t,Z_t)r_0(Y_t,Z_t)\pi_K^*(X_t,Y_t,Z_t)|Y_t\right]\\
					&-\mathbb{E}\left[v_K(X_t,Y_t,Z_t)r_0(Y_t,Z_t)\pi_K^*(X_t,Y_t,Z_t)|Z_t\right]+\mathbb{E}[v_{K}(X,Y,Z)r_0(Y,Z)\pi_K^*(X,Y,Z)]\\
					&
					-v_{K}(X_t,Y_t,Z_t)r_0(Y_t,Z_t)\pi_K^*(X_t,Y_t,Z_t)\bigg\}	r_0(y,z)\phi(x,y,z)dF_{X,Y,Z}(x,y,z)\\
					&+O_P\left(\left\{K_0^{-\omega_{r}}+\frac{\sqrt{\xi_K}K^{1/4}}{n^{(2-\epsilon)/2}}+\frac{\sqrt{\zeta_{K_0}}K_0^{3/4}}{n^{(2-\epsilon)/2}}+\frac{\xi_{K}K_0}{n}+ \frac{\xi_{K}K}{n}+\sqrt{K}K_0^{-\omega_0}\right\}\cdot \sqrt{\mathbb{E}[|\phi(X,Y,Z)|^2]}\right)\\
     =&\frac{1}{n}\sum_{t=1}^{n}\bigg\{\mathbb{E}\left[\mathbb{E}\left[r_0(Y,Z)v_K^\top(X,Y,Z)\phi(X,Y,Z)\right]H_K^{-1}v_K(X_t,Y_t,Z_t)r_0(Y_t,Z_t)\pi_{0}(X_t,Y_t,Z_t)|X_t,Z_t\right]\\	&+\mathbb{E}\left[\mathbb{E}\left[r_0(Y,Z)v_K^\top(X,Y,Z)\phi(X,Y,Z)\right]H_K^{-1}v_K(X_t,Y_t,Z_t)r_0(Y_t,Z_t)\pi_{0}(X_t,Y_t,Z_t)|Y_t\right]\\
     &-\mathbb{E}\left[\mathbb{E}\left[r_0(Y,Z)v_K^\top(X,Y,Z)\phi(X,Y,Z)\right]H_K^{-1}v_K(X_t,Y_t,Z_t)r_0(Y_t,Z_t)\pi_K^*(X_t,Y_t,Z_t)|Y_t\right]\\
					&-\mathbb{E}\left[\mathbb{E}\left[r_0(Y,Z)v_K^\top(X,Y,Z)\phi(X,Y,Z)\right]H_K^{-1}v_K(X_t,Y_t,Z_t)r_0(Y_t,Z_t)\pi_K^*(X_t,Y_t,Z_t)|Z_t\right]\\&+\mathbb{E}[\mathbb{E}\left[r_0(Y,Z)v_K^\top(X,Y,Z)\phi(X,Y,Z)\right]H_K^{-1}v_{K}(X,Y,Z)r_0(Y,Z)\pi_K^*(X,Y,Z)]\\
     &-\mathbb{E}\left[r_0(Y,Z)v_K^\top(X,Y,Z)\phi(X,Y,Z)\right]H_K^{-1}v_{K}(X_t,Y_t,Z_t)r_0(Y_t,Z_t)\pi_K^*(X_t,Y_t,Z_t)\bigg\}	\\
					&+O_P\left(\left\{K_0^{-\omega_{r}}+\frac{\sqrt{\xi_K}K^{1/4}}{n^{(2-\epsilon)/2}}+\frac{\sqrt{\zeta_{K_0}}K_0^{3/4}}{n^{(2-\epsilon)/2}}+\frac{\xi_{K}K_0}{n}+ \frac{\xi_{K}K}{n}+\sqrt{K}K_0^{-\omega_0}\right\}\cdot \sqrt{\mathbb{E}[|\phi(X,Y,Z)|^2]}\right)\\
   =&\frac{1}{n}\sum_{t=1}^{n}\phi_K(X_t, Y_t, Z_t)\\
			&	+O_P\left(\left\{K_0^{-\omega_{r}}+\frac{\sqrt{\xi_K}K^{1/4}}{n^{(2-\epsilon)/2}}+\frac{\sqrt{\zeta_{K_0}}K_0^{3/4}}{n^{(2-\epsilon)/2}}+\frac{\xi_{K}K_0}{n}+ \frac{\xi_{K}K}{n}+\sqrt{K}K_0^{-\omega_0}\right\}\cdot \sqrt{\mathbb{E}[|\phi(X,Y,Z)|^2]}\right),
    \label{Q2n2}
				\end{align*}	
                \end{small}
		where the last equality holds in light of \eqref{def:rojphi} and \eqref{def:phiK}. 
			Finally, combining this with \eqref{Q2n1}, we can conclude the desired result.
		\end{proof}
		\clearpage

\section{Proof of Lemma \ref{prop:vtilde}}	
		\label{app:lemma_vtilde}
		By the first order condition that $\widehat{H}_{K,K_0}\widehat{\boldsymbol{\beta}}_{K,K_0}=\widehat{h}_K$, we have
	\begin{equation*}
		\widehat{h}_K-\widehat{H}_{K,K_0}\boldsymbol{\beta}_K^*=\widehat{H}_{K,K_0}(\widehat{\boldsymbol{\beta}}_{K,K_0}-\boldsymbol{\beta}^*_K).
	\end{equation*}
	Denote
	\begin{align*}
		L.H.S. = 	\frac{1}{n(n-1)}\sum_{i=1}^{n}\sum_{j=1,j\neq i}^{n}v_{K}(X_i,Y_j,Z_i)-\frac{1}{n}\sum_{t=1}^{n}v_{K}(X_t, Y_t, Z_t)\widehat r_{K_0}(Y_t,Z_t)\pi_K^*(X_t, Y_t, Z_t),
	\end{align*}
	\begin{align*}
		R.H.S. = \frac{1}{n}\sum_{t=1}^{n}v_{K}(X_t, Y_t, Z_t)v_{K}^\top(X_t, Y_t, Z_t)\widehat r_{K_0}(Y_t,Z_t)(\widehat{\boldsymbol{\beta}}_{K,K_0}-\boldsymbol{\beta}^*_K).
	\end{align*}
	Then we have $L.H.S.=R.H.S$.  First, considering R.H.S., 
		\begin{align*}
			R.H.S.=&\frac{1}{n}\sum_{t=1}^{n}v_{K}(X_t, Y_t, Z_t)v_{K}^\top(X_t, Y_t, Z_t)(\widehat r_{K_0}(Y_t,Z_t)-r_0(Y_t,Z_t))(\widehat{\boldsymbol{\beta}}_{K,K_0}-\boldsymbol{\beta}^*_K)\\
			&+\frac{1}{n}\sum_{t=1}^{n}v_{K}(X_t, Y_t, Z_t)v_{K}^\top(X_t, Y_t, Z_t)r_{0}(Y_t,Z_t)(\widehat{\boldsymbol{\beta}}_{K,K_0}-\boldsymbol{\beta}^*_K)\\
			=&\frac{1}{n}\sum_{t=1}^{n}v_{K}(X_t, Y_t, Z_t)v_{K}^\top(X_t, Y_t, Z_t)r_{0}(Y_t,Z_t)(\widehat{\boldsymbol{\beta}}_{K,K_0}-\boldsymbol{\beta}^*_K)\\
			&+O_P\left(\xi_{K} \left\{K_0^{-\omega_r}+\sqrt{\frac{K_0}{n}}\right\}\left\{\sqrt{\frac{K_0}{n}}+K_0^{-\omega_r}+\sqrt{ \frac{K}{n}}\right\}\right),\\
			=& \mathbb{E}\left[r_0(Y,Z)v_K(X,Y,Z)v_K^\top(X,Y,Z)\right]\left( \widehat{\boldsymbol{\beta}}_{K,K_0}-\boldsymbol{\beta}^*_K\right)\\&+O_P\left( \xi_K\sqrt{\frac{K}{n}}\left\{\sqrt{\frac{K_0}{n}}+K_0^{-\omega_{r}}+\sqrt{\frac{K}{n}}\right\}\right) \\
			&+O_P\left(\xi_{K} \left\{K_0^{-\omega_r}+\sqrt{\frac{K_0}{n}}\right\}\left\{\sqrt{\frac{K_0}{n}}+K_0^{-\omega_r}+\sqrt{ \frac{K}{n}}\right\}\right) \\
			=&\mathbb{E}\left[r_0(Y,Z)v_K(X,Y,Z)v_K^\top(X,Y,Z)\right]\left( \widehat{\boldsymbol{\beta}}_{K,K_0}-\boldsymbol{\beta}^*_K\right) \\
            &+O_P\left(\xi_{K} \left\{K_0^{-\omega_r}+\sqrt{\frac{K_0}{n}}+\sqrt{\frac{K}{n}}\right\}^2\right)\\
=&\mathbb{E}\left[r_0(Y,Z)v_K(X,Y,Z)v_K^\top(X,Y,Z)\right]\left( \widehat{\boldsymbol{\beta}}_{K,K_0}-\boldsymbol{\beta}^*_K\right) +O_P\left( \frac{\xi_{K}K_0}{n}+ \frac{\xi_{K}K}{n}\right)
		\end{align*}\label{RHS}
	where the second equality comes from (\ref{vkphi}) and
	\begin{equation*}
		\begin{aligned}
			&\left\Vert\frac{1}{n}\sum_{t=1}^{n}v_{K}(X_t, Y_t, Z_t)v_{K}^\top(X_t, Y_t, Z_t)(\widehat r_{K_0}(Y_t,Z_t)-r_0(Y_t,Z_t))(\widehat{\boldsymbol{\beta}}_{K,K_0}-\boldsymbol{\beta}^*_K)\right\Vert^2\\
			\leq&O_P(1)\cdot\frac{1}{n}\sum_{t=1}^{n}\left| v_{K}^\top(X_t, Y_t, Z_t)(\widehat r_{K_0}(Y_t,Z_t)-r_0(Y_t,Z_t))(\widehat{\boldsymbol{\beta}}_{K,K_0}-\boldsymbol{\beta}^*_K)\right|^2\\
			\leq&O_P(1)\cdot\sup_{(x,y,z)\in\mathcal{X}\times\mathcal{Y}\times \mathcal{Z}}\left|\widehat \pi_{K,K_0}(x,y,z)-\pi_K^*(x,y,z)\right|^2\cdot\frac{1}{n}\sum_{t=1}^{n}\left|\widehat r_{K_0}(Y_t,Z_t)-r_0(Y_t,Z_t)\right|^2\\
			=&O_P\left(\xi_{K}^2\left\{\frac{K_0}{n}+K_0^{- 2\omega_r}+\frac{K}{n}\right\}\left\{K_0^{-2\omega_{r}}+\frac{K_0}{n}\right\} \right),
		\end{aligned}
	\end{equation*}
	and the third equality comes from Assumption \ref{ass:cbounded} (i), (\ref{eq:v_eigen}) and (\ref{prop_beta}). Note that
	\begin{align}
		L.H.S. = &\frac{1}{n(n-1)}\sum_{i=1}^{n}\sum_{j=1,j\neq i}^{n}v_{K}(X_i,Y_j,Z_i)-\frac{1}{n}\sum_{t=1}^{n}r_{0}(Y_t,Z_t)v_K(X_t, Y_t, Z_t)\pi_K^*(X_t, Y_t, Z_t)\label{lhs1}\\
		&-\frac{1}{n}\sum_{t=1}^{n}\left( \widehat r_{K_0}(Y_t,Z_t)-r_{K_0}^*(Y_t,Z_t)\right)v_K(X_t, Y_t, Z_t)\pi_{K}^*(X_t, Y_t, Z_t)\label{lhs2}\\
		&-\frac{1}{n}\sum_{t=1}^{n}\left(  r_{K_0}^*(Y_t,Z_t)-r_0(Y_t,Z_t)\right)v_K(X_t, Y_t, Z_t)\pi_{K}^*(X_t, Y_t, Z_t) \label{lhs3}
	\end{align}
For $(\ref{lhs1})$, using Lemma  \ref{uproj},  for any $0<\epsilon <1/4$,  we have
		\begin{align*}
			(\ref{lhs1})	=&\frac{1}{n}\sum_{t=1}^{n}\mathbb{E}[v_K(X_t,Y_t,Z_t)r_0(Y_t,Z_t)\pi_0(X_t,Y_t,Z_t)|X_t,Z_t]\\
            &+\frac{1}{n}\sum_{t=1}^{n}\mathbb{E}[v_K(X_t,Y_t,Z_t)r_0(Y_t,Z_t)\pi_0(X_t,Y_t,Z_t)|Y_t]\\
            &-\mathbb{E}[v_{K}(X,Y,Z)r_0(Y,Z)\pi_0(X,Y,Z)]\\
            &-\frac{1}{n}\sum_{t=1}^{n}v_{K}(X_t,Y_t,Z_t)r_0(Y_t,Z_t)\pi_K^*(X_t,Y_t,Z_t)	+O_P\left( \frac{\sqrt{\xi_K}K^{1/4}}{n^{(2-\epsilon)/2}}\right).
		\end{align*}
For $(\ref{lhs2})$, we have
\begin{equation}
\label{lhs2_1}
	\begin{aligned}
		&\frac{1}{n}\sum_{t=1}^{n}\left( \widehat r_{K_0}(Y_t,Z_t)-r_{K_0}^*(Y_t,Z_t)\right)v_K(X_t, Y_t, Z_t)\pi_{K}^*(X_t, Y_t, Z_t)\\
		=&\frac{1}{n}\sum_{t=1}^{n}v_K(X_t, Y_t, Z_t)\pi_{K}^*(X_t, Y_t, Z_t)u_{K_0}^\top(Y_t,Z_t)\left( 	\widehat \gamma_{K_0}-\gamma_{K_0}^*\right) \\
		=&\mathbb{E}\left[v_K(X,Y,Z)u_{K_0}^\top(Y,Z)\pi_K^*(X,Y,Z)\right]\left( 	\widehat \gamma_{K_0}-\gamma_{K_0}^*\right) +O_P\left(\frac{\xi_KK_0}{n} \right). 
	\end{aligned}
\end{equation}
	where the second equality comes from  \eqref{prop_gamma}  and
\begin{equation*}
    \begin{aligned}
    &\left\Vert\frac{1}{n}\sum_{t=1}^{n}v_{K}(X_t, Y_t, Z_t)u_{K_0}^\top(Y_t,Z_t)\pi_K^*(X_t, Y_t, Z_t)-\mathbb{E}\left[v_K(X,Y,Z)u_{K_0}^\top(Y,Z)\pi_K^*(X,Y,Z)\right]\right\Vert
\\=&O_P\left(\xi_K\sqrt{\frac{K_0}{n}}\right),   
    \end{aligned}
\end{equation*}
which is similar with \eqref{eq:vu_eigen}. 
Since $\widehat{\boldsymbol{\gamma}}_{K_0}=	\widehat{\Sigma}_{K_0}^{-1}\widehat{b}_{K_0}
$ and $\gamma_{K_0}^*=b_{K_0}=\mathbb{E}\left[u_{K_0}(Y,Z)r_0(Y,Z)\right]$    under the normalization $\mathbb{E}[ u_{K_0}(Y,Z)u^\top_{K_0}(Y,Z)]=I_{K_0\times K_0}$, we have an explicit expression for $ \widehat \gamma_{K_0}-\gamma_{K_0}^*$:
\begin{align*}
	\widehat \gamma_{K_0}-\gamma_{K_0}^*=&\left[\frac{1}{n}\sum_{t=1}^{n}u_{K_0}(Y_t,Z_t)u_{K_0}^\top(Y_t,Z_t)\right]^{-1}\cdot\frac{1}{n(n-1)}\sum_{j=1,j\neq i}^n\sum_{i=1}^nu_{K_0}(Y_i,Z_j)\\
    &-\mathbb{E}\left[u_{K_0}(Y,Z)r_0(Y,Z)\right].
\end{align*}
Again, by Lemma  \ref{uproj}, we have 
	\begin{align*}
		&\frac{1}{n(n-1)}\sum_{j=1,j\neq i}^n\sum_{i=1}^nu_{K_0}(Y_i,Z_j)\\
		=&\frac{1}{n}\sum_{t=1}^{n}	\int_{\mathcal{Z}}u_{K_0}(Y_t,z)f_Z(z)dz+\frac{1}{n}\sum_{t=1}^{n}\int_{\mathcal{Y}}u_{K_0}(y,Z_t)f_Y(y)dy-\mathbb{E}\left[u_{K_0}(Y,Z)r_0(Y,Z)\right]\\&+O_P\left( \frac{\sqrt{\zeta_{K_0}}K_0^{1/4}}{n^{(2-\epsilon)/2}}\right),
	\end{align*}
Then, under the normalization \eqref{eq:orthnormal}, we can obtain
\begin{equation}
\begin{aligned}
	\widehat \gamma_{K_0}-\gamma_{K_0}^*=&\frac{1}{n}\sum_{t=1}^{n}\left\{	\int_{\mathcal{Z}}u_{K_0}(Y_t,z)f_Z(z)dz+\int_{\mathcal{Y}}u_{K_0}(y,Z_t)f_Y(y)dy\right\}-2\mathbb{E}\left[u_{K_0}(Y,Z)r_0(Y,Z)\right]\\&+O_P\left( \frac{\sqrt{\zeta_{K_0}}K_0^{1/4}}{n^{(2-\epsilon)/2}}\right).	\label{gamma}
    \end{aligned}
\end{equation}
Combining \eqref{lhs2_1} with \eqref{gamma}, under the normalization $\mathbb{E}[ u_{K_0}(Y,Z)u^\top_{K_0}(Y,Z)]=I_{K_0\times K_0}$, we have
\begin{small}
	\begin{equation}
	\begin{aligned}\label{rvpi}
		&\frac{1}{n}\sum_{t=1}^{n}\left( \widehat r_{K_0}(Y_t,Z_t)-r_{K_0}^*(Y_t,Z_t)\right)v_K(X_t, Y_t, Z_t)\pi_{K}^*(X_t, Y_t, Z_t)\\
			=&\mathbb{E}\left[v_K(X,Y,Z)u_{K_0}^\top(Y,Z)\pi_K^*(X,Y,Z)\right]\bigg\{
			\frac{1}{n}\sum_{t=1}^{n}\left(\int_{\mathcal{Z}}u_{K_0}(Y_t,z)f_Z(z)dz+\int_{\mathcal{Y}}u_{K_0}(y,Z_t)f_Y(y)dy\right)\\
			&\qquad-2\mathbb{E}\left[u_{K_0}(Y,Z)r_0(Y,Z)\right]\bigg\}+O_P\left( \frac{\sqrt{\zeta_{K_0}}K_0^{3/4}}{n^{(2-\epsilon)/2}}\right)+O_P\left(\frac{\xi_KK_0}{n} \right)\\
			=&\frac{1}{n}\sum_{t=1}^{n}	\int_{\mathcal{Z}}\mathbb{E}\left[v_K(X_t,Y_t,Z_t)\pi_K^*(X_t,Y_t,Z_t)|Y_t,Z_t=z\right]r_0(Y_t,z)f(z|Y_t)dz\\
			&+\frac{1}{n}\sum_{t=1}^{n}\int_{\mathcal{Y}}\mathbb{E}\left[v_K(X_t,Y_t,Z_t)\pi_K^*(X_t,Y_t,Z_t)|Y_t=y,Z_t\right]r_0(y,Z_t)f(y|Z_t)dy\\
			&-2\mathbb{E}\left[v_K(X,Y,Z)\pi_K^*(X,Y,Z)r_0(Y,Z)\right]
		+O_P\left( \frac{\sqrt{\zeta_{K_0}}K_0^{3/4}}{n^{(2-\epsilon)/2}}\right)+O_P\left(\frac{\xi_KK_0}{n} \right)+O_P\left(\sqrt{K}K_0^{-\omega_0}\right)\\
	=&\frac{1}{n}\sum_{t=1}^{n}\left\{\mathbb{E}\left[v_K(X_t,Y_t,Z_t)r_0(Y_t,Z_t)\pi_K^*(X_t,Y_t,Z_t)|Y_t\right]+\mathbb{E}\left[v_K(X_t,Y_t,Z_t)r_0(Y_t,Z_t)\pi_K^*(X_t,Y_t,Z_t)|Z_t\right]\right\}\\&	-2\mathbb{E}\left[v_K(X,Y,Z)\pi_K^*(X,Y,Z)r_0(Y,Z)\right]+O_P\left( \frac{\sqrt{\zeta_{K_0}}K_0^{3/4}}{n^{(2-\epsilon)/2}}\right)+O_P\left(\frac{\xi_KK_0}{n} \right)+O_P\left(\sqrt{K}K_0^{-\omega_0}\right)
	\end{aligned}
	\end{equation}
    \end{small}
where the second equality comes from Assumption \ref{ass:proj_error}.

	For $(\ref{lhs3})$, by (\ref{vkphi}) and Assumption \ref{ass:cpi0-pi*} (i),
	we have $(\ref{lhs3})=O_P\left( K_0^{-\omega_{r}}\right)$.
	Finally, we obtain
			\begin{align*}
					L.H.S.
				=&\frac{1}{n}\sum_{t=1}^{n}\bigg\{\mathbb{E}[v_K(X_t,Y_t,Z_t)r_0(Y_t,Z_t)\pi_0(X_t,Y_t,Z_t)|X_t,Z_t]\\&+\mathbb{E}[v_K(X_t,Y_t,Z_t)r_0(Y_t,Z_t)\pi_0(X_t,Y_t,Z_t)|Y_t]\\
				&-\mathbb{E}[v_{K}(X,Y,Z)r_0(Y,Z)\pi_0(X,Y,Z)]-v_{K}(X_t,Y_t,Z_t)r_0(Y_t,Z_t)\pi_K^*(X_t,Y_t,Z_t)\bigg\}\\
				&	-\frac{1}{n}\sum_{t=1}^{n}\bigg\{\mathbb{E}\left[v_K(X_t,Y_t,Z_t)r_0(Y_t,Z_t)\pi_K^*(X_t,Y_t,Z_t)|Y_t\right]\\
                &+\mathbb{E}\left[v_K(X_t,Y_t,Z_t)r_0(Y_t,Z_t)\pi_K^*(X_t,Y_t,Z_t)|Z_t\right]\bigg\}\\
		&+2\mathbb{E}\left[v_K(X,Y,Z)\pi_K^*(X,Y,Z)r_0(Y,Z)\right]
				\\&+O_P\left(K_0^{-\omega_{r}}  +\frac{\sqrt{\xi_K}K^{1/4}}{n^{(2-\epsilon)/2}}+\frac{\sqrt{\zeta_{K_0}}K_0^{3/4}}{n^{(2-\epsilon)/2}}+\frac{\xi_KK_0}{n} +\sqrt{K}K_0^{-\omega_0}\right),
			\end{align*}
	and from $L.H.S.=R.H.S.$, 
	\begin{equation}	\label{betahat-*_H1n}
		\begin{aligned}
				\widehat{\boldsymbol{\beta}}_{K,K_0}-\boldsymbol{\beta}^*_K
			=&H_K^{-1}\cdot\frac{1}{n}\sum_{t=1}^{n}\widetilde v_{K}(X_t,Y_t,Z_t)\\&+O_P\left(K_0^{-\omega_{r}}  +\frac{\sqrt{\xi_K}K^{1/4}}{n^{(2-\epsilon)/2}}+\frac{\sqrt{\zeta_{K_0}}K_0^{3/4}}{n^{(2-\epsilon)/2}}+\frac{\xi_{K}K_0}{n}+ \frac{\xi_{K}K}{n}+\sqrt{K}K_0^{-\omega_0}\right),
		\end{aligned}
	\end{equation}
where 	
\begin{align*}
	\widetilde v_{K}(X_t,Y_t,Z_t)=&\mathbb{E}\left[v_K(X_t,Y_t,Z_t)r_0(Y_t,Z_t)\pi_{0}(X_t,Y_t,Z_t)|X_t,Z_t\right]\\&+\mathbb{E}\left[v_K(X_t,Y_t,Z_t)r_0(Y_t,Z_t)\pi_{0}(X_t,Y_t,Z_t)|Y_t\right]\\	&-\mathbb{E}\left[v_K(X_t,Y_t,Z_t)r_0(Y_t,Z_t)\pi_K^*(X_t,Y_t,Z_t)|Y_t\right]\\&-\mathbb{E}\left[v_K(X_t,Y_t,Z_t)r_0(Y_t,Z_t)\pi_K^*(X_t,Y_t,Z_t)|Z_t\right]\\
	&+2\mathbb{E}[v_{K}(X,Y,Z)r_0(Y,Z)\pi_K^*(X,Y,Z)]-\mathbb{E}[v_{K}(X,Y,Z)r_0(Y,Z)\pi_0(X,Y,Z)]
	\\
	&-v_{K}(X_t,Y_t,Z_t)r_0(Y_t,Z_t)\pi_K^*(X_t,Y_t,Z_t).
\end{align*}
	Since $H_{K}\boldsymbol{\beta}_K^*=h_K$,  i.e., 
	$
\mathbb{E}\left[r_0(Y,Z)v_K(X,Y,Z)v_K^\top(X,Y,Z)\right]\boldsymbol{\beta}_K^*=\int \mathbb{E}[v_K(X,y,Z)]f_Y(y)dy,
	$
	we have
	\begin{align}\label{proj}
\mathbb{E}\left[r_0(Y,Z)v_K(X,Y,Z)\pi_K^*(X,Y,Z)\right]=\mathbb{E}\left[r_0(Y,Z)v_K(X,Y,Z)\pi_0(X,Y,Z)\right].
	\end{align}
 Then 
 \begin{equation*}
      \begin{aligned}
	\widetilde v_{K}(X_t,Y_t,Z_t)=&\mathbb{E}\left[v_K(X_t,Y_t,Z_t)r_0(Y_t,Z_t)\pi_{0}(X_t,Y_t,Z_t)|X_t,Z_t\right]\\&+\mathbb{E}\left[v_K(X_t,Y_t,Z_t)r_0(Y_t,Z_t)\pi_{0}(X_t,Y_t,Z_t)|Y_t\right]\\	&-\mathbb{E}\left[v_K(X_t,Y_t,Z_t)r_0(Y_t,Z_t)\pi_K^*(X_t,Y_t,Z_t)|Y_t\right]\\&-\mathbb{E}\left[v_K(X_t,Y_t,Z_t)r_0(Y_t,Z_t)\pi_K^*(X_t,Y_t,Z_t)|Z_t\right]\\
	&+\mathbb{E}[v_{K}(X,Y,Z)r_0(Y,Z)\pi_K^*(X,Y,Z)]\\&
	-v_{K}(X_t,Y_t,Z_t)r_0(Y_t,Z_t)\pi_K^*(X_t,Y_t,Z_t).
\end{aligned}
 \end{equation*}
Under Assumptions \ref{ass:cK}-\ref{ass:proj_error},  we have
			\begin{align*}
			\widehat{\boldsymbol{\beta}}_{K,K_0}-\boldsymbol{\beta}^*_K =&\frac{1}{n}\sum_{t=1}^{n}H_{K}^{-1}\widetilde v_{K}(X_t,Y_t,Z_t)+o_P\left(\frac{K^{1/4}}{\sqrt{n}}\right).
		\end{align*}	
		Especially, when $H_0$ holds, $\pi_0(x,y,z)=\pi_K^*(x,y,z)\equiv 1$ a.s.,
	\begin{equation*}
\begin{aligned}
	\widetilde v_{K}(X_t,Y_t,Z_t)
=&\mathbb{E}\left[v_K(X_t,Y_t,Z_t)r_0(Y_t,Z_t)|X_t,Z_t\right]-\mathbb{E}\left[v_K(X_t,Y_t,Z_t)r_0(Y_t,Z_t)|Z_t\right]\\
&+\mathbb{E}\left[v_K(X,Y,Z)r_0(Y,Z)\right]-v_{K}(X_t,Y_t,Z_t)r_0(Y_t,Z_t).
\end{aligned}
\end{equation*}

		\clearpage
	\section{Proof of Theorem \ref{th_crate}} \label{sec:th_crate}
	By Assumption \ref{ass:eigen}, without loss of generality, we can assume that the sieve bases $u_{K_0}(y,z)$ and $v_{K}(x,y,z)$ are orthonormal:
		\begin{align}\label{eq:orthnormal}
			\mathbb{E}\left[u_{K_0}(Y,Z)u_{K_0}^\top(Y,Z)\right]=I_{K_0\times K_0}\ \text{and} \ \mathbb{E}[v_{K}(X,Y,Z)v_{K}^\top(X,Y,Z)]=I_{K\times K}.
		\end{align}
	We begin to prove part (i). Note that
			\begin{align}
				&\Vert\widehat{\boldsymbol{\gamma}}_{K_0}-\boldsymbol{\gamma}^*_{K_0}\Vert=\Vert\widehat \Sigma_{K_0}^{-1}(\widehat b_{K_0}-\widehat \Sigma_{K_0}\boldsymbol{\gamma}^*_{K_0})\Vert \notag\\
				=& \left\{(\widehat b_{K_0}-\widehat \Sigma_{K_0}\boldsymbol{\gamma}^*_{K_0})^{\top }\widehat{\Sigma}_{K_0}^{-1} \widehat{\Sigma}_{K_0}^{-1}(\widehat b_{K_0}-\widehat \Sigma_{K_0}\boldsymbol{\gamma}^*_{K_0})\right\}^{1/2}\notag \\
				\leq &\lambda^{-1}_{\min}\left(\widehat{\Sigma}_{K_0}\right)\cdot \Vert\widehat b_{K_0}-\widehat \Sigma_{K_0}\boldsymbol{\gamma}^*_{K_0}\Vert,\label{eq:gammahat-gamma*}
			\end{align}
			where $\lambda_{\min}(\widehat{\Sigma}_{K_0})$ is the smallest eigenvalue of $\widehat{\Sigma}_{K_0}$. By Lemma \ref{lemma:proj} and Assumption \ref{ass:eigen}, we  have 
			\begin{align}\label{eq:Shat-S}
				\|\widehat{\Sigma}_{K_0}-\Sigma_{K_0}\|=o_P(1) \  \text{and} \ \lambda^{-1}_{\min}\left(\widehat \Sigma_{K_0}\right)=O_P(1).
			\end{align}   
			
			Next, we derive the convergence rate for $ \Vert\widehat b_{K_0}-\widehat \Sigma_{K_0}\boldsymbol{\gamma}^*_{K_0}\Vert$. As  Lemma  \ref{uproj} stated,  for any $0<\epsilon<1/4$,
	\begin{align*}
		&\frac{1}{n(n-1)}\sum_{j=1,j\neq i}^n\sum_{i=1}^nu_{K_0}(Y_i,Z_j)\\
		=&\frac{1}{n}\sum_{t=1}^{n}	\int_{\mathcal{Z}}u_{K_0}(Y_t,z)f_Z(z)dz+\frac{1}{n}\sum_{t=1}^{n}\int_{\mathcal{Y}}u_{K_0}(y,Z_t)f_Y(y)dy-\mathbb{E}\left[u_{K_0}(Y,Z)r_0(Y,Z)\right]\\
        &+O_P\left( \frac{\sqrt{\zeta_{K_0}}K_0^{1/4}}{n^{(2-\epsilon)/2}}\right),
	\end{align*}
   then we have
			\begin{align}
				&\widehat b_{K_0}-\widehat \Sigma_{K_0}\boldsymbol{\gamma}^*_{K_0}\notag\\
				=&\frac{1}{n(n-1)}\sum_{i=1,i\neq j}^{n}\sum_{j=1}^{n}u_{K_0}(Y_i,Z_j)- \frac{1}{n}\sum_{t=1}^{n} u_{K_0}(Y_t,Z_t)u_{K_0}^\top(Y_t,Z_t) \boldsymbol{\gamma}^*_{K_0} \notag\\
				=&\frac{1}{n}\sum_{t=1}^{n}\psi_{K_0}(Y_t)- \frac{1}{n}\sum_{t=1}^{n} u_{K_0}(Y_t,Z_t)u_{K_0}^\top(Y_t,Z_t) \boldsymbol{\gamma}^*_{K_0},\label{eq:bhat-Shatgamma_1}\\
    &+\frac{1}{n}\sum_{t=1}^{n}\int_{\mathcal{Y}}u_{K_0}(y,Z_t)f_Y(y)dy-\mathbb{E}\left[u_{K_0}(Y,Z)r_0(Y,Z)\right]+O_P\left( \frac{\sqrt{\zeta_{K_0}}K_0^{1/4}}{n^{(2-\epsilon)/2}}\right)\label{eq:bhat-Shatgamma_2}
			\end{align}	
			where $\psi_{K_0}(y):=\int u_{K_0}(y,z)f_Z(z)dz=\left( \psi_{K_0, 1}(y), ... , \psi_{K_0, K_0}(y)\right) ^\top$.  
   For the term (\ref{eq:bhat-Shatgamma_2}),  by  $0<\epsilon<1/4$  
   and Assumption \ref{ass:cK}(i),   \begin{small}
   \begin{align*}
       (\ref{eq:bhat-Shatgamma_2})=&\int_{\mathcal{Y}}\left\{\frac{1}{n}\sum_{t=1}^{n}u_{K_0}(y,Z_t)-\int_{\mathcal{Z}}u_{K_0}(y,z)f_Z(z)dz\right\}f_Y(y)dy+O_P\left( \frac{\sqrt{\zeta_{K_0}}K_0^{1/4}}{n^{(2-\epsilon)/2}}\right)=O_P\left(
 \sqrt{\frac{K_0}{n}}\right).      \notag
   \end{align*}
   \end{small}
			For the term (\ref{eq:bhat-Shatgamma_1}),   by Assumptions \ref{ass:ciid},  \ref{ass:eigen} and Lemma \ref{Davydov},  we have
					\begin{align*}
						&\mathbb{E}\left[\left\Vert(\ref{eq:bhat-Shatgamma_1} )\right\Vert^2\right]\\
						=&\mathbb{E}\left[\left\Vert\frac{1}{n}\sum_{t=1}^{n} \left\{u_{K_0}(Y_t,Z_t)u_{K_0}^\top(Y_t,Z_t) \boldsymbol{\gamma}^*_{K_0}-\psi_{K_0}(Y_t)\right\}\right\Vert^2\right]\\
						=&\frac{1}{n}\mathbb{E}\left[\left\Vert u_{K_0}(Y_t,Z_t)u_{K_0}^\top(Y_t,Z_t) \boldsymbol{\gamma}^*_{K_0}-\psi_{K_0}(Y_t)\right\Vert^2\right]\\
						&+\frac{2}{n^2}\sum_{t=1}^{n}\sum_{\tau=1}^{n-t}\mathbb{E}\left[\left\{u_{K_0}(Y_t,Z_t)u_{K_0}^\top(Y_t,Z_t) \boldsymbol{\gamma}^*_{K_0}-\psi_{K_0}(Y_t)\right\}^\top\right.\\
                        &\left.\left\{u_{K_0}(Y_{t+\tau},Z_{t+\tau})u_{K_0}^\top(Y_{t+\tau},Z_{t+\tau}) \boldsymbol{\gamma}^*_{K_0}-\psi_{K_0}(Y_{t+\tau})\right\}\right]\\
						\leq&\frac{2}{n^2}\sum_{k=1}^{K_0}\sum_{t=1}^{n}\sum_{\tau=1}^{n-t}Cov\left( u_{K_0,k}(Y_t,Z_t)u_{K_0}^\top(Y_t,Z_t) \boldsymbol{\gamma}^*_{K_0}-\psi_{K_0, k}(Y_t),\right.\\
                        &\left.\qquad\qquad\qquad \qquad u_{K_0, k}(Y_{t+\tau},Z_{t+\tau})u_{K_0}^\top(Y_{t+\tau},
                        Z_{t+\tau}) \boldsymbol{\gamma}^*_{K_0}-\psi_{K_0, k}(Y_{t+\tau})\right) \\
						&+\frac{2}{n}\mathbb{E}\left[\left\Vert u_{K_0}(Y,Z)u_{K_0}^\top(Y,Z) \boldsymbol{\gamma}^*_{K_0}\right\Vert^2\right]+\frac{2}{n}\mathbb{E}\left[\left\Vert \psi_{K_0}(Y)\right\Vert^2\right]\\
						\leq&\frac{8}{n^2}\sum_{k=1}^{K_0}\sum_{t=1}^{n}\sum_{\tau=1}^{n-t}\left\{\mathbb{E}\left[|u_{K_0,k}(Y_t,Z_t)u_{K_0}^\top(Y_t,Z_t) \boldsymbol{\gamma}^*_{K_0}-\psi_{K_0, k}(Y_t)|^4\right]\right\}^{1/2}\beta_{\tau}^{1/2}\\
						&+\frac{2}{n}\mathbb{E}\left[\left\Vert u_{K_0}(Y,Z)\right\Vert^2\right]\cdot\sup_{(y,z)\in\mathcal{Y}\times\mathcal{Z}}|u_{K_0}^\top(y,z) \boldsymbol{\gamma}^*_{K_0}|^2+\frac{2}{n}\int\left\Vert u_{K_0}(y,z)\right\Vert^2f_Y(y)f_Z(z)dydz\\
						\leq&\frac{1}{n^2}\sum_{k=1}^{K_0}\sum_{t=1}^{n}C+\frac{2}{n}\mathbb{E}\left[\left\Vert u_{K_0}(Y,Z)\right\Vert^2\right]\left\{\sup_{(y,z)\in\mathcal{Y}\times\mathcal{Z}}|r_0(y,z)|^2+\sup_{(y,z)\in\mathcal{Y}\times\mathcal{Z}}|r_{K_0}^*(y,z)-r_0(y,z)|^2\right\}\\
						&+\frac{2}{n}\sup_{(y,z)\in\mathcal{Y}\times\mathcal{Z}}\frac{f_Y(y)f_Z(z)}{f_{Y,Z}(y,z)}\cdot\mathbb{E}\left[\left\Vert u_{K_0}(Y,Z)\right\Vert^2\right]\\
						=&O\left(\frac{K_0}{n} \right) 
					\end{align*}
   where $C$ is a universal constant and the first inequality comes from $$\left\{u_{K_0}(Y_t,Z_t)u_{K_0}^\top(Y_t,Z_t) \boldsymbol{\gamma}^*_{K_0}-\psi_{K_0}(Y_t)\right\}_{t=1}^n$$ has mean 0 since
			\begin{equation*}
				\begin{aligned}
					&\mathbb{E}\left[u_{K_0}(Y,Z)u_{K_0}^\top(Y,Z) \boldsymbol{\gamma}^*_{K_0}\right]\\
					=&\mathbb{E}\left[u_{K_0}(Y,Z)u_{K_0}^\top(Y,Z) \right]\Sigma_{K_0}^{-1}b_{K_0}\\
					=&\mathbb{E}\left[u_{K_0}(Y,Z)u_{K_0}^\top(Y,Z) \right]\left\{\mathbb{E}\left[u_{K_0}(Y,Z)u_{K_0}^\top(Y,Z) \right]\right\}^{-1}\int\mathbb{E}\left[u_{K_0}(y,Z)\right]f_Y(y)dy\\
					=&\int\mathbb{E}\left[u_{K_0}(y,Z)\right]f_Y(y)dy\\
					=&\mathbb{E}\left[\psi_{K_0}(Y)\right].
				\end{aligned}
			\end{equation*}
			Then by Chebyshev's inequality we have
			\begin{equation*}
				(\ref{eq:bhat-Shatgamma_1})=O_P\left(\sqrt{ \frac{K_0}{n}}\right). 
			\end{equation*}
			Combining with (\ref{eq:gammahat-gamma*}), (\ref{eq:Shat-S}) and the results for 	(\ref{eq:bhat-Shatgamma_1}), 	(\ref{eq:bhat-Shatgamma_2}), we have 
			\begin{equation}\label{gammarate}
				\Vert\widehat{\boldsymbol{\gamma}}_{K_0}-\boldsymbol{\gamma}^*_{K_0}\Vert  =O_P\left(\sqrt{\frac{K_0}{n}}\right).
			\end{equation}	
			
			We now establish the convergence rates for $r_{K_0}^*(\cdot)$ and $\widehat r_{K_0}(\cdot)$. By Assumption \ref{ass:cpi0-pi*} (i), there exist   $\boldsymbol{\gamma}_{K_0} \in \mathbb{R}^{K_0} $ and a positive constant $\omega_{r}>0$ such that
			$\sup_{(y,z)\in\mathcal{Y}\times\mathcal{Z}}|r_0(y,z)-\boldsymbol{\gamma}^\top_{K_0}u_{K_0}(y,z)|=O(K_0^{-\omega_{r}})$, we have
			\begin{align*}
				&\Vert\boldsymbol{\gamma}^*_{K_0}-\boldsymbol{\gamma}_{K_0}\Vert=\Vert\Sigma_{K_0}^{-1}( b_{K_0}- \Sigma_{K_0}\boldsymbol{\gamma}_{K_0})\Vert \notag\\
				=& \left\{( b_{K_0}- \Sigma_{K_0}\boldsymbol{\gamma}_{K_0})^{\top }\Sigma_{K_0}^{-1} \Sigma_{K_0}^{-1}( b_{K_0}- \Sigma_{K_0}\boldsymbol{\gamma}_{K_0})\right\}^{1/2}\notag \\
				\leq &\lambda^{-1}_{\min}\left(\Sigma_{K_0}\right)\cdot \Vert b_{K_0}- \Sigma_{K_0}\boldsymbol{\gamma}_{K_0}\Vert\notag\\
				=&O(1)\cdot\left\Vert\int\mathbb{E}\left[u_{K_0}(y,Z)\right]f_Y(y)dy-\mathbb{E}\left[u_{K_0}(Y,Z)u_{K_0}^\top(Y,Z)\right]\boldsymbol{\gamma}_{K_0}\right\Vert\notag\\
				=&O(1)\cdot\left\Vert\int\mathbb{E}\left[u_{K_0}(y,Z)\right]f_Y(y)dy-\mathbb{E}\left[u_{K_0}(Y,Z)\left(u_{K_0}^\top(Y,Z)\boldsymbol{\gamma}_{K_0}-r_0(Y,Z)\right)\right]\right.\\
                &\left.-\mathbb{E}\left[u_{K_0}(Y,Z)r_0(Y,Z)\right]\right\Vert\notag\\
				=&O(1)\cdot\left\Vert\mathbb{E}\left[u_{K_0}(Y,Z)\left(u_{K_0}^\top(Y,Z)\boldsymbol{\gamma}_{K_0}-r_0(Y,Z)\right)\right]\right\Vert\notag\\
				\leq& O(1)\cdot\sqrt{\mathbb{E}\left[\left(u_{K_0}^\top(Y,Z)\boldsymbol{\gamma}_{K_0}-r_0(Y,Z)\right)^2\right]}\\
				=&O(K_0^{-\omega_{r}}).
			\end{align*}
	Then we can obtain
				\begin{align*}
	&\sup_{(y,z)\in\mathcal{Y}\times\mathcal{Z}} \left|r_{K_0}^*(y,z)-r_0(y,z)\right|\\
					\leq&\sup_{(y,z)\in\mathcal{Y}\times\mathcal{Z}} \left|u_{K_0}^\top(y,z)\left({\boldsymbol{\gamma}}_{K_0}^*-\boldsymbol{\gamma}_{K_0} \right) \right|+\sup_{(y,z)\in\mathcal{Y}\times\mathcal{Z}}
					\left|u_{K_0}^\top(y,z)\boldsymbol{\gamma}_{K_0}-r_{0}(y,z)\right|\\
					=&\sup_{(y,z)\in\mathcal{Y}\times\mathcal{Z}}\Vert u_{K_0}(y,z)\Vert\cdot\Vert{\boldsymbol{\gamma}}_{K_0}^*-\boldsymbol{\gamma}_{K_0}\Vert+O(K_0^{-\omega_{r}})\\
					=&O\left( \zeta_{K_0}K_0^{-\omega_r}\right) ,
				\end{align*}
			and under the normalization  \eqref{eq:orthnormal}, we have
				\begin{align*}\label{pi*-pi0_L2}
					&\int\left|r_{K_0}^*(y,z)-r_0(y,z)\right|^2dF_{Y,Z}(y,z)\\
					\leq&2\int\left|u_{K_0}^\top(y,z)\left({\boldsymbol{\gamma}}_{K_0}^*-\boldsymbol{\gamma}_{K_0} \right) \right|^2dF_{Y,Z}(y,z)+2\int\left|u_{K_0}^\top(y,z)\boldsymbol{\gamma}_{K_0}-r_{0}(y,z)\right|^2dF_{Y,Z}(y,z)\\
					=&2\left\Vert{\boldsymbol{\gamma}}_{K_0}^*-\boldsymbol{\gamma}_{K_0}\right\Vert^2+O\left( K_0^{-2\omega_r}\right)\\
					=&O\left( K_0^{-2\omega_r}\right),
				\end{align*}
			and
				\begin{align*}
					&	\frac{1}{n}\sum_{t=1}^{n}\left|r_{K_0}^*(Y_t,Z_t)-r_0(Y_t,Z_t)\right|^2\\
					\leq&\frac{2}{n}\sum_{t=1}^{n}\left|u_{K_0}^\top(Y_t,Z_t)\left( {\boldsymbol{\gamma}}_{K_0}^*-\boldsymbol{\gamma}_{K_0} \right) \right|^2+\frac{2}{n}\sum_{t=1}^{n}\left|u_K^\top(Y_t,Z_t)\boldsymbol{\gamma}_{K_0}-r_0(Y_t,Z_t)\right|^2\\
					\leq&2\cdot \lambda_{\max}\left(\frac{1}{n}\sum_{t=1}^{n}u_{K_0}(Y_t,Z_t)u_{K_0}^\top(Y_t,Z_t)\right)\cdot\left\Vert{\boldsymbol{\gamma}}_{K_0}^*-\boldsymbol{\gamma}_{K_0}\right\Vert^2+O_P\left( K_0^{-2\omega_r}\right)\\
					=&O_P\left( K_0^{-2\omega_r}\right).
				\end{align*}
			Finally, we derive the rate for $\widehat{r}_{K_0}(y,z)$. First, from Assumption \ref{ass:cpi0-pi*}(i) and (\ref{gammarate}), we have
				\begin{align*}
&\sup_{(y,z)\in\mathcal{Y}\times\mathcal{Z}}|\widehat{r}_{K_0}(y,z)-r_0(y,z)|\\\leq&\sup_{(y,z)\in\mathcal{Y}\times\mathcal{Z}}|\widehat{r}_{K_0}(y,z)-r_{K_0}^*(y,z)|+\sup_{(y,z)\in\mathcal{Y}\times\mathcal{Z}}|r_{K_0}^*(y,z)-r_0(y,z)|\\
					\leq&\sup_{(y,z)\in\mathcal{Y}\times\mathcal{Z}}|u_{K_0}^\top(y,z)(\widehat{\boldsymbol{\gamma}}_{K_0}-\boldsymbol{\gamma}_{K_0}^*)|+\sup_{(y,z)\in\mathcal{Y}\times\mathcal{Z}} |r_{K_0}^*(y,z)-r_0(y,z)|\\
					=&O_P\left(\zeta_{K_0}\left\{\sqrt{\frac{K_0}{n}}+K_0^{- \omega_r} \right\} \right).  
				\end{align*}
			Second, under the normalization   \eqref{eq:orthnormal} and $\int|u_{K_0}^\top(y,z)(\widehat{\boldsymbol{\gamma}}_{K_0}-\boldsymbol{\gamma}_{K_0}^*)|^2dF_{Y,Z}(y,z)=\Vert \widehat{\boldsymbol{\gamma}}_{K_0}-\boldsymbol{\gamma}_{K_0}^*\Vert^2=O_P\left(K_0/n\right) 
			$,  from triangle inequality, we have
			\begin{equation*}
				\begin{aligned}
					&\int|\widehat{r}_{K_0}(y,z)-r_0(y,z)|^2dF_{Y,Z}(y,z)\\
					\leq&2\int|\widehat{r}_{K_0}(y,z)-r_{K_0}^*(y,z)|^2dF_{Y,Z}(y,z)+2\int|r_{K_0}^*(y,z)-r_0(y,z)|^2dF_{Y,Z}(y,z)\\
					=&2\int|u_{K_0}^\top(y,z)(\widehat{\boldsymbol{\gamma}}_{K_0}-\boldsymbol{\gamma}_{K_0}^*)|^2dF_{Y, Z}(y,z)+O\left(K_0^{-2\omega_r
					} \right) \\
					=&O_P\left(\frac{K_0}{n}+K_0^{-2 \omega_r}\right).
				\end{aligned}
			\end{equation*}
			Third,  since
			\begin{align}
				&\frac{1}{n}\sum_{t=1}^{n}\left [\widehat r_{K_0}(Y_t,Z_t)-r_{K_0}^*(Y_t,Z_t)\right]^2 \notag\\
				=&\frac{1}{n}\sum_{t=1}^{n}[u_{K}^\top(Y_t,Z_t)(\widehat{\boldsymbol{\gamma}}_{K_0}-\boldsymbol{\gamma}_{K_0}^*)]^2\notag\\
				\leq&\lambda_{\max}\left(\frac{1}{n}\sum_{t=1}^{n}u_{K_0}(Y_t,Z_t)u_{K_0}^\top(Y_t,Z_t)\right)\cdot \|\widehat{\boldsymbol{\gamma}}_{K_0}-\boldsymbol{\gamma}_{K_0}^*\|^2\notag\\
				=&O_P\left(\frac{K_0}{n}\right),	\label{crhat-r*}
			\end{align}
			Therefore,
				\begin{align*}
					&\frac{1}{n}\sum_{t=1}^{n}\left|\widehat{r}_{K_0}(Y_t,Z_t)-r_0(Y_t,Z_t)\right|^2\\
					\leq&\frac{2}{n}\sum_{t=1}^{n}\left|\widehat{r}_{K_0}(Y_t,Z_t)-r_{K_0}^*(Y_t,Z_t)\right|^2+\frac{2}{n}\sum_{t=1}^{n}\left|r_{K_0}^*(Y_t,Z_t)-r_0(Y_t,Z_t)\right|^2\\
					=&O_P\left(\frac{K_0}{n}+K_0^{- 2\omega_r}\right).
				\end{align*}

For part (ii),	we first establish the convergence rate for $\|\widehat{\boldsymbol{\beta}}_{K,K_0}-\boldsymbol{\beta}^*_K\|$. Note that
		\begin{align}
			&\Vert\widehat{\boldsymbol{\beta}}_{K,K_0}-\boldsymbol{\beta}^*_{K}\Vert=\Vert\widehat H_{K,K_0}^{-1}(\widehat h_K-\widehat H_{K, K_0}\boldsymbol{\beta}^*_{K})\Vert \notag\\
			=& \left\{(\widehat h_K-\widehat H_{K, K_0}\boldsymbol{\beta}^*_K)^{\top }\widehat{H}_{K,K_0}^{-1} \widehat{H}_{K,K_0}^{-1}(\widehat h_K-\widehat H_{K, K_0}\boldsymbol{\beta}^*_K)\right\}^{1/2}\notag \\
			\leq &\lambda^{-1}_{\min}\left(\widehat{H}_{K,K_0}\right)\cdot \Vert\widehat h_K-\widehat H_{K, K_0}\boldsymbol{\beta}^*_K\Vert,\label{eq:betahat-beta*}
		\end{align}
		where $\lambda_{\min}(\widehat{H}_{K,K_0})$ is the smallest eigenvalue of $\widehat{H}_{K,K_0}$.
		
		We show $\|\widehat{H}_{K,K_0}-H_K\|=o_P(1)$ and $\lambda_{\min}(\widehat{H}_{K,K_0})=O_P(1)$. Let $\phi_{K,K_0}(X_i,Y_i,Z_i):=\{\widehat  r_{K_0}(Y_i,Z_i)-r_{0}(Y_i,Z_i)\}v_K(X_i,Y_i,Z_i)$ and $\text{proj}_{n, v_K}\phi_{K,K_0}(x,y,z)$ be the least-square projection of $\phi_{K,K_0}(x,y,z)$ on the space linearly spanned by $v_K(x,y,z)$:
		\begin{align*}
			&\text{proj}_{n, v_K}\phi_{K,K_0}(x,y,z)\\
			=&\left[\sum_{t=1}^{n}\phi_{K,K_0}(X_t, Y_t, Z_t)v_K^\top(X_t, Y_t, Z_t)\right]\left[\sum_{t=1}^{n}v_K(X_t, Y_t, Z_t)v_K^\top(X_t, Y_t, Z_t)\right]^{-1}v_K(x,y,z).
		\end{align*} 
		We have
		\begin{align}
			&\Vert\widehat H_{K,K_0}-H_{K}\Vert^2\notag\\
			\leq &\left\Vert\frac{1}{n}\sum_{t=1}^{n}(\widehat  r_{K_0}(Y_t,Z_t)-r_{0}(Y_t,Z_t))v_K(X_t, Y_t, Z_t)v_K^\top(X_t, Y_t, Z_t)\right\Vert^2\notag\\
			&+\left\Vert\frac{1}{n}\sum_{t=1}^{n}r_{0}(Y_t,Z_t)v_K(X_t, Y_t, Z_t)v_K^\top(X_t, Y_t, Z_t)-\mathbb{E}[r_{0}(Y_t,Z_t)v_K(X_t, Y_t, Z_t)v_K^\top(X_t, Y_t, Z_t)]\right\Vert^2\notag\\
			=&\left\|\frac{1}{n}\sum_{t=1}^{n}\phi_{K,K_0}(X_t, Y_t, Z_t)v_{K}^\top(X_t, Y_t, Z_t)\right\|^2+O_P\left(\frac{\xi^2_KK}{n}\right) \ (\text{by Assumption \ref {ass:cbounded} (i) and \eqref{eq:v_eigen}}) \notag\\
			=&\tr\left\{\left[\frac{1}{n}\sum_{t=1}^{n}\phi_{K,K_0}(X_t, Y_t, Z_t)v_{K}^\top(X_t, Y_t, Z_t)\right]\left[\frac{1}{n}\sum_{t=1}^{n}\phi_{K,K_0}(X_t, Y_t, Z_t)v_{K}^\top(X_t, Y_t, Z_t)\right]^\top\right\}\notag\\
            &+O_P\left(\frac{\xi^2_KK}{n}\right)\notag\\
			=&\tr\left\{\left[\frac{1}{n}\sum_{t=1}^{n}\phi_{K,K_0}(X_t, Y_t, Z_t)v_{K}^\top(X_t, Y_t, Z_t)\right]\frac{\lambda_{\min}\left\{\left(\frac{1}{n}\sum_{t=1}^{n}v_K(X_t, Y_t, Z_t)v_K^\top(X_t, Y_t, Z_t)\right)^{-1}\right\}}{\lambda_{\min}\left\{\left(\frac{1}{n}\sum_{t=1}^{n}v_K(X_t, Y_t, Z_t)v_K^\top(X_t, Y_t, Z_t)\right)^{-1}\right\}}\right.\notag\\
			&\left.\cdot I_{K\times K}\cdot\left[\frac{1}{n}\sum_{t=1}^{n}\phi_{K,K_0}(X_t, Y_t, Z_t)v_{K}^\top(X_t, Y_t, Z_t)\right]^\top\right\}+O_P\left(\frac{\xi^2_KK}{n}\right)\notag\\
			\leq&\frac{1}{\lambda_{\min}\left\{\left[\frac{1}{n}\sum_{t=1}^{n}v_K(X_t, Y_t, Z_t)v_K^\top(X_t, Y_t, Z_t)\right]^{-1}\right\}}\tr\left\{\left[\frac{1}{n}\sum_{t=1}^{n}\phi_{K,K_0}(X_t, Y_t, Z_t)v_{K}^\top(X_t, Y_t, Z_t)\right]\right.\notag\\
			&\left.	\cdot\left[\frac{1}{n}\sum_{t=1}^{n}v_K(X_t, Y_t, Z_t)v_K^\top(X_t, Y_t, Z_t)\right]^{-1}\left[\frac{1}{n}\sum_{t=1}^{n}v_K(X_t, Y_t, Z_t)v_K^\top(X_t, Y_t, Z_t)\right]\right.\notag\\
			&\left.\cdot\left[\frac{1}{n}\sum_{t=1}^{n}v_K(X_t, Y_t, Z_t)v_K^\top(X_t, Y_t, Z_t)\right]^{-1}\left[\frac{1}{n}\sum_{t=1}^{n}\phi_{K,K_0}(X_t, Y_t, Z_t)v_{K}^\top(X_t, Y_t, Z_t)\right]^\top\right\} \notag\\
            &+O_P\left(\frac{\xi^2_KK}{n}\right)\notag\\
			\leq&O_P\left(1\right)\cdot \frac{1}{n}\sum_{t=1}^{n}\left\Vert \text{proj}_{n, v_K}\phi_{K,K_0}(X_t, Y_t, Z_t)\right\Vert^2+O_P\left(\frac{\xi^2_KK}{n}\right) \ (\text{by \eqref{eq:v_eigen}}) \notag\\
			\leq&O_P\left(1\right)\cdot \frac{1}{n}\sum_{t=1}^{n}\left\Vert \phi_{K,K_0}(X_t, Y_t, Z_t)\right\Vert^2+O_P\left(\frac{\xi^2_KK}{n}\right) \ \text{(by \eqref{eq:lsproj})} \notag \\
			\leq&O_P\left(1\right)\cdot \sup_{(x,y,z)\in\mathcal{X}\times\mathcal{Y}\times \mathcal{Z}}\Vert v_{K}(x,y,z)\Vert^2\cdot\frac{1}{n}\sum_{t=1}^{n}\left|\widehat  r_{K_0}(Y_t,Z_t)-r_{0}(Y_t,Z_t)\right|^2+O_P\left(\frac{\xi^2_KK}{n}\right)\notag\\
			=&O_P\left(\xi_K^2\left\{K_0^{-2\omega_{r}}+\frac{K_0}{n}+\frac{K}{n}\right\} \right)=o_P(1).  \ \text{(by Theorem \ref{th_crate} (i) and Assumption  \ref{ass:cK} (ii))} \label{Hhat-H}
		\end{align}
       
	 With Assumption \ref{ass:eigen} (ii), we  have 
		\begin{align}\label{eq:Hhat-H}
			\Vert\widehat H_{K,K_0}-H_{K}\Vert=o_P(1), \ \lambda^{-1}_{\min}\left(\widehat H_{K,K_0}\right)=O_P(1),
		\end{align} 
		and the eigenvalues of $\widehat H_{K,K_0}$ are uniformly bounded away from zero with probability approaching to one. 
		Next, we establish the convergence rate for $\|\widehat h_K -\widehat H_{K,K_0} \boldsymbol{\beta}_K^*\|$.
  As Lemma  \ref{uproj} stated, for $0<\epsilon<1/4$,
	\begin{equation*}
		\begin{aligned}
			&\frac{1}{n(n-1)}\sum_{i=1}^{n}\sum_{j=1,j\neq i}^{n}v_{K}(X_i,Y_j,Z_i)\\
			=&\frac{1}{n}\sum_{t=1}^{n}\mathbb{E}[v_K(X_t,Y_t,Z_t)r_0(Y_t,Z_t)\pi_0(X_t,Y_t,Z_t)|X_t,Z_t]\\
            &+\frac{1}{n}\sum_{t=1}^{n}\mathbb{E}[v_K(X_t,Y_t,Z_t)r_0(Y_t,Z_t)\pi_0(X_t,Y_t,Z_t)|Y_t]\\
			&-\mathbb{E}[v_K(X,Y,Z)r_0(Y,Z)\pi_0(X,Y,Z)]
			+O_P\left( \frac{\sqrt{\xi_K}K^{1/4}}{n^{(2-\epsilon)/2}}\right).
		\end{aligned}
	\end{equation*}
 Then we have
		\begin{align}
			&\widehat h_K -\widehat H_{K,K_0} \boldsymbol{\beta}_K^* \notag\\
			=&\frac{1}{n(n-1)}\sum_{i=1,i\neq j}^{n}\sum_{j=1}^{n}v_{K}(X_i,Y_j,Z_i)- \frac{1}{n}\sum_{t=1}^{n}\widehat r_{K_0}(Y_t,Z_t)v_{K}(X_t, Y_t, Z_t)v_K^\top(X_t, Y_t, Z_t) \boldsymbol{\beta}^*_{K} \notag\\
   	=&\frac{1}{n}\sum_{t=1}^{n}\int v_K(X_t,y,Z_t)f_Y(y)dy- \frac{1}{n}\sum_{t=1}^{n}\widehat r_{K_0}(Y_t,Z_t)v_{K}(X_t, Y_t, Z_t)v_K^\top(X_t, Y_t, Z_t) \boldsymbol{\beta}^*_{K}\notag\\
    &+\frac{1}{n}\sum_{t=1}^{n}\int v_K(x,Y_t,z)f_{X,Z}(x,z)dxdz
			-\mathbb{E}[v_K(X,Y,Z)r_0(Y,Z)\pi_0(X,Y,Z)]
			\notag\\
            &+O_P\left( \frac{\sqrt{\xi_K}K^{1/4}}{n^{(2-\epsilon)/2}}\right)\notag\\
			=&\frac{1}{n}\sum_{t=1}^{n}\varphi_K(X_t,Z_t)-\frac{1}{n}\sum_{t=1}^{n} r_{0}(Y_t,Z_t)v_{K}(X_t, Y_t, Z_t)v_K^\top(X_t, Y_t, Z_t) \boldsymbol{\beta}^*_{K} \label{eq:hhat-Hhatbeta_1}\\
			&-\frac{1}{n}\sum_{t=1}^{n}v_{K}(X_t, Y_t, Z_t)v_K^\top(X_t, Y_t, Z_t)\boldsymbol{\beta}^*_{K}\cdot [\widehat r_{K_0}(Y_t,Z_t)-r_{0}(Y_t,Z_t)]\label{eq:hhat-Hhatbeta_2}\\
   &+\frac{1}{n}\sum_{t=1}^{n}\int v_K(x,Y_t,z)f_{X,Z}(x,z)dxdz
			-\mathbb{E}[v_K(X,Y,Z)r_0(Y,Z)\pi_0(X,Y,Z)]\label{eq:hhat-Hhatbeta_3}
		\\
        &+O_P\left( \frac{\sqrt{\xi_K}K^{1/4}}{n^{(2-\epsilon)/2}}\right),\label{eq:hhat-Hhatbeta_4}
		\end{align}	
  where $\varphi_K(x,z):=\int v_K(x,y,z)f_Y(y)dy=\left( \varphi_{K,1}(x,z), ... ,\varphi_{K,K}(x,z)\right)^\top$.
  For the term \eqref{eq:hhat-Hhatbeta_3},
  \begin{align*}
     \eqref{eq:hhat-Hhatbeta_3} =\int\left\{ \frac{1}{n}\sum_{t=1}^{n}v_K(x,Y_t,z)-\int v_K(x,y,z)f_Y(y)dy\right\}f_{X,Z}(x,z)dxdz
			=O_P\left(\sqrt{\frac{K}{n}}\right)
  .\end{align*}
  For the term \eqref{eq:hhat-Hhatbeta_4}, 
  by  $0<\epsilon<1/4$ and Assumption \ref{ass:cK}(ii),  $\eqref{eq:hhat-Hhatbeta_4}=O_P\left(\sqrt{K/n}\right).$ For the term \eqref{eq:hhat-Hhatbeta_1},  by Assumption \ref{ass:ciid},  \ref{ass:eigen}  and Lemma \ref{Davydov}, we have
  \begin{small}
		\begin{align*}
			&\mathbb{E}\left[\|\eqref{eq:hhat-Hhatbeta_1}\|^2\right]
		=\mathbb{E}\left[\left\Vert\frac{1}{n}\sum_{t=1}^{n}\left\{r_0(Y_t,Z_t)v_K(X_t, Y_t, Z_t)v_K^\top(X_t, Y_t, Z_t)\boldsymbol{\beta}^*_{K}-\varphi_K(X_t,Z_t)\right\}\right\Vert^2\right] \\
			=&\frac{1}{n^2}\mathbb{E}\left[\sum_{t=1}^{n}\left\Vert\left\{r_0(Y_t,Z_t)v_K(X_t, Y_t, Z_t)v_K^\top(X_t, Y_t, Z_t)\boldsymbol{\beta}^*_{K}-\varphi_K(X_t,Z_t)\right\}\right\Vert^2\right]\\
			&+\frac{1}{n^2}\mathbb{E}\left[\sum_{i\neq j}\left\{r_0(Y_i,Z_i)v_K(X_i,Y_i,Z_i)v_K^\top(X_i,Y_i,Z_i)\boldsymbol{\beta}^*_{K}-\varphi_K(X_i,Z_i)\right\}^\top\right.\\
			&\left.\cdot\left\{r_0(Y_j,Z_j)v_K(X_j,Y_j,Z_j)v_K^\top(X_j,Y_j,Z_j)\boldsymbol{\beta}^*_{K}-\varphi_K(X_j,Z_j)\right\}\right]\\
			=&\frac{1}{n}\mathbb{E}\left[\left\Vert\left\{r_0(Y,Z)v_K(X,Y,Z)v_K^\top(X,Y,Z)\boldsymbol{\beta}^*_{K}-\varphi_K(X,Z)\right\}\right\Vert^2\right]\\
			&+\frac{2}{n^2}\sum_{k=1}^{K}\sum_{t=1}^{n}\sum_{\tau=1}^{n-t}Cov\left(r_0(Y_t,Z_t)v_{K,k}(X_t,Y_t,Z_t)v_K^\top(X_t,Y_t,Z_t)\boldsymbol{\beta}^*_{K}-\varphi_{K,k}(X_t,Z_t),\right.\\ &\left.r_0(Y_{t+\tau},Z_{t+\tau})v_{K,k}(X_{t+\tau},Y_{t+\tau},Z_{t+\tau})v_K^\top(X_{t+\tau},Y_{t+\tau},Z_{t+\tau})\boldsymbol{\beta}^*_{K}-\varphi_{K,k}(X_{t+\tau},Z_{t+\tau}) \right)  \\
			\leq&\frac{8}{n^2}\sum_{k=1}^{K}\sum_{t=1}^{n}\sum_{\tau=1}^{n-t}\left\{\mathbb{E}\left[|r_0(Y_t,Z_t)v_{K,k}(X_t,Y_t,Z_t)v_K^\top(X_t,Y_t,Z_t)\boldsymbol{\beta}^*_{K}-\varphi_{K,k}(X_t,Z_t)|^4\right]\right\}^{1/2}\beta_{\tau}^{1/2}\\
			&+\frac{2}{n}\mathbb{E}\left[\Vert r_0(Y,Z)v_K(X,Y,Z)v_K^\top(X,Y,Z)\boldsymbol{\beta}_K^*\Vert^2\right]+\frac{2}{n}\mathbb{E}\left[\Vert \varphi_K(X,Z)\Vert^2\right] \\
			\leq&\frac{1}{n^2}\sum_{k=1}^{K}\sum_{t=1}^{n}C+\frac{2}{n}\mathbb{E}\left[\Vert r_0(Y,Z)v_K(X,Y,Z)\Vert^2\right]\sup_{(x,y,z)\in\mathcal{X}\times\mathcal{Y}\times \mathcal{Z}}| v_K^\top(x,y,z)\boldsymbol{\beta}^*_{K}|^2+\frac{2}{n}\mathbb{E}\left[\Vert \varphi_K(X,Z)\Vert^2\right]\\
			\leq& O\left(\frac{1}{n}\right)\mathbb{E}\left[\Vert v_K(X,Y,Z)\Vert^2\right]\left\{\sup_{(x,y,z)\in\mathcal{X}\times\mathcal{Y}\times \mathcal{Z}}\pi_0^2(x,y,z)+\sup_{(x,y,z)\in\mathcal{X}\times\mathcal{Y}\times \mathcal{Z}}\left|\pi_0(x,y,z)-\pi_K^*(x,y,z)\right|^2\right\}\\
			&+\frac{2}{n}\sup_{(x,y,z)\in\mathcal{X}\times\mathcal{Y}\times \mathcal{Z}}\frac{f_{X,Z}(x,z)f_Y(y)}{f_{X,Y,Z}(x,y,z)}\mathbb{E}\left[\Vert v_K(X,Y,Z)\Vert^2\right]+O\left(\frac{K}{n} \right) \\
			=&O\left(\frac{K}{n}\right),
		\end{align*}
    \end{small}
	where $C$ is a universal constant and the third equality comes from Assumption \ref{ass:ciid},  implying that
		\begin{equation*}
			\left\{r_0(Y_t,Z_t)v_K(X_t, Y_t, Z_t)v_K^\top(X_t, Y_t, Z_t)\boldsymbol{\beta}^*_{K}-\varphi_K(X_t,Z_t)\right\}_{t=1}^{n}
		\end{equation*} 
		is a mean zero sequence since 
		\begin{equation*}
			\begin{aligned}
			&\mathbb{E}\left[r_0(Y_t,Z_t)v_K(X_t, Y_t, Z_t)v_K^\top(X_t, Y_t, Z_t)\boldsymbol{\beta}^*_{K}\right]\\
				=&\mathbb{E}\left[r_0(Y_t,Z_t)v_K(X_t, Y_t, Z_t)v_K^\top(X_t, Y_t, Z_t)\right]H_K^{-1}h_K\\
				=&\mathbb{E}\left[r_0(Y_t,Z_t)v_K(X_t, Y_t, Z_t)v_K^\top(X_t, Y_t, Z_t)\right]\left\{\mathbb{E}\left[r_0(Y,Z)v_K(X,Y,Z)v_K^\top(X,Y,Z)\right]\right\}^{-1}\\
				&\quad\cdot\int \mathbb{E}[v_K(X,y,Z)]f_Y(y)dy\\
				=&\int \mathbb{E}[v_K(X,y,Z)]f_Y(y)dy
				=\int v_K(x,y,z)f_Y(y)f_{X,Z}(x,z)dxdydz
				=\mathbb{E}\left\{\varphi_K(X_t,Z_t)\right\};
			\end{aligned}
		\end{equation*}
		The last inequality follows from the triangle inequality and
		$\mathbb{E}[\Vert\varphi(X,Z)\Vert^2]\leq\int\Vert v_K(x,y,z)\Vert^2$ $f_{X,Z}(x,z)f_Y(y)dxdydz$,  and the last equality follows from Assumptions \ref{ass:cbounded}(ii) and \ref{ass:cpi0-pi*}(ii).
		Then, by Chebyshev's inequality, we have
		\begin{align*}
			\eqref{eq:hhat-Hhatbeta_1}
			=&O_P\left(\sqrt{\frac{K}{n}}\right).
		\end{align*}
		Consider the term \eqref{eq:hhat-Hhatbeta_2}. Similar to the proof of establishing \eqref{Hhat-H}, we have
		\begin{align*}
			\|\eqref{eq:hhat-Hhatbeta_2}\|^2=&\left\|\frac{1}{n}\sum_{t=1}^{n}v_{K}(X_t, Y_t, Z_t)\pi_K^*(X_t, Y_t, Z_t)(\widehat r_{K_0}(Y_t,Z_t)-r_{0}(Y_t,Z_t))\right\|^2\notag\\
			\leq&O_P\left(1\right)\cdot \frac{1}{n}\sum_{t=1}^{n}| \pi_K^*(X_t, Y_t, Z_t)(\widehat r_{K_0}(Y_t,Z_t)-r_{0}(Y_t,Z_t))|^2\notag\\
			\leq &O_P(1)\cdot\frac{1}{n}\sum_{t=1}^{n}[\widehat r_{K_0}(Y_t,Z_t)-r_{0}(Y_t,Z_t)]^2
			=O_P\left(\frac{K_0}{n}+K_0^{-2\omega_{r}}\right),
		\end{align*}
		where the last equality holds in light of Theorem \ref{th_crate} (i).  Combining \eqref{eq:betahat-beta*}, \eqref{eq:Hhat-H}, the results for \eqref{eq:hhat-Hhatbeta_1},  
 \eqref{eq:hhat-Hhatbeta_2}, \eqref{eq:hhat-Hhatbeta_3} and \eqref{eq:hhat-Hhatbeta_4},  we have 
		\begin{equation}
			\Vert\widehat{\boldsymbol{\beta}}_{K,K_0}-\boldsymbol{\beta}^*_{K}\Vert=O_P\left(\left\{\sqrt{\frac{K_0}{n}}+K_0^{-\omega_{r}}\right\}+\sqrt{\frac{K}{n}}\right).\label{betarate}
		\end{equation}

		We turn to establish the rates of convergence for $\pi_K^*(\cdot)$ and $\widehat{\pi}_{K,K_0}(\cdot)$. By Assumption \ref{ass:cpi0-pi*} (ii), there exist   $\boldsymbol{\beta}_K \in \mathbb{R}^{K} $ and a positive constant $\omega_{\pi}>0$ such that
		$\sup_{(x,y,z)\in\mathcal{X}\times\mathcal{Y}\times \mathcal{Z}}|\pi_0(x,y,z)-\boldsymbol{\beta}^\top_Kv_{K}(x,y,z)|=O(K^{-\omega_{\pi}})$, we have
        \begin{small}
		\begin{align*}
			&\Vert{\boldsymbol{\beta}}_{K}^*-\boldsymbol{\beta}_K\Vert=\left\Vert H_K^{-1}\left(  h_K- H_K\boldsymbol{\beta}_K\right) \right\Vert\\
			=& \left\{ \left( h_K- H_K\boldsymbol{\beta}_K\right) ^{\top }H_{K}^{-1} H_{K}^{-1}\left(  h_K- H_{K}\boldsymbol{\beta}_K\right) \right\}^{1/2}\\
			\leq &\lambda_{\min}^{-1}\left(H_{K}\right)\cdot \left\Vert h_K- H_{K}\boldsymbol{\beta}_K\right\Vert \\
			=&O\left( 1\right) \cdot\left\Vert\int \mathbb{E}[v_K(X,y,Z)]f_Y(y)dy-\mathbb{E}\left[r_0(Y,Z)v_K(X,Y,Z)v_K^\top(X,Y,Z)\right]\boldsymbol{\beta}_K\right\Vert\\
			=&O\left( 1\right) \cdot\bigg\Vert\int \mathbb{E}[v_K(X,y,Z)]f_Y(y)dy-\mathbb{E}\left[r_0(Y,Z)v_K(X,Y,Z)\left(v_K^\top(X,Y,Z)\boldsymbol{\beta}_K-\pi_0(X,Y,Z)\right) \right]\\
			&\qquad\qquad -\mathbb{E}[r_0(Y,Z)v_K(X,Y,Z)\pi_0(X,Y,Z)]\bigg\Vert\\
			=&O\left( 1\right) \cdot\left\Vert\mathbb{E}\left[r_0(Y,Z)v_K(X,Y,Z)\left(v_K^\top(X,Y,Z)\boldsymbol{\beta}_K-\pi_0(X,Y,Z)\right) \right]\right\Vert\\
			\leq&O\left( 1\right)\cdot\sqrt{\mathbb{E}\left[\left|v_K^\top(X,Y,Z)\boldsymbol{\beta}_K-\pi_0(X,Y,Z)\right|^2\right]}\\
			\leq&O\left( 1\right)\cdot\sup_{(x,y,z)\in\mathcal{X}\times\mathcal{Y}\times \mathcal{Z}}|\pi_0(x,y,z)-\boldsymbol{\beta}^\top_Kv_{K}(x,y,z)|
			=O\left( K^{-\omega_{\pi}}\right), 
		\end{align*}
        \end{small}
		where the  fifth equality comes from
					\begin{align*}
				&\mathbb{E}[r_0(Y,Z)\pi_0(X,Y,Z)v_K(X,Y,Z)]\\
				=&\int v_K(x,y,z)\frac{f_Y(y)f_Z(z)}{f_{Y,Z}(y,z)}\frac{f_{X|Z}(x|z)f_{Y|Z}(y|z)}{f_{X,Y|Z}(x,y|z)}f_{X,Y,Z}(x,y,z)dxdydz\\
				=&\int v_K(x,y,z)f_Y(y)f_{X,Z}(x,z)dxdydz\\
				=&\int\mathbb{E}[v_K(X,y,Z)]f_Y(y)dy,
			\end{align*}
		and the second inequality comes from Assumption \ref{ass:cbounded} (i) and the property of least square projection. Then we can obtain
			\begin{align*}			&\sup_{(x,y,z)\in\mathcal{X}\times\mathcal{Y}\times\mathcal{Z}} \left|\pi_{K}^*(x,y,z)-\pi_0(x,y,z)\right|\\
				\leq&\sup_{(x,y,z)\in\mathcal{X}\times\mathcal{Y}\times\mathcal{Z}} \left|v_{K}^\top(x,y,z)\left({\boldsymbol{\beta}}_{K}^*-\boldsymbol{\beta}_K \right) \right|+\sup_{(x,y,z)\in\mathcal{X}\times\mathcal{Y}\times\mathcal{Z}} \left|v_K^\top(x,y,z)\boldsymbol{\beta}_K-\pi_{0}(x,y,z)\right|\\
				=&\sup_{(x,y,z)\in\mathcal{X}\times\mathcal{Y}\times \mathcal{Z}}\Vert v_K(x,y,z)\Vert\cdot\Vert{\boldsymbol{\beta}}_{K}^*-\boldsymbol{\beta}_K\Vert+O(K^{-\omega_{\pi}})\\
				=&O\left( \xi_KK^{-\omega_\pi}\right) ,
			\end{align*}
		and under the normalization  \eqref{eq:orthnormal}, we have
			\begin{align}	\label{pi*-pi0_L2}
				&\int\left|\pi_{K}^*(x,y,z)-\pi_0(x,y,z)\right|^2dF_{XYZ}(x,y,z)\notag\\
				\leq&2\int\left|v_{K}^\top(x,y,z)\left({\boldsymbol{\beta}}_{K}^*-\boldsymbol{\beta}_K \right) \right|^2dF_{XYZ}(x,y,z)\notag\\
                &+2\int\left|v_K^\top(x,y,z)\boldsymbol{\beta}_K-\pi_{0}(x,y,z)\right|^2dF_{XYZ}(x,y,z)\notag\\
				=&2\left\Vert{\boldsymbol{\beta}}_{K}^*-\boldsymbol{\beta}_K\right\Vert^2+O\left( K^{-2\omega_\pi}\right)\notag\\
				=&O\left( K^{-2\omega_\pi}\right),
			\end{align}
		and
		\begin{equation*}
			\begin{aligned}
				&	\frac{1}{n}\sum_{t=1}^{n}\left|\pi_{K}^*(X_t, Y_t, Z_t)-\pi_0(X_t, Y_t, Z_t)\right|^2\\
				\leq&\frac{2}{n}\sum_{t=1}^{n}\left|v_{K}^\top(X_t, Y_t, Z_t)\left( {\boldsymbol{\beta}}_{K}^*-\boldsymbol{\beta}_K \right) \right|^2+\frac{2}{n}\sum_{t=1}^{n}\left|v_K^\top(X_t, Y_t, Z_t)\boldsymbol{\beta}_K-\pi_0(X_t, Y_t, Z_t)\right|^2\\
				\leq&2\cdot \lambda_{\max}\left(\frac{1}{n}\sum_{t=1}^{n}v_{K}(X_t, Y_t, Z_t)v_{K}^\top(X_t, Y_t, Z_t)\right)\cdot\left\Vert{\boldsymbol{\beta}}_{K}^*-\boldsymbol{\beta}_K\right\Vert^2+O_P\left( K^{-2\omega_\pi}\right)\\
				=&O_P\left( K^{-2\omega_\pi}\right).\\
			\end{aligned}
		\end{equation*}

        We derive the rate for $\widehat{\pi}_{K,K_0}(x,y,z)$. First, from Assumption \ref{ass:cpi0-pi*} and (\ref{betarate}), we have
		\begin{equation*}
			\begin{aligned}
				&\sup_{(x,y,z)\in\mathcal{X}\times\mathcal{Y}\times\mathcal{Z}}|\widehat{\pi}_{K,K_0}(x,y,z)-\pi_0(x,y,z)|\\\leq&\sup_{(x,y,z)\in\mathcal{X}\times\mathcal{Y}\times\mathcal{Z}}|\widehat{\pi}_{K,K_0}(x,y,z)-\pi_{K}^*(x,y,z)|+\sup_{(x,y,z)\in\mathcal{X}\times\mathcal{Y}\times\mathcal{Z}}|\pi_{K}^*(x,y,z)-\pi_0(x,y,z)|\\
				\leq&\sup_{(x,y,z)\in\mathcal{X}\times\mathcal{Y}\times\mathcal{Z}}|v_{K}^\top(x,y,z)(\widehat{\boldsymbol{\beta}}_{K,K_0}-\boldsymbol{\beta}_K^*)|+\sup_{(x,y,z)\in\mathcal{X}\times\mathcal{Y}\times\mathcal{Z}} |\pi_{K}^*(x,y,z)-\pi_0(x,y,z)|\\
				=&O_P\left(\xi_{K}\left\{\sqrt{\frac{K_0}{n}}+K_0^{- \omega_r}+\sqrt{\frac{K}{n}} \right\}+\xi_KK^{-\omega_\pi} \right).  
			\end{aligned}
		\end{equation*}
		Second, under the normalization  $\mathbb{E}[v_{K}(X,Y,Z)v_{K}^\top(X,Y,Z)]=I_{K\times K}$ and $\int|v_{K}^\top(x,y,z)(\widehat{\boldsymbol{\beta}}_{K,K_0}-\boldsymbol{\beta}_K^*)|^2dF_{XYZ}(x,y,z)=\Vert \widehat{\boldsymbol{\beta}}_{K,K_0}-\boldsymbol{\beta}_K^*\Vert^2=O_P\left(K_0/n+K_0^{-2 \omega_r}+K/n\right) 
		$,  from triangle inequality, we have
			\begin{align*}
		&\int|\widehat{\pi}_{K,K_0}(x,y,z)-\pi_0(x,y,z)|^2dF_{XYZ}(x,y,z)\\
				\leq&2\int|\widehat{\pi}_{K,K_0}(x,y,z)-\pi_{K}^*(x,y,z)|^2dF_{XYZ}(x,y,z)\\&+2\int|\pi_{K}^*(x,y,z)-\pi_0(x,y,z)|^2dF_{XYZ}(x,y,z)\\
				=&2\int|v_{K}^\top(x,y,z)(\widehat{\boldsymbol{\beta}}_{K,K_0}-\boldsymbol{\beta}_K^*)|^2dF_{XYZ}(x,y,z)+O\left(K^{-2\omega_\pi
				} \right) \\
				=&O_P\left(\left\{\frac{K_0}{n}+K_0^{-2 \omega_r}\right\}+\left\{\frac{K}{n}+K^{-2\omega_\pi}\right\}\right).
			\end{align*}
		Third,  since
  \begin{equation}\label{cpihat-pi*}
		\begin{aligned}
			&\frac{1}{n}\sum_{t=1}^{n}\left [\widehat\pi_{K,K_0}(X_t, Y_t, Z_t)-\pi_K^*(X_t, Y_t, Z_t)\right]^2 \\
			=&\frac{1}{n}\sum_{t=1}^{n}[v_{K}^\top(X_t, Y_t, Z_t)(\widehat{\boldsymbol{\beta}}_{K,K_0}-\boldsymbol{\beta}_K^*)]^2\\
			\leq&\lambda_{\max}\left(\frac{1}{n}\sum_{t=1}^{n}v_{K}(X_t, Y_t, Z_t)v_{K}^\top(X_t, Y_t, Z_t)\right)\cdot \|\widehat{\boldsymbol{\beta}}_{K,K_0}-\boldsymbol{\beta}_K^*\|^2\\
			=&O_P\left(\frac{K_0}{n}+K_0^{-2 \omega_r}+\frac{K}{n}\right),	
		\end{aligned}
  \end{equation}
		Therefore,
		\begin{equation*}
			\begin{aligned}
				&\frac{1}{n}\sum_{t=1}^{n}\left|\widehat{\pi}_{K,K_0}(X_t, Y_t, Z_t)-\pi_0(X_t, Y_t, Z_t)\right|^2\\
				\leq&\frac{2}{n}\sum_{t=1}^{n}\left|\widehat{\pi}_{K,K_0}(X_t, Y_t, Z_t)-\pi_{K}^*(X_t, Y_t, Z_t)\right|^2+\frac{2}{n}\sum_{t=1}^{n}\left|\pi_{K}^*(X_t, Y_t, Z_t)-\pi_0(X_t, Y_t, Z_t)\right|^2\\
				=&O_P\left(\left\{\frac{K_0}{n}+K_0^{- 2\omega_r}\right\}+\left\{\frac{K}{n}+K^{-2\omega_\pi}\right\}\right).
			\end{aligned}
		\end{equation*}

		\clearpage

	\section{Proof of Theorem \ref{th_cH0_fixk}}
\label{app:th_cH0_fixk}
Under $H_0$, $\pi_0(x,y,z)=\pi_K^*(x,y,z)\equiv 1$ a.s.. 
By Theorem \ref{th_crate} and Chebyshev inequality, we have
\begin{equation*}
	\begin{aligned}
	2	\widehat{I}_{K,K_0}&=\frac{1}{n}\sum_{t=1}^n\left\{\widehat{\pi}_{K,K_0}(X_t, Y_t, Z_t)-1\right\}^2\widehat{r}_{K_0}(Y_t,Z_t)\\
		&=\frac{1}{n}\sum_{t=1}^{n}\{\widehat{\pi}_{K,K_0}(X_t, Y_t, Z_t)-\pi_{K}^*(X_t, Y_t, Z_t)\}^2r_{0}(Y_t,Z_t)\\
		&\quad+	\frac{1}{n}\sum_{t=1}^{n}\{\widehat{\pi}_{K,K_0}(X_t, Y_t, Z_t)-\pi_{K}^*(X_t, Y_t, Z_t)\}^2\{\widehat{r}_{K_0}(Y_t,Z_t)-{r}_{0}(Y_t,Z_t)\}\\
		&=\int r_0(y,z)\left[v_{K}^\top(x,y,z)\left( \widehat{\boldsymbol{\beta}}_{K,K_0}-\boldsymbol{\beta}^*_K\right) \right]^2dF_{X,Y,Z}(x,y,z)\\
         &\quad+O_P\left(\xi_K\sqrt{\frac{K}{n}}\left\{\frac{K_0}{n}+K_0^{-2\omega_r}+\frac{K}{n}\right\}   \right)\\
		 &\quad+O_P\left(\xi_K\left\{K_0^{-\omega_{r}}+\sqrt{\frac{K_0}{n}}\right\}\left\{\frac{K_0}{n}+K_0^{-2\omega_r}+\frac{K}{n}\right\} \right)  \\
		&=\left( \widehat{\boldsymbol{\beta}}_{K,K_0}-\boldsymbol{\beta}^*_K\right)^{\top} \mathbb{E}\left[r_0(Y,Z)v_K(X,Y,Z)v_K(X,Y,Z)^{\top}\right]\left( \widehat{\boldsymbol{\beta}}_{K,K_0}-\boldsymbol{\beta}^*_K\right)\\
		&\quad+O_P\left(\xi_K\left\{K_0^{-\omega_{r}}+\sqrt{\frac{K_0}{n}}+\sqrt{\frac{K}{n}}\right\}\left\{\frac{K_0}{n}+K_0^{-2\omega_r}+\frac{K}{n}\right\} \right), 
	\end{aligned}
\end{equation*}
where the third equality comes from
	\begin{align*}
		&	\frac{1}{n}\sum_{t=1}^nr_0(Y_t,Z_t)\left\{v_{K}^\top(X_t, Y_t, Z_t)(\widehat{\boldsymbol{\beta}}_{K,K_0}-\boldsymbol{\beta}^*_K)\right\}^2
       -\mathbb{E}\left[ r_0(Y,Z)\left\{v_{K}^\top(X,Y,Z)(\widehat{\boldsymbol{\beta}}_{K,K_0}-\boldsymbol{\beta}^*_K)\right\}^2\right]\\
		=&(\widehat{\boldsymbol{\beta}}_{K,K_0}-\boldsymbol{\beta}^*_K)^\top\left\{\frac{1}{n}\sum_{t=1}^{n}r_0(Y_t,Z_t)v_K(X_t, Y_t, Z_t)v_K^\top(X_t, Y_t, Z_t)\right.\\
        &\left.-\mathbb{E}\left[r_0(Y,Z)v_K(X,Y,Z)v_K^\top(X,Y,Z)\right]\right\}(\widehat{\boldsymbol{\beta}}_{K,K_0}-\boldsymbol{\beta}^*_K)\\
		\leq&\left\Vert\frac{1}{n}\sum_{t=1}^{n}r_0(Y_t,Z_t)v_K(X_t, Y_t, Z_t)v_K^\top(X_t, Y_t, Z_t)-\mathbb{E}\left[r_0(Y,Z)v_K(X,Y,Z)v_K^\top(X,Y,Z)\right]\right\Vert\\
        &\cdot\left\Vert\widehat{\boldsymbol{\beta}}_{K,K_0}-\boldsymbol{\beta}^*_K\right\Vert^2\\
		=&O_P\left(\xi_K\sqrt{\frac{K}{n}}\left\{\frac{K_0}{n}+K_0^{-2\omega_r}+\frac{K}{n}\right\}   \right).
	\end{align*}
Then, by Lemma \ref{prop:vtilde}, we can obtain
\begin{align}
2	\widehat{I}_{K,K_0}
	=&\left( \frac{1}{n}\sum_{t=1}^{n} \widetilde v_{K}^\top(X_t,Y_t,Z_t)\right) H_{K}^{-1}\left( \frac{1}{n}\sum_{t=1}^{n}\widetilde v_{K}(X_t,Y_t,Z_t)\right) +o_P\left( \frac{\sqrt{K}}{n}\right), 	\label{Ichat}
\end{align}	
where under $H_0$
\begin{equation*}
\begin{aligned}
	\widetilde v_{K}(X_t,Y_t,Z_t)
=&\mathbb{E}\left[v_K(X_t,Y_t,Z_t)r_0(Y_t,Z_t)|X_t,Z_t\right]-\mathbb{E}\left[v_K(X_t,Y_t,Z_t)r_0(Y_t,Z_t)|Z_t\right]\\
&+\mathbb{E}\left[v_K(X,Y,Z)r_0(Y,Z)\right]-v_{K}(X_t,Y_t,Z_t)r_0(Y_t,Z_t).
\end{aligned}
\end{equation*}
		Define
		\begin{equation*}
			S_0(X_t, Y_t, Z_t;x,y,z)=\sqrt{r_0(y,z)}v_{K}^\top(x,y,z)H_K^{-1}\widetilde{v}_{K}(X_t,Y_t,Z_t).
		\end{equation*}
Then we have		\begin{equation}\label{nhatI}			2n\times\widehat{I}_{K,K_0}
			=\int\left\{\frac{1}{\sqrt{n}}\sum_{t=1}^{n}S_0(X_t, Y_t, Z_t;x,y,z)\right\}^2dF_{X,Y,Z}(x,y,z)+o_P(1)
		\end{equation}
and
\begin{small}
		\begin{equation}
			\begin{aligned}
				&\lim_{n\rightarrow \infty} Cov\left(\frac{1}{\sqrt{n}}\sum_{t=1}^{n}S_0(X_t, Y_t, Z_t;x,y,z),\frac{1}{\sqrt{n}}\sum_{t=1}^{n}S_0(X_t, Y_t, Z_t;x',y',z') \right) \\
				=&2\sqrt{r_0(y,z)}v_{K}^\top(x,y,z)H_K^{-1}\sum_{k=1}^{\infty }\mathbb{E}\left[\widetilde{v}_{K}(X_1,Y_1,Z_1)\widetilde{v}_{K}^\top(X_{1+k},Y_{1+k},Z_{1+k})\right]H_K^{-1}v_K(x',y',z')\sqrt{r_0(y',z')}\\
                &+\sqrt{r_0(y,z)}v_{K}^\top(x,y,z)H_K^{-1}\mathbb{E}\left[\widetilde{v}_{K}(X_t,Y_t,Z_t)\widetilde{v}_{K}^\top(X_t,Y_t,Z_t)\right]H_K^{-1}v_K(x',y',z')\sqrt{r_0(y',z')}\\
				:=&V_{K}((x,y,z),(x',y',z')).\label{cov_VK}
			\end{aligned}
		\end{equation}
        \end{small}
		 For every fixed $K$, it is easy to verify that $\{S_0(X_t, Y_t, Z_t;x,y,z):x \in \mathcal{X},  y \in \mathcal{Y}, z \in \mathcal{Z}\}$ is a uniformly bounded VC subgraph  \citep{kosorok2008introduction}  using the smoothness of $u_{K_0}(y,z)$ and $v_{K}(x,y,z)$ and the compactness of $ \mathcal{X}$, $ \mathcal{Y}$ and $ \mathcal{Z}$. Combining with $\beta_m=O(\rho ^m)$  for $0<\rho<1$ in Assumption \ref{ass:ciid}, conditions in Corollary 2.1 in \cite{arcones1994central} are satisfied,  then $n^{-1/2}\sum_{t=1}^nS_0(X_t, Y_t, Z_t;\cdot,%
		\cdot,\cdot)$ converges to a Gaussian process $	\mathbb{G}_0(x,y,z)$ with mean zero and
		covariance function $\{ V_{K}((x,y,z),(x^{\prime },y^{\prime },z^{\prime })):(x,y,z),
		(x^{\prime },y^{\prime },z^{\prime })\in \mathbb{R}^{d_X}\times \mathbb{R}^{d_Y} \times \mathbb{R}^{d_Z}\}$.
		Hence, 
		\begin{align*}
			2n \times \widehat{I}_{K,K_0}&=\int\Bigg\{\frac{1}{\sqrt{n}}%
			\sum_{t=1}^n	S_0(X_t, Y_t, Z_t;x,y,z)\Bigg\}^2dF_{X,Y,Z}(x,y,z) +o_P(1)\\
			& \xrightarrow{d}\int
			|\mathbb{G}_0(x,y,z)|^2dF_{X,Y,Z}(x,y,z).
		\end{align*}
		In Chapter 1, Section 2 of \cite{kuo1975gaussian}, $\int
		|	\mathbb{G}_0(x,y,z)|^2dF_{X,Y,Z}(x,y,z)$ has the following representation: 
		\begin{align*}
			\int |	\mathbb{G}_0(x,y,z)|^2dF_{X,Y,Z}(x,y,z)\overset{d}{=} \sum_{j=1}^\infty \lambda_j
			\chi^2_j(1),
		\end{align*}
		where $\chi^2_j(1)$s are independent chi-squared random variables with one
		degree of freedom, and the nonnegative constants, $\{\lambda_j\}$, are
		eigenvalues of the following equation: 
		\begin{align}  \label{eq:eigenvalues}
			\int V_{K}((x,y,z),(x^{\prime },y^{\prime
			},z^{\prime}))\phi(x,y,z)dF_{X,Y,Z}(x,y,z)=\lambda\phi(x^{\prime },y^{\prime },z^{\prime}).
		\end{align}
		
		\clearpage
    \section{Proof of Theorem \ref{th_cH0}}
		\label{app:th_cH0}
  To aid the presentation, we denote $O_t:=(X_t^\top,Y_t^\top,Z_t^\top)^\top$,  $o:=(x^\top,y^\top,z^\top)^\top$, $\mathcal{O}=\mathcal{X}\times\mathcal{Y}\times\mathcal{Z}$. 
		By (\ref{Ichat}), we have
		\begin{equation}
			\begin{aligned}
   \label{IU1}
				&2\widehat{I}_{K,K_0}-\mathbb{E}\left[\left\Vert\frac{1}{n}\sum_{t=1}^{n}H_K^{-\frac{1}{2}}\widetilde{v}_{K}(O_t)\right\Vert^2\right]\\
				=&\left\Vert\frac{1}{n}\sum_{t=1}^{n}H_K^{-\frac{1}{2}}\widetilde v_{K}(O_t)\right\Vert^2-\mathbb{E}\left[\left\Vert\frac{1}{n}\sum_{t=1}^{n}H_K^{-\frac{1}{2}}\widetilde{v}_{K}(O_t)\right\Vert^2\right]+o_P\left(\frac{\sqrt {K}}{n}\right) \\
				=&\frac{2}{n^2}\sum_{1\leq i<j\leq n}\left\{\widetilde V_{K}(O_i,O_j)-\mathbb{E}\left[\widetilde{V}_{K}(O_i,O_j)\right]\right\}\\
                &+\frac{1}{n^2}\sum_{t=1}^{n}\left\{\widetilde{V}_{K}(O_t,O_t)-\mathbb{E}\left[\widetilde{V}_{K}(O_t,O_t)\right]\right\}+o_P\left(\frac{\sqrt {K}}{n}\right),
			\end{aligned}
		\end{equation}
		where
		$
		\widetilde{V}_{K}(O_i,O_j)=\widetilde v_{K}^\top(O_i)H_K^{-1}\widetilde v_{K}(O_j)
		.$   Note that
  \begin{align*}
     & \mathbb{E}\left[\left\{\frac{1}{n}\sum_{t=1}^n\widetilde{V}_{K}(O_t,O_t)-\mathbb{E}\left[\widetilde{V}_{K}(O_t,O_t)\right]\right\}^2\right]\\
      =&\frac{1}{n^2}\sum_{t=1}^n\mathbb{E}\left[\left\{\widetilde{V}_{K}(O_t,O_t)-\mathbb{E}\left[\widetilde{V}_{K}(O_t,O_t)\right]\right\}^2\right]\\
      &+\frac{2}{n^2}\sum_{i<j}\mathbb{E}\left[\left\{\widetilde{V}_{K}(O_i,O_i)-\mathbb{E}\left[\widetilde{V}_{K}(O_i,O_i)\right]\right\}\left\{\widetilde{V}_{K}(O_j,O_j)-\mathbb{E}\left[\widetilde{V}_{K}(O_j,O_j)\right]\right\}\right]\\
      \leq&\frac{1}{n^2}\sum_{t=1}^n
\mathbb{E}\left[\widetilde{V}_{K}^2(O_t,O_t)\right]+\frac{2}{n^2}\sum_{t=1}^n\sum_{\tau=1}^{n-t}Cov\left(\widetilde{V}_{K}(O_t,O_t), \widetilde{V}_{K}(O_{t+\tau},O_{t+\tau})\right)\\
\leq&\frac{8}{n^2}\sum_{t=1}^n\sum_{\tau=1}^{n-t}\left\{\mathbb{E}\left[\widetilde{V}_{K}^4(O_t,O_t)\right]\right\}^{\frac{1}{2}}\sqrt{\beta_\tau}+O\left(\frac{K^2}{n}\right)\\
=&O\left(\frac{\xi_K^2K}{n}\right),
      \end{align*}
  where the second inequality comes from Assumption \ref{ass:eigen} (iii),  Lemma \ref{Davydov}, and
  \begin{equation*}
      \begin{aligned} \mathbb{E}\left[\widetilde{V}_{K}^2(O_t,O_t)\right]
  =&\mathbb{E}\left[ \widetilde{v}_{K}^\top(O_t)H_K^{-1}\widetilde{v}_{K}(O_t)\widetilde{v}_{K}^\top(O_t)H_K^{-1}\widetilde{v}_{K}(O_t)\right]\\
  \leq& \lambda_{\min}^{-2}\left( H_K\right)\cdot \tr\left\{\mathbb{E}\left[\widetilde{v}_{K}(O_t)\widetilde{v}_{K}^\top(O_t)\widetilde{v}_{K}(O_t)\widetilde{v}_{K}^\top(O_t)\right]\right\}\\
    =&O(1)\cdot\sum_{i=1}^{K}\sum_{j=1}^K
\mathbb{E}\left[\widetilde{v}_{K,i}^2(O_t)\widetilde{v}_{K,j}^2(O_t)\right] \\
\leq&O(1)\cdot\sum_{i=1}^{K}\sum_{j=1}^K
\left\{\mathbb{E}\left[\widetilde{v}_{K,i}^4(O_t)\right] \right\}^{\frac{1}{2}}\left\{\mathbb{E}\left[\widetilde{v}_{K,j}^4(O_t)\right] \right\}^{\frac{1}{2}}\\
=&O\left(K^2\right),
\end{aligned}
\end{equation*}
  and the last equality comes from  Assumption \ref{ass:ciid} and
  \begin{equation*}
\mathbb{E}\left[\widetilde{V}_{K}^4(O_t,O_t)\right]
\leq\sup_{o\in\mathcal{O}}|\widetilde v_{K}^\top(o)H_K^{-1}\widetilde v_{K}(o)|^2
\cdot\mathbb{E}\left[\widetilde{V}_{K}^2(O_t,O_t)\right]=O\left(\xi_K^4K^2\right).
  \end{equation*}
  Then we can obtain $\frac{1}{n}\sum_{t=1}^{n}\left\{\widetilde{V}_{K}(O_t,O_t)-\mathbb{E}\left[\widetilde{V}_{K}(O_t,O_t)\right]\right\}=O_P\left(\xi_K\sqrt{K/n}\right)$ and from (\ref{IU1})
  \begin{small}
		\begin{equation}
			\label{IU2}
			\begin{aligned}
				2\widehat{I}_{K,K_0}-\mathbb{E}\left[\left\Vert\frac{1}{n}\sum_{t=1}^{n}H_K^{-\frac{1}{2}}\widetilde{v}_{K}(O_t)\right\Vert^2\right]
				 =\frac{2}{n^2}\sum_{1\leq i<j\leq n}\left\{\widetilde V_{K}(O_i,O_j)-\mathbb{E}\left[\widetilde{V}_{K}(O_i,O_j)\right]\right\}+o_P\left(\frac{\sqrt {K}}{n}\right).
			\end{aligned}
		\end{equation}
        \end{small}
Define \begin{align*}
\sigma_{K}^2:=2\mathbb{E}\left[\widetilde{V}_{K}^2(O_0,\overline{O}_0)\right]=2\mathbb{E}\left[\{\widetilde v_{K}^\top(O_0)H_K^{-1}\widetilde v_{K}(\overline{O}_0)\}^2\right],
		\end{align*}
        and 
        \begin{align*}U_n=\frac{1}{n}\sum_{1\leq i<j\leq n} \left\{\frac{\widetilde{V}_{K}(O_i,O_j)}{\sigma_{K}}-\frac{\mathbb{E}[\widetilde{V}_{K}(O_i,O_j)]}{\sigma_{K}}\right\},
        \end{align*}
        where $\overline{O}_0$ is an independent copy of $O_0$.
       It is obvious that 
       $  \sigma_{K}^2=O(K)$. We verify the conditions  in Lemma \ref{T1997} below. 
For $1\leq i\leq n$, since
\begin{align*}
&\mathbb{E}\left[\widetilde{V}_{K}(O_i,O_0)^4\right]\\
\leq&\sup_{(o_0,o_i)\in\left(\mathcal{O}_0\times\mathcal{O}_i\right)}\left|\widetilde{v}_{K}^\top(o_i)H_K^{-1}\widetilde{v}_{K}(o_0)\right|^2\cdot\mathbb{E}\left[\left\{\widetilde{v}_{K}^\top(O_i)H_K^{-1}\widetilde{v}_{K}(O_0)\right\}^2\right]\\
\leq&\xi_K^4\ \lambda_{\min}^{-2}\left( H_K\right)\cdot \mathbb{E}\left[\widetilde{v}_{K}^\top(O_i)H_K^{-1}\widetilde{v}_{K}(O_0)\widetilde{v}_{K}^\top(O_0)H_K^{-1}\widetilde{v}_{K}(O_i)\right]\\
\leq&\xi_K^4\ \lambda_{\min}^{-4}\left( H_K\right)\cdot \sup_{(o_0,o_i)\in\left(\mathcal{O}_0\times\mathcal{O}_i\right)} \frac{f_{O_0,O_i}\left(o_0,o_i\right)}{f_{O_0}(o_0)f_{O_i}(o_i)}\cdot\tr\left\{\mathbb{E}\left[\widetilde{v}_{K}(O_i)\widetilde{v}_{K}^\top(O_i)\right]\mathbb{E}\left[\widetilde{v}_{K}(O_0)\widetilde{v}_{K}^\top(O_0)\right]\right\}\\
\leq&O\left(\xi_K^4\right)\cdot\lambda_{\max}\left(\mathbb{E}\left[\widetilde{v}_{K}(O_i)\widetilde{v}_{K}^\top(O_i)\right]\right)\cdot \tr\left(\mathbb{E}\left[\widetilde{v}_{K}(O_0)\widetilde{v}_{K}^\top(O_0)\right]\right)\\
=&O\left(\xi_K^4K\right),
\end{align*}
for $\delta_0>0$, we have
\begin{align*}
\mathbb{E}\left[\left\{\frac{\widetilde{V}_{K}(O_i,O_0)}{\sigma_{K}}\right\}^{4+\delta_0}\right]\leq O\left(\frac{\xi_K^{2\delta_0}}{K^{(4+\delta_0)/2}}\right)\cdot \mathbb{E}\left[\widetilde{V}_{K}(O_i,O_0)^4\right]=O\left(\frac{\xi_K^{4+2\delta_0}}{K^{(4+\delta_0)/2}}K \right), 
\end{align*}
and 
$u_n(4+\delta_0)=O\left( \xi_K^{\frac{4+2\delta_0}{4+\delta_0}}K^{\frac{1}{4+\delta_0}}/\sqrt{K}\right) =O\left(n^{\gamma_0} \right) $ with $\gamma_0=(2+\delta_0)/(8+2\delta_0)\in\left(0,1/2\right) $. Thus   condition (1) in Lemma \ref{T1997} is satisfied.
Under Assumption \ref{ass:fo_ratio}, we have
\begin{align*}
&\mathbb{E}\left[G_{n0}^2(O_i,O_0)\right]=\frac{1}{\sigma_{K}^4}\mathbb{E}\left[\left\{\widetilde{v}_{K}^\top(o_i)H_K^{-1}\mathbb{E}\left[\widetilde{v}_{K}(O_0)\widetilde{v}_{K}^\top(O_0)\right]H_K^{-1}\widetilde{v}_{K}(o_0)\right\}^2\right]\\
=&O\left(\frac{1}{K^2}\right)\int\left\{\widetilde{v}_{K}^\top(o_i)H_K^{-1}\mathbb{E}\left[\widetilde{v}_{K}(O_0)\widetilde{v}_{K}^\top(O_0)\right]H_K^{-1}\widetilde{v}_{K}(o_0)\right\}^2f_{\mathcal{O}_0, \mathcal{O}_i}(o_0,o_i)do_0do_i\\
\leq&O\left(\frac{1}{K^2}\right)\sup_{(o_0,o_i)\in\left(\mathcal{O}_0\times\mathcal{O}_i\right)} \frac{f_{O_0,O_i}\left(o_0,o_i\right)}{f_{O_0}(o_0)f_{O_i}(o_i)}\\
&\cdot  \int\left\{\widetilde{v}_{K}^\top(o_i)H_K^{-1}\mathbb{E}\left[\widetilde{v}_{K}(O_0)\widetilde{v}_{K}^\top(O_0)\right]H_K^{-1}\widetilde{v}_{K}(o_0)\right\}^2 f_{O_0}(o_0)f_{O_i}(o_i) do_0do_i\\
	=&O\left(\frac{1}{K^2}\right)\tr\bigg\{\mathbb{E}\left[\widetilde{v}_{K}(O_i)\widetilde{v}_{K}^\top(O_i)\right]H_K^{-1}\mathbb{E}\left[\widetilde{v}_{K}(O_0)\widetilde{v}_{K}^\top(O_0)\right]H_K^{-1}\mathbb{E}\left[\widetilde{v}_{K}(O_0)\widetilde{v}_{K}^\top(O_0)\right]\\
    &H_K^{-1}\mathbb{E}\left[\widetilde{v}_{K}(O_0)\widetilde{v}_{K}^\top(O_0)\right]H_K^{-1}\bigg\}\\
 \leq&O\left(\frac{1}{K^2}\right) \lambda_{\min}^{-4}\left( H_K\right)\cdot\lambda_{\max}^3\left(\mathbb{E}\left[\widetilde{v}_{K}(O_0)\widetilde{v}_{K}^\top(O_0)\right]\right) \cdot\tr\left\{\mathbb{E}\left[\widetilde{v}_{K}(O_i)\widetilde{v}_{K}^\top(O_i)\right]\right\}\\
	=&O\left(\frac{1}{K} \right) .
\end{align*}
Similarly, we can show 	$ \mathbb{E}\left[G_{n0}^2(O_0,\overline{O}_0)\right]=O\left(1/K\right)$, thus $v_n(2)=O\left(1/\sqrt{K} \right) =o(1)$ and   condition (2) in Lemma \ref{T1997} is satisfied.
In the same way, for condition (3), we have
	\begin{align*}	&\mathbb{E}\left[G_{n0}^{2+\delta_0/2}(O_0,O_0)\right]\\
	=&\frac{1}{\left(\sigma_K^2\right)^{2+\delta_0/2}}\cdot\sup_{o_0\in\mathcal{O}_0}\left|\widetilde{v}_{K}^\top(o_0)H_K^{-1}\mathbb{E}\left[\widetilde{v}_{K}(O_0)\widetilde{v}_{K}^\top(O_0)\right]H_K^{-1}\widetilde{v}_{K}(o_0)\right|^{\delta_0/2}\\
    &\cdot\mathbb{E}\left[\left\{\widetilde{v}_{K}^\top(O_0)H_K^{-1}\mathbb{E}\left[\widetilde{v}_{K}(O_0)\widetilde{v}_{K}^\top(O_0)\right]H_K^{-1}\widetilde{v}_{K}(O_0)\right\}^{2}\right]\\
	\leq&\frac{1}{\left(\sigma_K^2\right)^{2+\delta_0/2}}\cdot O\left(\left\{\xi_K^2K\right\}^{\delta_0/2}\right)\cdot\lambda_{\max}^2\left( \mathbb{E}\left[\widetilde{v}_{K}(O_0)\widetilde{v}_{K}^\top(O_0)\right]\right)\cdot\lambda_{\min}^{-4}\left( H_K\right)\\
    &\cdot\tr\left\{\mathbb{E}\left[\widetilde{v}_{K}(O_0)\widetilde{v}_{K}^\top(O_0)\widetilde{v}_{K}(O_0)\widetilde{v}_{K}^\top(O_0)\right]\right\}\\
=&O\left(\frac{\xi_K^{\delta_0}K^{2+\delta_0/2}}{\left(\sigma_K^2\right)^{2+\delta_0/2}}\right),
	\end{align*}
and  $ w_n(2+\delta_0/2)=O\left( \xi_K^{2\delta_0/(4+\delta_0)} \right) =o\left( \sqrt{n}\right).$
 Similarly, we have
		\begin{align*}
&\mathbb{E}\left[G_{nj}^2(O_i,O_0)\right]=\frac{1}{\sigma_K^4}\mathbb{E}\left[\left\{\widetilde{v}_{K}^\top(O_i)H_K^{-1}\mathbb{E}\left[\widetilde{v}_{K}(O_j)\widetilde{v}_{K}^\top(O_0)\right]H_K^{-1}\widetilde{v}_{K}(O_0)\right\}^2\right]\\
  =&O\left(\frac{1}{K^2}\right)\int\left\{\widetilde{v}_{K}^\top(o_i)H_K^{-1}\mathbb{E}\left[\widetilde{v}_{K}(O_j)\widetilde{v}_{K}^\top(O_0)\right]H_K^{-1}\widetilde{v}_{K}(o_0)\right\}^2f_{\mathcal{O}_0, \mathcal{O}_i}(o_0,o_i)do_0do_i\\
\leq&O\left(\frac{1}{K^2}\right)\sup_{(o_0,o_i)\in\left(\mathcal{O}_0\times\mathcal{O}_i\right)} \frac{f_{O_0,O_i}\left(o_0,o_i\right)}{f_{O_0}(o_0)f_{O_i}(o_i)} \int\left\{\widetilde{v}_{K}^\top(o_i)H_K^{-1}\mathbb{E}\left[\widetilde{v}_{K}(O_j)\widetilde{v}_{K}^\top(O_0)\right]H_K^{-1}\widetilde{v}_{K}(o_0)\right\}^2 \\&\cdot f_{O_0}(o_0)f_{O_i}(o_i) do_0do_i\\
  =&O\left(\frac{1}{K^2}\right)\tr\left\{\mathbb{E}\left[\widetilde{v}_{K}(O_i)\widetilde{v}_{K}^\top(O_i)\right]H_K^{-1}\mathbb{E}\left[\widetilde{v}_{K}(O_j)\widetilde{v}_{K}^\top(O_0)\right]H_K^{-1}\mathbb{E}\left[\widetilde{v}_{K}(O_0)\widetilde{v}_{K}^\top(O_0)\right]H_K^{-1}\right.\\
  &\left.\mathbb{E}\left[\widetilde{v}_{K}(O_0)\widetilde{v}_{K}^\top(O_j)\right]H_K^{-1}\right\}\\	
 \leq&O\left(\frac{1}{K^2}\right) \lambda_{\min}^{-4}\left( H_K\right)\cdot\lambda_{\max}^2\left(\mathbb{E}\left[\widetilde{v}_{K}(O_0)\widetilde{v}_{K}^\top(O_0)\right]\right) \cdot\tr\left\{\mathbb{E}\left[\widetilde{v}_{K}(O_0)\widetilde{v}_{K}^\top(O_j)\right]\mathbb{E}\left[\widetilde{v}_{K}(O_j)\widetilde{v}_{K}^\top(O_0)\right]\right\}
  \\
		=&O\left(\frac{1}{K}\right).
		\end{align*}
 Similarly, we can show $\mathbb{E}\left[G_{nj}^2(O_0,O_i)\right]=O\left(1/K\right)$ and  $\mathbb{E}\left[G_{nj}^2(O_0,\overline{O}_0)\right]=O\left(1/K\right)$. Then we have $z_n(2)=O\left( 1/\sqrt{K}\right)$ and $z_n(2)n^{\gamma_1}=O\left(n^{\gamma_1}/\sqrt{K}\right)=O(1)$  by Assumption \ref{ass:Kinf}; thus condition (4) in Lemma \ref{T1997} is satisfied.
	Therefore, from Lemma \ref{T1997} we have
  $ U_n\xrightarrow{d}\mathcal{N}(0,1/4)$.

    Next, we  show $\frac{1}{n}\sum_{1\leq i<j\leq n}\mathbb{E}\left[\widetilde V_{K}(O_i,O_j)\right]/\sigma_K=o(1).$
    \begin{comment}
     \begin{equation}\label{EVij}
    \begin{aligned}
          &\frac{1}{n\sigma_K}\sum_{i<j}\mathbb{E}\left[\widetilde V_{K}(O_i,O_j)\right]\\
          \leq &\frac{1}{n\sigma_K}\sum_{t=1}^n\sum_{\tau=1}^{n-t}\left|\mathbb{E}\left[\widetilde V_{K}(O_t,O_{t+\tau})\right]\right|\\
          \leq&\frac{1}{n\sigma_K}\sum_{t=1}^n\sum_{\tau=1}^{n-t}4M^{\frac{1}{1+\delta}}\beta_\tau^{\frac{\delta}{1+\delta}}\\
          =&O\left(\frac{1}{\sqrt{K}}\right)\cdot M^{\frac{1}{1+\delta}}\\
          =&o(1),
    \end{aligned}
    \end{equation}
    where $M=\max\left\{\int |\widetilde V_{K}(o_i,o_j)|^{1+\delta}dF(o_i,o_j), \int |\widetilde V_{K}(o_i,o_j)|^{1+\delta}dF(o_i)dF(o_j) \right\}.$
	\end{comment}
	Let $m=\left[L\log n\right]$ (the integer part of $L \log n$), where $L$ is a large positive constant so that $n^4\beta_m^{\delta/(1+\delta)}=o(1)$ for some $\delta >0$ by Assumption \ref{ass:ciid}.
          We consider two different cases: $j-i>m$ and $j-i\leq m$. For $j-i>m$, we have
	\begin{equation*}
		\begin{aligned}
			\frac{1}{n\sqrt{K}}\sum_{j-i>m}\mathbb{E}\left[\widetilde V_{K}(O_i,O_j)\right]&\leq\frac{1}{n\sqrt{K}}\sum_{j-i>m}4M^{1/(1+\delta)}\beta_m^{\delta/(1+\delta)}\\
		&	\leq\frac{1}{n\sqrt{K}}\sum_{j-i>m}4\xi_K^2\beta_m^{\delta/(1+\delta)}\\
		&	\leq 4n\xi_K^2\beta_m^{\delta/(1+\delta)}/\sqrt{K}\\
		&=o(1),
		\end{aligned}
	\end{equation*}
	where  $M=\max\left\{\int |\widetilde V_{K}(o_i,o_j)|^{1+\delta}dF(o_i,o_j), \int |\widetilde V_{K}(o_i,o_j)|^{1+\delta}dF(o_i)dF(o_j) \right\}$.  For $j-i\leq m$, under Assumption \ref{ass:Kinf} and \ref{ass:covK}, we have 
	\begin{equation*}
		\begin{aligned}
			\frac{1}{n\sqrt{K}}\sum_{0<j-i\leq m}\mathbb{E}\left[\widetilde V_{K}(O_i,O_j)\right]&\leq\frac{1}{n\sqrt{K}}\cdot nm\cdot\max_{i<j\leq i+m}\left|\mathbb{E}\left[\widetilde V_{K}(O_i,O_j)\right]\right|\\
			&=\frac{m}{\sqrt{K}}\max_{i<j\leq i+m}\left|\mathbb{E}\left[\widetilde v_{K}^\top(O_i)H_K^{-1}\widetilde v_{K}(O_j)\right]\right|\\
			&=O\left(\frac{m}{\sqrt{K}}\right) \\
           & =o(1).
		\end{aligned}
	\end{equation*}	
   Finally, combining with 
	\begin{equation*}
B_{K}:=\frac{1}{2n}\mathbb{E}\left[\widetilde{v}_{K}^\top(O)H_K^{-1}\widetilde{v}_{K}(O)\right],
		\end{equation*}
    and   (\ref{IU2}), we can obtain
		\begin{equation*}
		\frac{2n\{\widehat{I}_{K,K_0}-B_{K}\}}{\sigma_{K}}\xrightarrow{d}\mathcal{N}(0,1) \text{ under } H_0.
	\end{equation*}

 \begin{comment}
 	\begin{equation*}
		\begin{aligned}
			\widehat v_{K,K_0}(X_i,Y_i,Z_i)
			=&\frac{1}{n}\sum_{j=1}^{n}v_K(X_i,Y_j,Z_i)-\left\{\frac{1}{n}\sum_{i=1}^{n}v_K(X_i,Y_i,Z_i)u_{K_0}^\top(Y_i,Z_i)\right\}\left\{\frac{1}{n}\sum_{j=1}^nu_{K_0}(Y_j,Z_i)\right\}\\
   &+\left\{\frac{1}{n}\sum_{i=1}^{n}v_K(X_i,Y_i,Z_i)\widehat r_{K_0}(Y_i,Z_i)\right\}
			-v_K(X_i,Y_i,Z_i)\widehat r_{K_0}(Y_i,Z_i),
		\end{aligned}
	\end{equation*}
	\begin{equation*}
		\widehat{\sigma}_{K,K_0}^2=	\frac{1}{nK}\sum_{t=1}^{n}\widehat  v_{K,K_0}^\top(X_t,Y_t,Z_t)\widehat{H}_{K,K_0}^{-1}\widehat{\Omega}_{K,K_0} \widehat{H}_{K,K_0}^{-1}\widehat  v_{K,K_0}(X_t,Y_t,Z_t),
	\end{equation*}
\begin{equation*}
		\begin{aligned}
	\widehat{\Omega}_{K,K_0}=\frac{1}{n}\sum_{t=1}^{n}\widehat  v_{K,K_0}(X_t,Y_t,Z_t)\widehat v_{K,K_0}^\top(X_t,Y_t,Z_t),	
		\end{aligned}
	\end{equation*}
 \end{comment}
		\clearpage

		\section{Proof of Theorem $\ref{th_cH1n_fixK}$}
		\label{sec:th_cH1n_fixK}
	Under $H_{1n}$, we have
		\begin{comment}
			Denote
			\begin{equation*}
				\begin{aligned}
					\widetilde v_{K}(O_i)=&H_K^{-1}\left\{v_{K}(X_i,Y_i,Z_i)r_0(Y_i,Z_i)\pi_n(X_i,Y_i,Z_i)+\mathbb{E}[v_{K}(X,Y,Z)r_0(Y,Z)\pi_n(X,Y,Z)]\right.\\
					&\left.-\mathbb{E}[v_K(X_i,Y_i,Z_i)r_0(Y_i,Z_i)\pi_n(X_i,Y_i,Z_i)|X_i,Z_i]-\mathbb{E}[v_K(X_i,Y_i,Z_i)r_0(Y_i,Z_i)\pi_n(X_i,Y_i,Z_i)|Y_i]\right\}.\\	
				\end{aligned}
			\end{equation*}
		\end{comment}
\begin{small}
			\begin{equation}
            \begin{aligned}
			2	\widehat{I}_{K,K_0}
				=&\frac{1}{n}\sum_{t=1}^{n}\left\{\widehat \pi_K(X_t, Y_t, Z_t)-1\right\}^2\widehat{r}_{K_0}(Y_t,Z_t) \\
				=&\frac{1}{n}\sum_{t=1}^{n}\left\{\widehat \pi_K(X_t, Y_t, Z_t)-1\right\}^2r_0(Y_t,Z_t)\\
                &+\frac{1}{n}\sum_{t=1}^{n}\left\{\widehat \pi_K(X_t, Y_t, Z_t)-1\right\}^2\left( \widehat{r}_{K_0}(Y_t,Z_t)-r_0(Y_t,Z_t)\right)  \\
				=&\frac{1}{n}\sum_{t=1}^{n}\left\{\widehat\pi_K(X_t, Y_t, Z_t)-\pi_K^*(X_t, Y_t, Z_t)+\pi_K^*(X_t, Y_t, Z_t)-1\right\}^2r_0(Y_t,Z_t) \\
				&+\frac{1}{n}\sum_{t=1}^{n}\left\{\widehat \pi_K(X_t, Y_t, Z_t)-\pi_K^*(X_t,Y_t,Z_t)+\pi_K^*(X_t,Y_t,Z_t)-1\right\}^2\left( \widehat{r}_{K_0}(Y_t,Z_t)-r_0(Y_t,Z_t)\right)  \\
				=&\frac{1}{n}\sum_{t=1}^{n}\left\{\widehat\pi_K(X_t, Y_t, Z_t)-\pi_K^*(X_t, Y_t, Z_t)+\pi_K^*(X_t, Y_t, Z_t)-1\right\}^2r_0(Y_t,Z_t) \\
				&+O_P\left( \xi_K\left\{\frac{K_0}{n}+K_0^{-2\omega_{r}}+\frac{K}{n}\right\}\left\{K_0^{-\omega_{r}}+\sqrt{\frac{K_0}{n}}\right\}\right) +O_P\left(d_n^2\left\{K_0^{-\omega_{r}}+\sqrt{\frac{K_0}{n}}\right\} \right) \\
					=&\frac{1}{n}\sum_{t=1}^{n}\left\{\widehat\pi_K(X_t, Y_t, Z_t)-\pi_K^*(X_t, Y_t, Z_t)+\pi_K^*(X_t, Y_t, Z_t)-1\right\}^2r_0(Y_t,Z_t) +o_P\left(\frac{\sqrt{K}}{n} \right) 
				\label{nHat_H1n}
                \end{aligned}
			\end{equation}
            \end{small}
	where the third equality comes from
 \begin{equation}\label{dpi_dr}
	\begin{aligned}
	&\frac{1}{n}\sum_{t=1}^{n}\left\{\widehat \pi_K(X_t, Y_t, Z_t)-\pi_K^*(X_t,Y_t,Z_t)\right\}^2\left( \widehat{r}_{K_0}(Y_t,Z_t)-r_0(Y_t,Z_t)\right)\\
	=& 	\left( \widehat{\boldsymbol{\beta}}_{K,K_0}-\boldsymbol{\beta}^*_K\right) ^\top\frac{1}{n}\sum_{t=1}^{n}\left( \widehat{r}_{K_0}(Y_t,Z_t)-r_0(Y_t,Z_t)\right)v_K(X_t,Y_t,Z_t)v_K^\top(X_t,Y_t,Z_t)	\left( \widehat{\boldsymbol{\beta}}_{K,K_0}-\boldsymbol{\beta}^*_K\right)\\
	\leq&\left\Vert\frac{1}{n}\sum_{t=1}^{n}\left( \widehat{r}_{K_0}(Y_t,Z_t)-r_0(Y_t,Z_t)\right)v_K(X_t,Y_t,Z_t)v_K^\top(X_t,Y_t,Z_t)\right\Vert\cdot	\left\Vert \widehat{\boldsymbol{\beta}}_{K,K_0}-\boldsymbol{\beta}^*_K\right\Vert^2\\
	=&O_P\left( \xi_K\left\{\frac{K_0}{n}+K_0^{-2\omega_{r}}+\frac{K}{n}\right\}\left\{K_0^{-\omega_{r}}+\sqrt{\frac{K_0}{n}}\right\}\right)\ \text{(by (\ref{Hhat-H}) and (\ref{prop_beta}))}
	\end{aligned}
 \end{equation}
and
\begin{align*}
&\frac{1}{n}\sum_{t=1}^{n}\left\{\pi_K^*(X_t,Y_t,Z_t)-1\right\}^2\left( \widehat{r}_{K_0}(Y_t,Z_t)-r_0(Y_t,Z_t)\right)\\
\leq&\sup_{(x,y,z)\in\mathcal{X}\times\mathcal{Y}\times \mathcal{Z}}|\pi_K^*(X,Y,Z)-1|^2
\cdot\frac{1}{n}\sum_{t=1}^{n}\left( \widehat{r}_{K_0}(Y_t,Z_t)-r_0(Y_t,Z_t)\right)\\
=&O_P\left(d_n^2\left\{K_0^{-\omega_{r}}+\sqrt{\frac{K_0}{n}}\right\} \right) .
\end{align*}
Under $H_{1n}$, using Lemma \ref{prop:vtilde} with $\pi_0$ replaced by $ \pi_n$  , we get
		\begin{equation}	
			\begin{aligned}
\widehat{\boldsymbol{\beta}}_{K,K_0}-\boldsymbol{\beta}^*_K
				=&H_K^{-1}\cdot\frac{1}{n}\sum_{t=1}^{n}\widetilde{v}_K(X_t,Y_t,Z_t)+o_P\left(\frac{K^{1/4}}{\sqrt{n}}\right),
			\end{aligned}
		\end{equation}
where
    \begin{equation}
    \begin{aligned}\label{vtilde_H1n}
\widetilde{v}_K(X_t,Y_t,Z_t)=&\mathbb{E}\left[v_K(X_t,Y_t,Z_t)r_0(Y_t,Z_t)\pi_{n}(X_t,Y_t,Z_t)|X_t,Z_t\right]\\
&+\mathbb{E}\left[v_K(X_t,Y_t,Z_t)r_0(Y_t,Z_t)\pi_{n}(X_t,Y_t,Z_t)|Y_t\right]\\	&-\mathbb{E}\left[v_K(X_t,Y_t,Z_t)r_0(Y_t,Z_t)\pi_K^*(X_t,Y_t,Z_t)|Y_t\right]\\
&-\mathbb{E}\left[v_K(X_t,Y_t,Z_t)r_0(Y_t,Z_t)\pi_K^*(X_t,Y_t,Z_t)|Z_t\right]\\
	&+\mathbb{E}[v_{K}(X,Y,Z)r_0(Y,Z)\pi_K^*(X,Y,Z)]
	\\
	&	-v_{K}(X_t,Y_t,Z_t)r_0(Y_t,Z_t)\pi_K^*(X_t,Y_t,Z_t).	
    \end{aligned}
	\end{equation}
	Define
    \begin{align*}
        S(X_t,Y_t,Z_t;x,y,z)
=\sqrt{r_0(y,z)}v_{K}^\top(x,y,z)H_K^{-1}\widetilde{v}_K(X_t,Y_t,Z_t).
    \end{align*}
		  For a fix $K$, 	similar to (\ref{nhatI}), we have 
        \begin{small}
            \begin{align*}
	2n\times\widehat{I}_{K,K_0}
			=&\int\left\{\frac{1}{\sqrt{n}}\sum_{t=1}^{n}S(X_t, Y_t, Z_t;x,y,z)+\sqrt{n}\sqrt{r_0(y,z)}\left(\pi_K^*(x,y,z)-1\right)\right\}^2dF_{X,Y,Z}(x,y,z)\\
            &+o_P(1).
		\end{align*}
        \end{small}
		
		Then, similar to the
		proof of Theorem \ref{th_cH0_fixk} , we have the following weak
		convergence result 
		\begin{align}  \label{eq:StoG}
			\frac{1}{\sqrt{n}} \sum_{t=1}^nS(X_t, Y_t, Z_t;\cdot,\cdot,\cdot)\xRightarrow{w}%
			\mathbb{G}(\cdot,\cdot,\cdot),
		\end{align}
		where $	\mathbb{G}(\cdot)$ is a Gaussian process with mean zero and covariance function $%
		\{ V_{K}((x,y,z),\\(x^{\prime },y^{\prime },z^{\prime })):(x,y,z), (x^{\prime },y^{\prime },z^{\prime })\in 
		\mathbb{R}^{d_X}\times \mathbb{R}^{d_Y} \times \mathbb{R}^{d_Z}\}$ defined 
 in (\ref{cov_VK}) in which $\widetilde{v}_K(X_t,Y_t,Z_t)$ is replaced by (\ref{vtilde_H1n}).
		
		Next, 	let $S_{K}$ denote the linear space spanned by $v_K(x,y,z)$ and consider the weighted least-square projection (w.r.t. the norm $L^2(r_0(y,z)dF_{X,Y,Z}(x,y,z))$) of a function $\phi(x,y,z)$ on $S_{K}$, i.e. $\text{proj}^{wls}_{v_K}\phi(x,y,z)=v_K^\top(x,y,z)\boldsymbol{\beta}^{wls} $ and
		\begin{equation*}
			\boldsymbol{\beta}^{wls}:=\arg\min_{\boldsymbol{\beta} \in\mathbb{R}^K}\mathbb{E}\left[r_0(Y,Z)\left\{\phi (X,Y,Z)-v_K^\top(X,Y,Z)\boldsymbol{\beta}\right\}^2\right].
		\end{equation*}
		Then we have 
		\begin{equation*}
		    \begin{aligned}
		        &\text{proj}^{wls}_{v_K}\phi(X,Y,Z)\\
			=&\mathbb{E}\left[r_0(Y,Z)\phi(X,Y,Z)v_K^\top(X,Y,Z)\right]\left\{\mathbb{E}\left[r_0(Y,Z)v_K(X,Y,Z)v_K^\top(X,Y,Z)\right]\right\}^{-1}v_K(X,Y,Z).
		    \end{aligned}
		\end{equation*}
 From (\ref{proj}), we have
		$
	\text{proj}^{wls}_{v_K}\pi_K^*(X,Y,Z)=\text{proj}^{wls}_{v_K}\pi_n(X,Y,Z)$,
		and
		\begin{equation*}
			\begin{aligned}
				&\sqrt{n}\left(\pi_{K}^*(X,Y,Z)-1 \right)=\sqrt{n}\left(v_K^\top(X,Y,Z)\boldsymbol{\beta}_K^*-1 \right) =\sqrt{n}\left(v_K^\top(X,Y,Z)H_K^{-1}h_K-1 \right)\\
				=&\sqrt{n}\ \text{proj}^{wls}_{v_K}\pi_n(X,Y,Z)-\sqrt{n}
				=\sqrt{n}\ \text{proj}^{wls}_{v_K}\left( 1+d_n\Delta(X,Y,Z)\right) -\sqrt{n} 
				=\text{proj}^{wls}_{v_K}\Delta(X,Y,Z).
			\end{aligned}
		\end{equation*}
		\begin{comment}
			\begin{equation*}
				\begin{aligned}
					&Proj^{wls}\left\{\sqrt{n}\left(\pi_{K}^*(X,Y,Z)-1 \right)\right\}\\
					=&\mathbb{E}\left[\sqrt{n}r_0(Y,Z)\left(\pi_{K}^*(X,Y,Z)-1 \right)v_K^\top(X,Y,Z)\right]H_K^{-1}v_K(X,Y,Z)\\
					=&\mathbb{E}\left[\sqrt{n}r_0(Y,Z)\left\{\pi_{K}^*(X,Y,Z)-\pi_n(X,Y,Z)+d_n\Delta(X,Y,Z) \right\}v_K^\top(X,Y,Z)\right]H_K^{-1}v_K(X,Y,Z)\\
					=&\mathbb{E}\left[\sqrt{n}r_0(Y,Z)\left(\pi_{K}^*(X,Y,Z)-\pi_n(X,Y,Z) \right)v_K^\top(X,Y,Z)\right]H_K^{-1}v_K(X,Y,Z)\\
					&+\mathbb{E}\left[r_0(Y,Z)\Delta(X,Y,Z)v_K^\top(X,Y,Z)\right]H_K^{-1}v_K(X,Y,Z)\\
					=&\sqrt{n}Proj^{wls}\pi_{K}^*(X,Y,Z) -\sqrt{n}Proj^{wls}\pi_n(X,Y,Z) +Proj^{wls}\Delta(X,Y,Z)\\
					=&Proj^{wls}\Delta(X,Y,Z)\\
				\end{aligned}
			\end{equation*}
			content...
		\end{comment}
		Combining with (\ref{nHat_H1n}) and (\ref{eq:StoG}),  we can obtain
		\begin{equation*}
			2n\times \widehat I_{K,K_0}\xrightarrow{d}\int\left\{\mathbb{G}(x,y,z)+\sqrt{r_0(y,z)}\ \text{proj}^{wls}_{v_K}\Delta(x,y,z)\right\}^2dF_{X,Y,Z}(x,y,z)\text{ under }  H_{1n}.
		\end{equation*}
		
		\clearpage
		\section{Proof of Theorem $\ref{th_cH1n}$}
		\label{sec:th_cH1n}
	By 	\eqref{nHat_H1n}, we have
		\begin{equation}\label{IhatQ}
			\begin{aligned}
			2	\widehat{I}_{K,K_0}=&\frac{1}{n}\sum_{t=1}^{n}\{\widehat{\pi}_{K,K_0}(X_t, Y_t, Z_t)-1\}^2 \widehat{r}_{K_0}(Y_t,Z_t) \\
				=&\frac{1}{n}\sum_{t=1}^{n}\left\{\widehat{\pi}_{K,K_0}(X_t, Y_t, Z_t)-\pi_K^*(X_t, Y_t, Z_t)\right\}^2r_0(Y_t,Z_t)\\
				&+\frac{2}{n}\sum_{t=1}^{n}\{\widehat\pi_{K,K_0}(X_t, Y_t, Z_t)-\pi_K^*(X_t, Y_t, Z_t)\}\{\pi_K^*(X_t, Y_t, Z_t)-1\}r_0(Y_t,Z_t)\\
				&+\frac{1}{n}\sum_{t=1}^{n}\{\pi_K^*(X_t, Y_t, Z_t)-1\}^2r_0(Y_t,Z_t)+o_P\left(\frac{\sqrt{K}}{n} \right) \\
				\equiv& Q_{1n}+Q_{2n}+Q_{3n}+o_P\left(\frac{\sqrt{K}}{n} \right) 
			\end{aligned}
		\end{equation}
		where the definitions of $Q_{1n}$, $Q_{2n}$, and $Q_{3n}$ are obvious. Similar to the proof of Theorem \ref{th_cH0},  we can show that
	\begin{equation*}
		\frac{n\{Q_{1n}-2B_{K}\}}{\sigma_{K}}\xrightarrow{d}\mathcal{N}(0,1).
	\end{equation*}

		For $Q_{2n}$,   by Lemma \ref{lemma_inf} with $\phi(x,y,z)=\pi_K^*(x,y,z)-1$,  since $\mathbb{E}\left[\phi_K(X_t, Y_t, Z_t)\right]=0$, we have
	\begin{equation*}
		\begin{aligned}
	&	\mathbb{E}\left[\left|\frac{1}{n}\sum_{t=1}^{n}\phi_K(X_t, Y_t, Z_t)\right|^2\right]\\
	=&\frac{1}{n}\mathbb{E}\left[|\phi_K(X_t, Y_t, Z_t)|^2\right]+\frac{2}{n^2}\sum_{1\leq i<j\leq n}\mathbb{E}\left[\phi_K(X_i, Y_i, Z_i)\phi_K(X_j,Y_j,Z_j)\right]\\
=&\frac{1}{n}\mathbb{E}\left[|\phi_K(X_t, Y_t, Z_t)|^2\right]+\frac{2}{n^2}\sum_{t=1}^{n}\sum_{\tau =1}^{n-t}Cov\left(\phi_K(X_t, Y_t, Z_t), \phi_K(X_{t+\tau},Y_{t+\tau},Z_{t+\tau})\right)\\
\leq&\frac{1}{n}\mathbb{E}\left[|\phi_K(X_t, Y_t, Z_t)|^2\right]+\frac{8}{n^2}\sum_{t=1}^{n}\sum_{\tau =1}^{n-t}\sup_{(x,y,z)\in\mathcal{X}\times\mathcal{Y}\times \mathcal{Z}}|\phi_{K}(x,y,z)|^2\beta_{\tau}\\
\leq&O(1)\cdot\frac{1}{n}\mathbb{E}\left[|\phi(X_t, Y_t, Z_t)|^2\right]+O(1)\cdot\frac{1}{n}\sup_{(x,y,z)\in\mathcal{X}\times\mathcal{Y}\times \mathcal{Z}}|\phi(x,y,z)|^2\\
=&O\left( \frac{1}{n}\right) \cdot\sup_{(x,y,z)\in\mathcal{X}\times\mathcal{Y}\times \mathcal{Z}}|\phi(x,y,z)|^2
		\end{aligned}
	\end{equation*}
		and $$\frac{1}{n}\sum_{t=1}^{n}\phi_K(X_t, Y_t, Z_t)=O_P\left(\frac{1}{\sqrt{n}}\sup_{(x,y,z)\in\mathcal{X}\times\mathcal{Y}\times \mathcal{Z}}|\phi(x,y,z)| \right).$$ 
		Note that under $H_{1n}:  \pi_n(X,Y,Z)= 1+d_n\cdot\Delta(X,Y,Z)$, 
		\begin{align*}
			\mathbb{E}\left[\left\{\pi_K^*(X,Y,Z)-1\right\}^2\right]\leq&2\cdot\mathbb{E}\left[\left\{\pi_K^*(X,Y,Z)-\pi_n(X,Y,Z)\right\}^2\right]+2\cdot d_n^2\mathbb{E}\left[\Delta^2(X,Y,Z)\right]\\
			=&O\left(d_n^2\epsilon_K^2\right)+O\left(d_n^2\right) =O\left( d_n^2\right),
		\end{align*}
	and 	
	\begin{equation*}
	    \begin{aligned}
&\sup_{(x,y,z)\in\mathcal{X}\times\mathcal{Y}\times \mathcal{Z}}|\pi_K^*(X,Y,Z)-1|\\
\leq&\sup_{(x,y,z)\in\mathcal{X}\times\mathcal{Y}\times \mathcal{Z}}|\pi_K^*(X,Y,Z)-\pi_n(X,Y,Z)|+\sup_{(x,y,z)\in\mathcal{X}\times\mathcal{Y}\times \mathcal{Z}} |d_n\Delta(x,y,z)|\\
=&O\left(d_n\epsilon_K\right)+O\left(d_n\right) =O\left( d_n\right),
	    \end{aligned}
	\end{equation*}
		then  under Assumptions \ref{ass:cK}-\ref{ass:proj_error}  with $d_n=\sigma_{K}^{1/2}/n^{1/2}\asymp {K}^{1/4}/n^{1/2}$, we have
		\begin{align*}
				Q_{2n}=&\frac{2}{n}\sum_{t=1}^{n}\{\widehat\pi_{K,K_0}(X_t, Y_t, Z_t)-\pi_K^*(X_t, Y_t, Z_t)\}\{\pi_K^*(X_t, Y_t, Z_t)-1\}r_0(Y_t,Z_t)\\
			=&O_P\left(\left\{K_0^{-\omega_{r}}+\frac{\sqrt{\xi_K}K^{1/4}}{n^{(2-\epsilon)/2}}+\frac{\sqrt{\zeta_{K_0}}K_0^{3/4}}{n^{(2-\epsilon)/2}}+\frac{\xi_{K}K_0}{n}+ \frac{\xi_{K}K}{n}+\sqrt{K}K_0^{-\omega_0}+\frac{1}{\sqrt{n}} \right\}\cdot d_n\right)\\
            &+O_P\left(\sqrt{\frac{K}{n}}\left\{\sqrt{\frac{K_0}{n}}+K_0^{-\omega_r}+\sqrt{\frac{K}{n}}\right\} \cdot d_n\right)\\
			=&o_P\left(\frac{\sqrt{K}}{n} \right). 
		\end{align*}
		For $Q_{3n}$, with $d_n=\sigma_{K}^{1/2}/n^{1/2}\asymp {K}^{1/4}/n^{1/2}$, we have
			\begin{align*}
				Q_{3n}&=\frac{1}{n}\sum_{t=1}^{n}\left\{\pi_K^*(X_t, Y_t, Z_t)-1\right\}^2r_0(Y_t,Z_t)\\
				&=\frac{1}{n}\sum_{t=1}^{n}\left\{\pi_K^*(X_t, Y_t, Z_t)-\pi_n(X_t, Y_t, Z_t)+d_n \Delta (X_t, Y_t, Z_t)\right\}^2r_0(Y_t,Z_t)\\
				&=\frac{1}{n}\sum_{t=1}^{n}\left\{\pi_K^*(X_t, Y_t, Z_t)-\pi_n(X_t, Y_t, Z_t)\right\}^2r_0(Y_t,Z_t)\\
				&+\frac{2d_n}{n}\sum_{t=1}^{n}\left[\pi_K^*(X_t, Y_t, Z_t)-\pi_n(X_t, Y_t, Z_t)\right]\Delta(X_t, Y_t, Z_t)r_0(Y_t,Z_t)\\&+\frac{d_n^2}{n}\sum_{t=1}^{n}\Delta^2(X_t, Y_t, Z_t)r_0(Y_t,Z_t)\\
				&=O_p\left(d_n^2\varepsilon_K^2\right)+O_p\left(d_n^2\varepsilon_K\right)+d_n^2\mathbb{E}[\Delta^2(X,Y,Z)r_0(Y,Z)]+o_p\left(d_n^2 \right)\\
				&=\frac{\sigma_{K}}{n}\mathbb{E}[\Delta^2(X,Y,Z)r_0(Y,Z)]+O_p\left(d_n^2\varepsilon_K\right)+o_p\left(d_n^2 \right)\\
				&=\frac{\sigma_{K}}{n}\mathbb{E}[\Delta^2(X,Y,Z)r_0(Y,Z)]+o_p\left(\frac{\sqrt{K}}{n}\right).
			\end{align*}
		Therefore, by  \eqref{IhatQ} we have
		\begin{equation*}
			\begin{aligned}
				\frac{2n}{ \sigma_{K}}(\widehat{I}_{K,K_0}- B_{K})\xrightarrow{d}\mathbb{E}[\Delta^2(X,Y,Z)r_0(Y,Z)]+\mathcal N\left( 0,1\right). 
			\end{aligned}
		\end{equation*}

		\clearpage
		\section{Proof of Theorem 	\ref{thm:cglobal_fixK}}
		\label{app:cglobal_fixK}
		Similar with	\eqref{nHat_H1n}, we have
			\begin{align*}
			&2	\widehat{I}_{K,K_0}=\frac{1}{n}\sum_{t=1}^{n}\{\widehat{\pi}_{K,K_0}(X_t, Y_t, Z_t)-1\}^2 \widehat{r}_{K_0}(Y_t,Z_t) \\
   =&\frac{1}{n}\sum_{t=1}^{n}\{\widehat{\pi}_{K,K_0}(X_t, Y_t, Z_t)-\pi_K^*(X_t, Y_t, Z_t)+\pi_K^*(X_t, Y_t, Z_t)-1\}^2 \left\{ \widehat{r}_{K_0}(Y_t,Z_t) -r_0(Y_t,Z_t)\right\}	 \\
   &+\frac{1}{n}\sum_{t=1}^{n}\{\widehat{\pi}_{K,K_0}(X_t, Y_t, Z_t)-\pi_K^*(X_t, Y_t, Z_t)+\pi_K^*(X_t, Y_t, Z_t)-1\}^2 r_0(Y_t,Z_t) \\
				=&\frac{1}{n}\sum_{t=1}^{n}\left\{\widehat{\pi}_{K,K_0}(X_t, Y_t, Z_t)-\pi_K^*(X_t, Y_t, Z_t)\right\}^2r_0(Y_t,Z_t)\\
				&+\frac{2}{n}\sum_{t=1}^{n}\{\widehat\pi_{K,K_0}(X_t, Y_t, Z_t)-\pi_K^*(X_t, Y_t, Z_t)\}\{\pi_K^*(X_t, Y_t, Z_t)-1\}r_0(Y_t,Z_t)\\
				&+\frac{1}{n}\sum_{t=1}^{n}\{\pi_K^*(X_t, Y_t, Z_t)-1\}^2r_0(Y_t,Z_t)\\
    &+\frac{1}{n}\sum_{t=1}^{n}\{\pi_K^*(X_t, Y_t, Z_t)-1\}^2\left\{ \widehat{r}_{K_0}(Y_t,Z_t) -r_0(Y_t,Z_t)\right\}	\\
    &+ \frac{2}{n}\sum_{t=1}^n\{\widehat\pi_{K,K_0}(X_t, Y_t, Z_t)-\pi_K^*(X_t, Y_t, Z_t)\}\{\pi_K^*(X_t, Y_t, Z_t)-1\}\left\{ \widehat{r}_{K_0}(Y_t,Z_t) -r_0(Y_t,Z_t)\right\}\\
    &+O_P\left( \xi_K\left\{\frac{K_0}{n}+K_0^{-2\omega_{r}}+\frac{K}{n}\right\}\left\{K_0^{-\omega_{r}}+\sqrt{\frac{K_0}{n}}\right\}\right)\\
				=& \widehat I_{1K}+ \widehat I_{2K}+ \widehat I_{3K}+\widehat I_{4K}+\widehat I_{5K}+o_P\left(\frac{\sqrt{K}}{n} \right),
			\end{align*}
		where the third equality comes from (\ref{dpi_dr}) and the definitions of $\widehat I_{1K}$, $\widehat I_{2K}$, $\widehat I_{3K}$  $\widehat I_{4K}$ and $\widehat I_{5K}$ are obvious.
		\begin{itemize}
			\item If $\mathbb{P}\left(\pi _{K}^{\ast }(X,Y,Z)=1\right) =1,$ we have $ \widehat I_{2K}= \widehat I_{3K}=\widehat I_{4K}=\widehat I_{5K}=0$ and
			\begin{equation*}
				\begin{aligned}
				2	\widehat{I}_{K,K_0}=\frac{1}{n}\sum_{t=1}^{n}\{ \widehat \pi_{K,K_0}(X_t, Y_t, Z_t)- \pi_{K}^*(X_t, Y_t, Z_t)\}^2r_0(Y_t,Z_t).
				\end{aligned}
			\end{equation*}
			Similar to (\ref{nhatI}), under $H_1$  and $\pi_K^*(x,y,z)\equiv
			1$ a.s., for a fixed $K$, we  have
			\begin{equation*}
				\begin{aligned}
					2n\times \widehat{I}_{K,K_0}	&=\int\left\{\frac{1}{\sqrt{n}}\sum_{t=1}^{n}S_1(X_t, Y_t, Z_t;x,y,z)\right\}^2dF_{X,Y,Z}(x,y,z)+o_P(1)\\
					&\xrightarrow{d} \int |\mathbb{G}_1(x,y,z)|^2dF_{X,Y,Z}(x,y,z), 
				\end{aligned}
			\end{equation*}
			where 
				\begin{equation*}
				\begin{aligned}
					S_1(X_t,Y_t,Z_t;x,y,z)
=\sqrt{r_0(y,z)}v_{K}^\top(x,y,z)H_K^{-1}\widetilde{v}_K(X_t,Y_t,Z_t),
				\end{aligned}
			\end{equation*}
	with
			\begin{equation}
				\begin{aligned}\label{vtilde_H1}
			\widetilde v_{K}(X_t,Y_t,Z_t)=&\mathbb{E}\left[v_K(X_t,Y_t,Z_t)r_0(Y_t,Z_t)\pi_{0}(X_t,Y_t,Z_t)|X_t,Z_t\right]\\
   &+\mathbb{E}\left[v_K(X_t,Y_t,Z_t)r_0(Y_t,Z_t)\pi_{0}(X_t,Y_t,Z_t)|Y_t\right]\\	&-\mathbb{E}\left[v_K(X_t,Y_t,Z_t)r_0(Y_t,Z_t)\pi_K^*(X_t,Y_t,Z_t)|Y_t\right]\\
   &-\mathbb{E}\left[v_K(X_t,Y_t,Z_t)r_0(Y_t,Z_t)\pi_K^*(X_t,Y_t,Z_t)|Z_t\right]\\
			&+\mathbb{E}[v_{K}(X,Y,Z)r_0(Y,Z)\pi_K^*(X,Y,Z)]
			\\
			&-v_{K}(X_t,Y_t,Z_t)r_0(Y_t,Z_t)\pi_K^*(X_t,Y_t,Z_t),
	\end{aligned}
			\end{equation}
			and $\mathbb{G}_1(x,y,z)$ is a Gaussian process with mean zero and covariance function $%
			\{ V_{K}((x,y,z),\\(x^{\prime },y^{\prime },z^{\prime })):(x,y,z), (x^{\prime },y^{\prime },z^{\prime })\in 
			\mathbb{R}^{d_X}\times \mathbb{R}^{d_Y} \times \mathbb{R}^{d_Z}\}$  defined 
 in (\ref{cov_VK}) in which $\widetilde{v}_K(X_t,Y_t,Z_t)$ is replaced by (\ref{vtilde_H1}).
			
			\item If $\mathbb{P}\left(\pi _{K}^{\ast }(X,Y,Z)=1\right) < 1$,   since $\beta_m=O(\rho ^m)$ under Assumption \ref{ass:ciid},  using Corollary 5.1 in \cite{hall2014martingale}, we have
            \begin{small}
			\begin{align*}
				&\sqrt{n}\left\{\widehat{I}_{3K}-\mathbb{E}\left[\{ \pi_{K}^*(X,Y,Z)-1\}^2r_0(Y,Z)\right]
				\right\}\xrightarrow{d}\mathcal N\left(0,Var\left( \{ \pi_{K}^*(X,Y,Z)-1\}^2r_0(Y,Z)\right) \right).
			\end{align*}
        \end{small}
			With Lemma \ref{lemma_inf}, since $\text{proj}^{wls}_{v_K}\left\{\pi_K^*(X_t,Y_t,Z_t)-1\right\}=\pi_K^*(X_t,Y_t,Z_t)-1$, we have
				\begin{align*}
					\widehat I_{2K}=&\frac{2}{n}\sum_{t=1}^{n}\{ \pi_{K}^*(X_t, Y_t, Z_t)-1\}\{ \widehat \pi_{K,K_0}(X_t, Y_t, Z_t)- \pi_{K}^*(X_t, Y_t, Z_t)\}r_0(Y_t,Z_t)
					\\ =&\frac{2}{n}\sum_{t=1}^{n}\phi_{1,K}(X_t, Y_t, Z_t)\\
            &+O_P\left(\sqrt{\frac{K}{n}}\left\{\sqrt{\frac{K_0}{n}}+K_0^{-\omega_r}+\sqrt{\frac{K}{n}}\right\}\left\{\mathbb{E}\left[|\pi_K^*(X_t,Y_t,Z_t)-1|^{4+\eta}\right]\right\}^{\frac{1}{4+\eta}}\right)
				\\&+O_P\left(\left\{K_0^{-\omega_{r}}+\frac{\sqrt{\xi_K}K^{1/4}}{n^{(2-\epsilon)/2}}+\frac{\sqrt{\zeta_{K_0}}K_0^{3/4}}{n^{(2-\epsilon)/2}}+\frac{\xi_{K}K_0}{n}+ \frac{\xi_{K}K}{n}+\sqrt{K}K_0^{-\omega_0}\right\}\right.\\
                &\qquad\qquad \left.\times  \sqrt{\mathbb{E}\left[\left\{\pi_K^*(X,Y,Z)-1\right\}^2\right]}\right),
				\end{align*}
		where
			\begin{align*}
				\phi_{1,K}(X_t,Y_t,Z_t)
    =&\mathbb{E}\left[r_0(Y_t,Z_t)\pi_{0}(X_t,Y_t,Z_t)\left\{\pi_K^*(X_t,Y_t,Z_t)-1\right\}|X_t,Z_t\right]\\
&+\mathbb{E}\left[r_0(Y_t,Z_t)\pi_{0}(X_t,Y_t,Z_t)\left\{\pi_K^*(X_t,Y_t,Z_t)-1\right\}|Y_t\right]\\	&-\mathbb{E}\left[r_0(Y_t,Z_t)\pi_K^*(X_t,Y_t,Z_t)\left\{\pi_K^*(X_t,Y_t,Z_t)-1\right\}|Y_t\right]\\
				&-\mathbb{E}\left[r_0(Y_t,Z_t)\pi_K^*(X_t,Y_t,Z_t)\left\{\pi_K^*(X_t,Y_t,Z_t)-1\right\}|Z_t\right]\\
				&+\mathbb{E}[r_0(Y,Z)\pi_K^*(X,Y,Z)\left\{\pi_K^*(X,Y,Z)-1\right\}]
				\\
				&-r_0(Y_t,Z_t)\pi_K^*(X_t,Y_t,Z_t)\left\{\pi_K^*(X_t,Y_t,Z_t)-1\right\},
			\end{align*}
			and with (\ref{cpihat-pi*}) and Theorem \ref{th_crate}
		(i), 
        \begin{equation*}
				\begin{aligned}
					\widehat I_{1K}=&\frac{1}{n}\sum_{t=1}^{n}\{ \widehat \pi_{K,K_0}(X_t, Y_t, Z_t)- \pi_{K}^*(X_t, Y_t, Z_t)\}^2r_0(Y_t,Z_t)\\
                    =&O_P\left( \frac{K_0}{n}+K_0^{-2\omega_{r}}+\frac{K}{n}\right),
				\end{aligned}
                \end{equation*}
   \begin{small}
				\begin{align*}
    \widehat I_{5K}= & \frac{2}{n}\sum_{t=1}^n\{\widehat\pi_{K,K_0}(X_t, Y_t, Z_t)-\pi_K^*(X_t, Y_t, Z_t)\}\{\pi_K^*(X_t, Y_t, Z_t)-1\}\left\{ \widehat{r}_{K_0}(Y_t,Z_t) -r_0(Y_t,Z_t)\right\}\\
      \leq&2\sqrt{ \frac{1}{n}\sum_{t=1}^n\{\widehat\pi_{K,K_0}(X_t, Y_t, Z_t)-\pi_K^*(X_t, Y_t, Z_t)\}^2}\\
      &\times \sqrt{ \frac{1}{n}\sum_{t=1}^n\{\pi_K^*(X_t, Y_t, Z_t)-1\}^2\left\{ \widehat{r}_{K_0}(Y_t,Z_t) -r_0(Y_t,Z_t)\right\}^2}\\
      =&O_P\left(\left\{\sqrt{\frac{K_0}{n}}+K_0^{-\omega_r}+\sqrt{\frac{K}{n}}\right\}\left\{K_0^{-\omega_{r}}+\sqrt{\frac{K_0}{n}}\right\}\right).
 \end{align*}
  \end{small}
			Then under Assumptions \ref{ass:cbounded}-\ref{ass:cpi0-pi*} and \ref{ass:cK}-\ref{ass:proj_error},  we have
			\begin{equation*}
				\begin{aligned}
					\sqrt{n}	\widehat I_{2K}
					=\frac{2}{\sqrt{n}}\sum_{t=1}^{n}\phi_{1,K}(X_t, Y_t, Z_t)+o_P(1),
				\end{aligned}
			\end{equation*}
			\begin{equation*}
				\begin{aligned}
					\sqrt{n}	\widehat I_{1K}=O_P\left(\sqrt{n}\left\{ \frac{K_0}{n}+K_0^{-2\omega_{r}}+\frac{K}{n}\right\}\right) =o_P(1),
				\end{aligned}
			\end{equation*}
   and\begin{equation*}
				\begin{aligned}
					\sqrt{n}	\widehat I_{5K}=O_P\left(\sqrt{n}\left\{\sqrt{\frac{K_0}{n}}+K_0^{-\omega_r}+\sqrt{\frac{K}{n}}\right\}\left\{K_0^{-\omega_{r}}+\sqrt{\frac{K_0}{n}}\right\}\right)=o_P(1).
				\end{aligned}
			\end{equation*}
   Note that under Assumptions \ref{ass:ciid}, \ref{ass:eigen} and \ref{ass:cK},
       \begin{align*}
           \widehat I_{4K}=&\frac{1}{n}\sum_{t=1}^{n}\{\pi_K^*(X_t, Y_t, Z_t)-1\}^2\left\{ \widehat{r}_{K_0}(Y_t,Z_t) -r_{K_0}^*(Y_t,Z_t)\right\}\\
           &+\frac{1}{n}\sum_{t=1}^{n}\{\pi_K^*(X_t, Y_t, Z_t)-1\}^2\left\{ r_{K_0}^*(Y_t,Z_t)-r_0(Y_t,Z_t)\right\}\\
           =&\frac{1}{n}\sum_{t=1}^{n}\{\pi_K^*(X_t, Y_t, Z_t)-1\}^2u_{K_0}^\top(Y_t,Z_t)\left\{ \widehat{\gamma}_{K_0}-\gamma_{K_0}^*\right\}+O_P\left(K_0^{-\omega_r}\right)\\
           =&\mathbb{E}\left[\{\pi_K^*(X_t, Y_t, Z_t)-1\}^2u_{K_0}^\top(Y_t,Z_t)\right]\left\{ \widehat{\gamma}_{K_0} -\gamma_{K_0}^*\right\}\\
           &+O_P\left(\frac{\zeta_{K_0}\sqrt{K_0}}{n}\right)+O_P\left(K_0^{-\omega_r}\right)\\
           =&\mathbb{E}\left[\{\pi_K^*(X_t, Y_t, Z_t)-1\}^2u_{K_0}^\top(Y_t,Z_t)\right]\left\{\frac{1}{n}\sum_{t=1}^{n}	\int_{\mathcal{Z}}u_{K_0}(Y_t,z)f_Z(z)dz\right.\\
           &\left.+\frac{1}{n}\sum_{t=1}^{n}	\int_{\mathcal{Y}}u_{K_0}(y,Z_t)f_Y(y)dy -2\mathbb{E}\left[u_{K_0}(Y,Z)r_0(Y,Z)\right]\right\}\\
           &+O_P\left( \frac{\sqrt{\zeta_{K_0}}K_0^{1/4}}{n^{(2-\epsilon)/2}}\right)+O_P\left(\frac{\zeta_{K_0}\sqrt{K_0}}{n}\right)+O_P\left(K_0^{-\omega_r}\right)\ (\text{by} (\ref{gamma}))\\
          =&\frac{1}{n}\sum_{t=1}^{n}\phi_{2,K}(X_t, Y_t, Z_t)+o_P\left(\frac{1}{\sqrt{n}}\right), 
       \end{align*}
   where
   \begin{equation*}
       \begin{aligned}
           \phi_{2,K}(X_t, Y_t, Z_t)=&\mathbb{E}\left[\left(\pi_K^*(X_t,Y_t,Z_t)-1\right)^2r_0(Y_t,Z_t)|Y_t\right]\\
           &+\mathbb{E}\left[\left(\pi_K^*(X_t,Y_t,Z_t)-1\right)^2r_0(Y_t,Z_t)|Z_t\right]\\
           &-2\mathbb{E}\left[\left(\pi_K^*(X_t,Y_t,Z_t)-1\right)^2r_0(Y_t,Z_t)\right].       \end{aligned}
   \end{equation*}
   Taking $\phi_K(X_t,Y_t,Z_t) :=2\phi_{1,K}(X_t, Y_t, Z_t)+\phi_{2,K}(X_t, Y_t, Z_t)$, we have
   \begin{equation}
\begin{aligned}\label{phiK'}
           \phi_K(X_t,Y_t,Z_t) 
=&2\mathbb{E}\left[r_0(Y_t,Z_t)\pi_{0}(X_t,Y_t,Z_t)\{\pi_K^*(X_t,Y_t,Z_t)-1\}|X_t,Z_t\right]\\
     &+2\mathbb{E}\left[r_0(Y_t,Z_t)\pi_{0}(X_t,Y_t,Z_t)\{\pi_K^*(X_t,Y_t,Z_t)-1\}|Y_t\right]\\			
    &-\mathbb{E}\left[r_0(Y_t,Z_t)\{\pi_K^*(X_t,Y_t,Z_t)\}^2|Y_t\right]\\
				&-\mathbb{E}\left[r_0(Y_t,Z_t)\{\pi_K^*(X_t,Y_t,Z_t)\}^2|Z_t\right]\\
		&+2\mathbb{E}[r_0(Y,Z)\pi_K^*(X,Y,Z)]\\
        &-2r_0(Y_t,Z_t)\pi_K^*(X_t,Y_t,Z_t)\{\pi_K^*(X_t,Y_t,Z_t)-1\}.
       \end{aligned}
   \end{equation}
			It worth noting that $\mathbb{E}\left[\phi_K(X,Y,Z)\right]=0$.
			Therefore, we can obtain that
			\begin{equation*}
				\frac{\sqrt{n}\left\{ 2\widehat{I}_{K,K_0}-\mathbb{E}\left[ \left\{\pi_K^*(X,Y,Z)-1\right\}^2r_0(Y,Z) \right] \right\} }{\sqrt{Var\left( \left\{\pi_K^*(X,Y,Z)-1\right\}^2r_0(Y,Z) +\phi_{K}\left( X,Y,Z\right) \right) }}\xrightarrow{d}\mathcal N(0,1).
			\end{equation*}%
		\end{itemize}

		\clearpage
		\section{Proof of Theorem \ref{thm:cglobal}}
		\label{app:cglobal}
	
		As $K \rightarrow \infty$, we have $\pi_K^*(X,Y,Z) \rightarrow \pi_0(X,Y,Z)$  in $L^2(dF_{X,Y,Z})$ by (\ref{pi*-pi0_L2}). Theorem \ref{thm:cglobal} holds in light of this result, Theorem \ref{thm:cglobal_fixK}, and \eqref{phiK'} converges to $\phi_0(X_t,Y_t,Z_t) $ as $K \rightarrow \infty$, where
        \begin{small}
		\begin{equation}
			\begin{aligned}
				&\phi_0(X_t,Y_t,Z_t)
=2\mathbb{E}\left[r_0(Y_t,Z_t)\pi_{0}^2(X_t,Y_t,Z_t)|X_t,Z_t\right]+\mathbb{E}\left[r_0(Y_t,Z_t)\pi_{0}^2(X_t,Y_t,Z_t)|Y_t\right]\\
&\qquad-\mathbb{E}[r_0(Y_t,Z_t)\pi_0^2(X_t,Y_t,Z_t)|Z_t]-2r_0(Y_t,Z_t)\pi_0(X_t,Y_t,Z_t)\{\pi_0(X_t,Y_t,Z_t)-1\}-2.
			\end{aligned}
		\end{equation}
	\end{small}

	\end{document}